\documentclass[11pt]{article}
\usepackage{paralist}
\usepackage{color}

\usepackage{amsfonts}
\usepackage{amssymb}

\usepackage{amsmath}
\usepackage{constants}
\usepackage{amsthm}
\makeatletter
\newtheorem*{rep@theorem}{\rep@title}
\newcommand{\newreptheorem}[2]{%
\newenvironment{rep#1}[1]{%
 \def\rep@title{#2 \ref{##1}}%
 \begin{rep@theorem}}%
 {\end{rep@theorem}}}
\makeatother

\allowdisplaybreaks

\usepackage{ifpdf}
\newcommand{\mydriver}{hypertex}
\ifpdf
 \renewcommand{\mydriver}{pdftex}
\fi
\usepackage[breaklinks,\mydriver]{hyperref}

\usepackage[margin=1in]{geometry}

\usepackage{tikz}
\usetikzlibrary{shapes,calc,positioning,arrows,}
\newcommand{\ie}{\text{i.\,e.}\ }

\newcommand{\ar}{\operatorname{ar}}
\newcommand{\Npos}{\mathbb N_{> 0}}
\newcommand{\FO}{\operatorname{FO}}
\newcommand{\FOsigma}{\operatorname{FO}[\sigma]}

\newcommand{\F}{\mathcal{F}}

\newcommand{\x}{\textbf{x}}
\newcommand{\y}{\textbf{y}}

\newcommand{\rot}{\operatorname{ROT}}
\newcommand{\dist}{\operatorname{dist}}
\newcommand{\indexSetRotation}{[D]^2}
\newcommand{\indexSetH}{([D]^2)^2}
\newcommand{\sampler}{\textsc{EstimateFrequencies}}
\newcommand{\setOfMaxCliques}[1][G]{\mathcal{K}^{#1}}
\newcommand{\maxcl}{\operatorname{maxcl}}
\renewcommand{\phi}{\varphi}
\renewcommand{\mod}{\mathrel{\mathrm{mod}}}

\newcommand*\circled[1]{\tikz[baseline=(char.base)]{
		\node[shape=circle,draw,inner sep=1pt] (char) {#1};}}
\DeclareMathOperator{\zigzag}{\circled{{\rm z}}}

\newcommand{\margin}[1]{$\bullet$\marginpar{\raggedright\footnotesize #1}}

\providecommand{\abs}[1]{\left\lvert#1\right\rvert}

\theoremstyle{plain}
\newtheorem{theorem}{Theorem}[section]
\newtheorem{fact}[theorem]{Fact}
\newtheorem{lemma}[theorem]{Lemma}

\newtheorem{claim}[theorem]{Claim}
\newtheorem{proposition}[theorem]{Proposition}

\newreptheorem{theorem}{Theorem}
\newreptheorem{lemma}{Lemma}
\newtheorem{definition}[theorem]{Definition}

\theoremstyle{definition}

\newtheorem{example}[theorem]{Example}

\newcommand{\junk}[1]{{}}

\providecommand{\abs}[1]{\lvert#1\rvert} 

\title{On Testability of First-Order Properties in  Bounded-Degree Graphs}

\author{Isolde Adler
	\thanks{School of Computing, University of Leeds, UK. Email: \url{I.M.Adler@leeds.ac.uk}. }
	\and
	Noleen K\"ohler
	\thanks{School of Computing, University of Leeds, UK. Email: 	\url{scnk@leeds.ac.uk}.}
	\and
	Pan Peng
	\thanks{Department of Computer Science, University of Sheffield, UK. Email: \url{p.peng@sheffield.ac.uk}.}
}

\date{}

\begin{document}
	
	
\begin{titlepage}
		
\maketitle
		
\thispagestyle{empty}
\begin{abstract}
We study property testing of properties that are definable in first-order logic (FO) in the
bounded-degree graph and relational structure models. We show that any FO property that is defined by a formula with quantifier prefix $\exists^*\forall^*$ is testable (i.e., testable with constant query complexity), while there exists an FO property that is expressible by a formula with quantifier prefix  $\forall^*\exists^*$ that is not testable. In the dense graph model, a similar picture is long known (Alon,
Fischer, Krivelevich, Szegedy, Combinatorica 2000), despite the very different nature of the two models. In particular, we obtain our lower bound by 
a first-order formula that defines a class of bounded-degree expanders,
based on zig-zag products of graphs. We expect this to be of independent
interest.



We then prove testability of some first-order
properties that speak about isomorphism types of neighbourhoods,
including testability of $1$-neigh\-bour\-hood-freeness, and 
$r$-neighbourhood-freeness under a mild assumption on the degrees.
	
\end{abstract}
		
\end{titlepage}
	
\section{Introduction}
Graph property testing is a framework for studying sampling-based algorithms that solve a relaxation of classical decision problems on graphs. Given a graph $G$ and a property $P$ (e.\,g.\ triangle-freeness), the goal of a property testing algorithm, called a \emph{property tester}, is to distinguish if a graph satisfies $P$ or is \emph{far} from satisfying $P$, where the definition of \emph{far} depends on the model. The general notion of property testing was first proposed by Rubinfeld and Sudan \cite{rubinfeld1996robust}, with the motivation for the study of program checking. Goldreich, Goldwasser and Ron \cite{goldreich1998property} then introduced the property testing for combinatorial objects and graphs. They formalized the \emph{dense graph model} for testing graph properties, in which the algorithm can query if any pair of vertices of the input graph $G$ with $n$ vertices are adjacent or not, and the goal is to distinguish, with probability at least $2/3$, the case of $G$ satisfying a property $P$ from the case that one has to modify (delete or insert) more than $\varepsilon n^2$ edges to make it satisfy $P$, for any specified proximity parameter $\varepsilon\in (0,1]$.
A property $P$ is called testable (in the dense graph model), if it can be tested with constant query complexity, i.e., the number of queries made by the tester is bounded by a function of $\varepsilon$ and is independent of the size of the input graph. Since \cite{goldreich1998property}, much effort has been made on the testability of graph properties in this model, culminating in the work by Alon et al.~\cite{alon2009combinatorial}, who showed that a property is testable if and only if it can be reduced to testing for a finite number of regular partitions. 

Since Goldreich and Ron's seminal work~\cite{GoldreichRon2002} introducing property testing on bounded-degree graphs, much 
attention has been paid to property testing in sparse graphs. 
Nevertheless, our understanding of testability of properties in such graphs is still limited. In the \emph{bounded-degree graph model}~\cite{GoldreichRon2002}, the algorithm has oracle access to the input graph $G$ with maximum degree $d$, which is assumed to be a constant, and is allowed to perform \emph{neighbour queries} to the oracle. That is, for any specified vertex $v$ and index $i\leq d$, the oracle returns the $i$-th neighbour of $v$ if it exists or a special symbol $\bot$ otherwise in constant time. 
A graph $G$ with $n$ vertices is called \emph{$\varepsilon$-far} from satisfying a property $P$, if one needs to modify more than $\varepsilon dn$ edges to make it satisfy $P$. 
The goal now becomes to distinguish, with probability at least $2/3$, if $G$ satisfies a property $P$ or is $\varepsilon$-far from satisfying $P$, for any specified proximity parameter $\varepsilon\in (0,1]$. 
Again, a property $P$ is testable in the bounded-degree model, if it can be tested with constant query complexity, where the constant can depend on $\varepsilon, d$ while being  independent of $n$. So far, it is known that some properties are testable, including subgraph-freeness, $k$-edge connectivity, cycle-freeness, being Eulerian, degree-regularity~\cite{GoldreichRon2002}, minor-freeness~\cite{benjamini2010every,hassidim2009local,kumar2019random}, hyperfinite properties \cite{NewmanSohler2013}, $k$-vertex connectivity~\cite{yoshida2012property,forster2019computing}, and subdivision-freeness~\cite{kawarabayashi2013testing}. 

In this paper, we study the testability of properties definable in first-order  logic (FO) in the bounded-degree graph model. Recall that formulas of first-order logic on graphs are built from predicates for the edge relation and equality, using
Boolean connectives $\vee,
\wedge,\neg$ and universal and existential quantifiers $\forall,\exists$, where the variables represent graph vertices. First-order logic can e.\,g.\ express subgraph-freeness (i.\,e., no isomorphic copy of some fixed graph $H$ appears as a subgraph) and subgraph containment (i.\,e., an isomorphic copy of some fixed $H$ appears as a subgraph).  
Note however, that there are constant-query testable properties, such as connectivity and cycle-freeness, that cannot be expressed in first-order logic. We study the question of which first-order properties are testable in the bounded-degree graph model. Our study extends to the bounded-degree \emph{relational structure} model \cite{AdlerH18}, while we focus on the classes of relational structures  
with binary relations, i.e., edge-coloured directed graphs. In this model for relational structures, one can perform neighbour queries for each edge
colour class, querying for both in- and out-neighbours via edges in that class.
This model is natural in the context of relational databases, where each
(edge-)relation is given by a list of the tuples it contains. 

We consider the testability of first-order 
properties in the bounded degree model according to 
quantifier alternation, inspired by a similar study for dense graphs by Alon et al.~\cite{alon2000efficient}. On relational structures of bounded degree over a fixed finite signature, 
we have the following simple observation: Any 
first-order property definable by a sentence \emph{without} quantifier alternations is testable.
This means the sentence either consists of a quantifier prefix of the form $\exists^*$ (any
finite number of existential quantifications), followed by a quantifier-free formula, or it consists of
a quantifier prefix of the form $\forall^*$ (any finite number of universal quantifications), 
followed by a quantifier-free formula.
%
Basically, every property of the form $\exists^*$ is testable because the structure required by 
the quantifier free part of the formula can be planted
with a small number of tuple modifications if the input structure is large enough (depending on the formula), and
we can use an exact algorithm to determine the answer in constant time otherwise.
Every property of the form $\forall^*$ is testable because a formula of the form $\forall \bar x \phi(\bar x)$,
where $\phi$ is quantifier free, is logically equivalent to a formula of the form 
$\neg \exists \bar x \psi(\bar x)$, where $\psi$ is quantifier free. Testing $\neg \exists \bar x \psi(\bar x)$
then amounts to testing for the absence of a finite number of induced substructures, which can be done
similar to testing subgraph freeness~\cite{GoldreichRon2002}. 
The testability of a property becomes less clear if it is defined by a sentence \emph{with} quantifier alternations. 
Formally, we let $\Pi_2$ (resp. $\Sigma_2$) denote the set of properties that can be expressed by a formula in the $\forall^*\exists^*$-prefix (resp. $\exists^*\forall^*$-prefix) class. 
We obtain the following.
\\[-0.3cm]

\textit{Every first-order property in $\Sigma_2$ is testable in the bounded-degree model (Theorem \ref{thm:sigma2}). On the other hand, there is a first-order property in $\Pi_2$, that is not testable in the bounded-degree model (Theorem~\ref{thm:pi2}).}\\[-0.3cm]

The theorems that we refer to in the above statement are for relational structures, while we also give a lower bound on graphs (Theorem~\ref{thm:simpleDelta2}),
so the statement also holds when restricted to FO on graphs. 
Interestingly, the above dividing line is the
same as for FO properties in dense graph model~\cite{alon2000efficient},
despite the very different nature of the two models.
Our proof uses a number of new proof techniques, combining graph theory, combinatorics and logic. 

We remark that our lower bound, i.e., the existence of a property in $\Pi_2$ that is not testable, is somewhat 
astonishing (on an intuitive level)  
due to the following two reasons. Firstly, it is proven by constructing a 
first-order definable class of structures that encode a class of expander graphs, which highlights that FO is surprisingly expressive on bounded degree graphs, despite its locality~\cite{Hanf1965,Gaifman82,FaginStockmeyerVardi1995}.   
Secondly, it is known that property testing algorithms in the bounded-degree model proceed by sampling vertices from the input graph and exploring their local neighbourhoods, and FO can only express `local' properties, while our lower bound shows that this is not sufficient for testability. We elaborate this in more details in the following. 
On one hand, Hanf's Theorem~\cite{Hanf1965} gives insight into first-order logic on graphs of bounded degree and 
implies a strong normal form, called \emph{Hanf Normal Form} (HNF) in~\cite{BolligKuske2012}, which we briefly sketch.
For a graph $G$ of maximum degree $d$ and a vertex $x$ in $G$, the
neighbourhood of fixed radius $r$ around $x$ in $G$ can be described by a first-order formula $\tau_r(x)$, up to isomorphism.
A \emph{Hanf sentence} is a first-order sentence of the form `there are at least $\ell$ vertices $x$ of 
neighbourhood (isomorphism) type $\tau_r(x)$'. A  first-order sentence is in HNF, if it is a Boolean combination of Hanf sentences.
By Hanf's Theorem, every first-order sentence is equivalent to a sentence in HNF on bounded-degree graphs~\cite{Hanf1965,FaginStockmeyerVardi1995,EF95}. Note that Hanf sentences only speak about local 
neighbourhoods. Hence this theorem gives evidence that first-order logic can only express local
properties. 
On the other hand, 
if a property is constant-query testable in the bounded-degree graph model, then it can be tested by approximating the distribution of local neighbourhoods (see \cite{CzumajPS16} and~\cite{goldreich2011proximity}). That is, a constant-query tester can essentially only test properties that are close to being defined by a distribution of local neighbourhoods. For these reasons\footnote{Furthermore, previously, typical FO properties are all known to be testable, including degree-regularity for a fixed given degree, containing a $k$-clique and a dominating set of size $k$ for fixed $k$ (which are trivially testable), and the aforementioned subgraph-freeness and subgraph containment (see e.g.~\cite{goldreich2017introduction}).}, a priori, it could be true that every property that can be expressed in first-order logic is testable in the bounded degree model. Indeed, the validity of this statement was raised as an open question in~\cite{AdlerH18}. However, our lower bound gives a negative answer to this question. 

Motivated by our above results, we further study testability of graph properties described through Hanf sentences or negated Hanf sentences, 
which are first-order properties that speak about isomorphism types of neighbourhoods. Given a bounded degree graph, an \emph{$r$-ball} around a vertex $x$ is the neighbourhood of radius $r$ around $x$ in the graph. We call the isomorphism types of $r$-balls \emph{$r$-types}. 
We consider two basic such properties, called \emph{$\tau$-neighbourhood regularity} and \emph{$\tau$-neighbourhood-freeness}, that correspond to ``all vertices have $r$-type $\tau$'' and ``no vertex has $r$-type $\tau$'', respectively. (Neighbourhood-regularity can be seen as a generalisation of degree-regularity, which is known to be testable \cite{goldreich2017introduction}.) 
As we show in Lemma \ref{lem:existencesigma2}, there exist $1$-types $\tau,\tau'$ such that neither $\tau$-neighbourhood-freeness nor $\tau'$-neighbourhood regularity can be defined by a formula in $\Sigma_2$. 
Thus, our previous tester for $\Sigma_2$ cannot be applied to these properties. We give constant-query testers for them under certain conditions on 
	$\tau$ (Theorem~\ref{thm:dNeighbourhoodFreeness}, Theorem~\ref{thm:1NeighbourhoodFreeness}, 
	Theorem~\ref{thm:neighbourhoodRegularity}).
	Both $\tau$-neighbourhood-freeness and $\tau$-neighbourhood regularity can be defined by formulas in $\Pi_2$ for any neighbourhood type $\tau$. Thus, our results imply that there are properties defined by formulas in $\Pi_2\setminus \Sigma_2$ which are testable. 

\paragraph{Our techniques} To show that every property $P$ defined by a formula $\varphi$ in $\Sigma_2$ (i.e.\ of the form $\exists^*\forall^*$) is testable, we show that $P$ is equivalent to 
{the union of properties $P_i$, each of which is} `indistinguishable' from a property $Q_i$ that is defined by a formula of form $\forall^*$. Here the indistinguishability means we can transform any structure satisfying $P_i$, into a structure satisfying $Q_i$ by modifying a small fraction of the tuples of the structure and vice versa. This allows us to reduce the problem of testing $P$ to testing properties defined by $\forall^*$ formulas. 
Then the testability of $P$ follows, as any property of the form $\forall^*$ is testable and testable properties are closed under union \cite{goldreich2017introduction}. The main challenge here is to deal with the interactions between existentially quantified variables and universally quantified variables.
Intuitively, the degree bound limits the structure that can be imposed by the universally quantified variables. Using this, we are able to deal with the existential variables together with these interactions by `planting' a required
constant size substructure in such a way, that we are only a constant number of modifications `away' from a formula of the form $\forall^*$. 

Complementing this, we use Hanf's theorem to observe that every FO property on degree-regular structures is in $\Pi_2$ (see Lemma \ref{lem:d-regHNF}). Thus to prove that there exists a property defined by a formula in $\Pi_2$ which is not testable, it suffices to show the existence of an FO property that is not testable and degree-regular. 
%
For the latter, 
we note that it suffices to construct a formula $\phi$, that defines a class of relational structures with binary relations only (edge-coloured directed graphs) whose underlying undirected graphs are expander graphs. To see this, we use an earlier result that if a property is constant-query testable, then the distance between the local (constant-size) neighbourhood distributions of a relational structure $\mathcal{A}$ satisfying the property $\phi$ and a relational structure $\mathcal{B}$ that is $\varepsilon$-far from having the property must be relatively large (see \cite{AdlerH18} which in turn is built upon the so-called ``canonical testers'' for bounded-degree graphs in \cite{CzumajPS16,goldreich2011proximity}). We then exploit a result of Alon (see Proposition 19.10 in~\cite{LovaszBook2012}), that the neighbourhood distribution of an arbitrarily large relational structure $\mathcal{A}$
can be approximated by the neighbourhood distribution of a structure $\mathcal{H}$ of small constant size. Thus, for any $\mathcal{A}$ in $\phi$, by taking the union of ``many'' disjoint copies of the ``small'' structure $\mathcal{H}$, we obtain another structure $\mathcal{B}$ such that the local neighbourhood distributions of $\mathcal{A}$ and $\mathcal{B}$ have small distance. If the underlying undirected graphs of the structures in $\phi$ are expander graphs, it immediately follows that $\mathcal{B}$ is far from the property defined by the formula $\phi$, from which we can conclude that the property $\phi$ is not testable. We remark that for simple undirected graphs, it was known before that any property that only consists of expander graphs is not testable \cite{fichtenberger2019every}. 


Now we construct a formula $\phi$, that defines a class of relational structures with binary relations only whose underlying
undirected graphs are expander graphs, arising from the zig-zag product by Reingold, Vadhan and Wigderson~\cite{ZigZagProductIntroduction}. For expressibility in FO, we hybridise the zig-zag construction of 
expanders with a tree structure. Roughly speaking, we start with a small graph
$H$, which is a good expander, and the formula $\phi$ expresses that each model
\footnote{When the context is clear, we use ``model'' to indicate that a structure satisfies some formula. This should not be confused with the names for our computational models, e.g., the bounded-degree model.}
looks like a rooted $k$-ary tree (for a suitable fixed $k$), where level $0$ consists of the root only,  
level $1$ contains $G_1:=H^2$, and level $i$ contains the zig-zag product
of $G_{i-1}^2$ with $H$. The class of trees is not definable in FO. 
However, we achieve that every finite model of our formula is connected and looks like a 
$k$-ary tree with the desired graphs on the levels.
This structure is
obtained by a recursive `copying-inflating' mechanism, 
to mimic the expander construction locally between consecutive levels.
For this we use a constant number of edge-colours, one set of colours for the edges of the tree, and  
another for the edges of the `level' graphs $G_i$.
On the way, many technicalities need to be tackled, such as encoding the zig-zag construction
into the local copying mechanism (and achieving the right degrees), and finally proving connectivity.
We then show that the underlying 
undirected graphs of the models of $\phi$ are expander graphs. 
Finally, we extend this construction to simple undirected graphs, 
by using carefully designed gadgets to encode the different edge-colours and maintain degree regularity.


To give our testers for $\tau$-neighbourhood regularity and $\tau$-neighbourhood-freeness, we show that if a graph $G$ is $\varepsilon$-far from having the property, it contains a linear fraction of constant-size neighbourhoods certifying that $G$ does not satisfy the property. Such a statement may be intuitively true, but it is tricky to prove. Assume we want to test for $\tau$-freeness, for some fixed $r$-neighbourhood type $\tau$,
	and assume a graph $G$ has one vertex $x$ with forbidden neighbourhood of type $\tau$.
	Changing the $r$-neighbourhood of $x$  
	by edge modifications, in order to remove $\tau$, might introduce new forbidden 
	neighbourhoods around vertices close to $x$, triggering
	a `chain reaction' of necessary modifications.
	This means that a graph might be $\epsilon$-far from being $\tau$-free, but we do not see it by sampling constantly many neighbourhoods in the graph. Such a subtle difficulty has already been observed for testing degree-regularity (see Claim 8.5.1 in \cite{goldreich2017introduction}). We show that under appropriate assumptions, such a `chain reaction' can be bypassed 
	by carefully fixing the neighbourhood of $x$ without changing the neighbourhood type of the vertices surrounding $x$. 
	Though fairly simple, it provides non-trivial analysis, handling the subtle difficulty of relating local distance to global distance without triggering a `chain reaction'.

\paragraph{Other related work}
Besides the aforementioned works on testing properties with constant query complexity in the bounded-degree graph model, Goldreich and Ron~\cite{goldreich2011proximity} have obtained a characterisation for a class of properties that are testable by a constant-query proximity-oblivious tester in bounded-degree graphs (and dense graphs). Such a class is a rather restricted subset of the class of all constant-query testable properties. 
Fichtenberger et al.~\cite{fichtenberger2019every} showed that every testable property is either finite or contains an infinite hyperfinite subproperty (see Definition~\ref{def:hyperfinite}). Ito et al.~\cite{ito2019characterization} gave characterisations of one-sided error (constant-query) testable monotone graph properties, and one-sided error testable hereditary graph properties in the bounded-degree (directed and undirected) graph model.

In the bounded-degree graph model, 
there are also properties (e.g.\ bipartiteness, expansion, $k$-clusterability) that require $\Omega(\sqrt{n})$ 
queries, and properties (e.g.\ $3$-colorability) that require $\Omega(n)$ queries. We refer the reader to Goldreich's recent book~\cite{goldreich2017introduction}.

Property testing on relational structures was recently motivated by the application in databases. Besides the aforementioned work \cite{AdlerH18}, Chen and Yoshida  \cite{chen2019testability} studied the testability of relational database queries for each relational structure in the framework of property testing.


\paragraph{Further discussions and open problems} The question whether first-order definable properties are testable with
a \emph{sublinear} number of queries  (e.g. $\sqrt{n}$) in the bounded-degree model is left open. 


We believe it is natural to study the problem of testing properties of neighbourhood types. Firstly, our previous results can be seen as an indication that
quantifier prefix classes are perhaps less suitable when searching for a dividing line
between testable and non-testable first-order properties in the bounded-degree model. 
Since subgraph-freeness and subgraph containment are testable, Hanf's normal form suggests 
studying testability of Hanf sentences and their negations, i.\,e.\ neighborhood properties, as a next step. Secondly, studying such properties helps us gain more insights on which properties that are defined by
distributions of neighbourhood types are testable, which is crucial to solve one of the most important open questions in this area, namely to characterise the combinatorial structure of testable properties in the bounded-degree model.


Furthermore, we remark that our testers for neighbourhood properties have \emph{one-sided error}, i.\,e.\ the testers always accept the graphs that satisfy the property. We note that in contrast to subgraph-freeness and induced subgraph-freeness, the properties $\tau$-neighbourhood regularity and $\tau$-neighbourhood-freeness are neither \emph{monotone} nor \emph{hereditary}, which are properties that are closed under edge deletion and closed under vertex deletion, respectively.  
As we mentioned before, Ito et al. \cite{ito2019characterization} recently characterised one-sided error (constant-query) testable monotone and hereditary graph properties in the bounded-degree (directed and undirected) graph model. In order to give a full characterisation of one-sided error testable properties in the bounded-degree graph model, it is important to take a step beyond monotone and hereditary graph properties.


\paragraph{Structure of the paper}
Section~\ref{sec:preliminaries} contains the preliminaries, including logic,
property testing and the zig-zag construction of expander graphs.
In Section~\ref{sec: definitionFormula} we construct the FO formula $\phi$ and prove properties 
of its models.
In Section~\ref{sec:FOnontestability}, we prove that there is a $\Pi_2$-property that is not testable, by proving that the property $P_\phi$ defined by $\phi$ on bounded-degree structures is not constant-query testable. We also 
provide a $\Pi_2$-property of simple, undirected graphs that is non-testable. In Section \ref{sec:testableSigma2}, we show that all $\Sigma_2$ properties are testable. 
In Section~\ref{sec:freeness} we give positive results for some first-order properties that
speak about isomorphism types of neighbourhoods. We conclude in Section \ref{sec:conclusion}. 

\section{Preliminaries}\label{sec:preliminaries}

We give some basics of graphs, relational structures and first-order logic in Appendix~\ref{sec:basic_appendix}. We use standard definitions and notation unless otherwise specified.

\subsection{The bounded-degree relational structure model}\label{sec:boundeddeg_model}
Let $\sigma=\{R_1,\dots,R_\ell\}$ be a relational signature and let $\mathcal{A}$ be a $\sigma$-structure with universe $A$.
The \emph{degree} of an element $a\in A$ denoted by $\deg_\mathcal{A}(a)$ is defined to be the number of tuples in $\mathcal{A}$ containing $a$.
We define the \textit{degree} of $\mathcal{A}$ denoted by $\deg(\mathcal{A})$ to be the maximum degree of its elements. For any $d\in \mathbb{N}$ we let $C_d$ be the class of all $\sigma$-structures of bounded degree $d$. Let us remark that $\deg(\mathcal{A})$ and the degree of the Gaifman graph of  $\mathcal{A}$ only differ by at most a constant factor (cf.\ e.\,g.~\cite{DBLP:journals/tocl/DurandG07}), so the definitions are equivalent in the sense that the same classes have bounded degree.
A \textit{property} on any class of structures $C$ is a subset $P\subseteq C$ of structures that is closed under isomorphism.
We say that a structure $\mathcal{A}\in C$ has property $P$ if $\mathcal{A}\in P$.
On $C_d$, every FO-sentence $\phi$ defines a property $P_\phi\subseteq C_d$, where
$P_\phi=\{\mathcal A\in C_d\mid \mathcal{A}\models \varphi\}$.

 We describe the model for bounded-degree relational structures as defined in~\cite{AdlerH18}. This extends the bounded-degree model for undirected graphs introduced in ~\cite{GoldreichRon2002} and conforms with the bidirectional model of ~\cite{CzumajPS16}.

 An algorithm that processes a $\sigma$-structure $\mathcal{A}\in C_d$
does not obtain an encoding of $\mathcal{A}$ as a bit string in the usual way. Instead, we assume that the algorithm receives the number $n$ of elements of $\mathcal{A}$, and that the elements of $\mathcal{A}$ are numbered $1,2,\ldots,n$. In addition, the algorithm has direct access to $\mathcal{A}$ using an \emph{oracle} which answers \emph{neighbour queries}
in $\mathcal{A}$ in constant time. That is, the oracle accepts queries of the form $(i,j,k)$, for
$i\in \{1,\dots,n\}$, $j\in\{1,\dots,\ell\}$ and $k\in \{1,\dots,d\}$, to which it responds
with the $k$-th tuple in $R_j^\mathcal{A}$ containing  
$i$, or with $\bot$ if $i$ is contained in less than $k$ tuples. 

The \emph{running time} of the
algorithm is defined as usual, \ie with respect to the size of the structure $n$. We assume
a uniform cost model, i.\,e.\ we assume that all basic arithmetic operations 
including random sampling can be performed in constant time, regardless of 
the size of the numbers involved.

\textbf{Distance.}
For two $\sigma$-structures $\mathcal{A}$ and $\mathcal{B}$, both of size $n$,
$\dist(\mathcal{A},\mathcal{B})$ denotes the minimum number of tuples that
have to be modified (\ie inserted or removed) in $\mathcal{A}$ and $\mathcal{B}$ to make  $\mathcal{A}$ and $\mathcal{B}$
isomorphic. 
For $\epsilon \in (0,1]$,
we say $\mathcal{A}$ and $\mathcal{B}$, both of size $n$ and with degree bound $d$, are 
$\epsilon$-\emph{close} if $\dist(\mathcal{A},\mathcal{B}) \leq \epsilon d n$.
If $\mathcal{A},\mathcal{B}$ are not $\epsilon$-\emph{close}, then they are \emph{$\epsilon$-far}.
Note that in particular, $\mathcal{A}$ and $\mathcal{B}$ are $\epsilon$-far if their size differs. 
A $\sigma$-structure $\mathcal{A}$  is \emph{$\epsilon$-close} to a property $P$ if $\mathcal{A}$
is $\epsilon$-close to some $\mathcal{B} \in P$. Otherwise, $\mathcal{A}$ is
\emph{$\epsilon$-far} from~$P$.
\begin{definition}[$\epsilon$-tester]
	Let $P \subseteq C_{d}$ be a property and $\epsilon\in(0,1]$. 
	An $\epsilon$-\emph{tester for $P$} is a probabilistic
	algorithm with oracle access to an input $\mathcal{A}\in C_{d}$  and auxiliary input
	$n:=\abs{A}$. The algorithm does the following.
	\begin{enumerate}
		\item If $\mathcal{A} \in P$, then the $\epsilon$-tester accepts with probability at
		least ${2}/{3}$.
		\item If $\mathcal{A}$ is $\epsilon$-far from $P$, then the $\epsilon$-tester rejects
		with probability at least ${2}/{3}$.
	\end{enumerate}
\end{definition}
The \emph{query complexity} of an $\epsilon$-tester is the maximum number of oracle queries made.
A property $P$ is \emph{testable}, if for each $\epsilon\in (0,1]$ and each $n$, there is an $\epsilon$-tester 
for $P\cap \{\mathcal{A}\in C_d\mid \abs{A}=n\}$ on inputs from $\{\mathcal{A}\in C_d\mid \abs{A}=n\}$ with constant query complexity, \ie the \emph{query complexity} is  independent of $n$.
A property $P$ is \emph{uniformly testable}, if for each $\epsilon \in (0,1]$
there is an $\epsilon$-tester for $P$, that
has \emph{constant query complexity}. 
Note that 
{`uniformly'} emphasises that this tester  
must work for all~$n$.

\subsection{Quantifier alternations of first-order formulas}
%
	
Let $\sigma$ be any relational signature and $C_d$ the set of $\sigma$-structures of bounded degree $d$.
We use the following recursive definition, classifying first-order formulas according to the number of quantifier alterations in their quantifier prefix. Let $\Sigma_0=\Pi_0$ be the class of all quantifier free first-order formulas over $\sigma$. Then for every $i\in \Npos$ we let $\Sigma_i$ be the set of all FO formulas $\varphi(y_1,\dots,y_\ell)$ for which there is $k\in \mathbb{N}$ and a formula $\psi(x_1,\dots,x_k,y_1,\dots,y_\ell)\in \Pi_{i-1}$ such that $\varphi\equiv \exists x_1\dots \exists x_k\psi(x_1,\dots,x_k,y_1,\dots,y_\ell)$. Analogously, $\Pi_i$ consists of all FO formulas $\varphi(y_1,\dots,y_\ell)$ for which there is $k\in \mathbb{N}$ and a formula $\psi(x_1,\dots,x_k,y_1,\dots,y_\ell)\in \Sigma_{i-1}$ such that\linebreak $\varphi\equiv \exists x_1\dots \exists x_k\psi(x_1,\dots,x_k,y_1,\dots,y_\ell)$. 
\begin{example}[Substructure freeness]
	Let $\mathcal B$ be a $\sigma$-structure, and let $d\in \mathbb N$. The property \[P:=\{\mathcal A\in C_d\mid \mathcal
	A \text{ does not contain }\mathcal B\text{ as substructure}\}\]
	is in $\Pi_1$ and is uniformly testable on $C_d$ with constant
	running time which can be proven similar to substructure freeness in simple graphs \cite{GoldreichRon2002}.
\end{example}
As we discussed in the introduction, every FO property in $\Sigma_1$ or $\Pi_1$, i.e., without quantifier alternation, is testable.


\subsection{Expansion and the zig-zag product}\label{sec:zigZagProduct}


In this section we  recall a construction of a class of expanders introduced in ~\cite{ZigZagProductIntroduction}. This construction uses undirected graphs with parallel edges and self-loops. We therefore encode a graph $G$ as a triple $(G,E,f)$ where $V$ is a finite sets of vertices, $E$ is a finite  set of edges, and $f$ is the incidence  map from $E$ to the set $\{x\subseteq V\mid 1\leq |x|\leq 2\}$.

Let $G=(V,E,f)$ be an undirected $D$-regular graph on $N$ vertices.  We follow the convention that each self-loop counts $1$ towards the degree. Let $I$ be a set of size $D$. 
Then a \emph{rotation map} of $G$ is a function $\rot_G:V\times I\rightarrow V\times I$ such that for every two not necessarily different vertices $u,v\in V$
\begin{displaymath}
|\{(i,j)\in I\times I \mid \rot_G(u,i)=(v,j)\}|=|\{e\in E\mid f(e)=\{u,v\} \}|
\end{displaymath}
and $\rot_{G}$ is self inverse, i.e. $\rot_G(\rot_G(v,i))=(v,i)$ for all $v\in V, i\in I$. A rotation map is a representation of a graph that additionally for every vertex $v$ fixes an order on all edges incident to $v$. We let the normalised adjacency matrix $M$ of $G$ be 
define by 
\begin{displaymath}
M_{u,v}:=\frac{1}{D}\cdot|\{e \mid f(e) = \{u,v\}\}|.
\end{displaymath}
Let $1=\lambda_1\geq \lambda_2\geq \dots\geq \lambda_N\geq -1$ denote the eigenvalues of $M$. Since $M$ is real, symmetric, contains no negative entries and all columns sum up to $1$, all its eigenvalues are  in the real  interval $[-1,1]$. We let $\lambda(G):=\max\{|\lambda_2|,|\lambda_N|\}$. Note that these notions do not depend on the rotation map. We say that a graph is an $(N,D,\lambda)$-graph, if $G$ has $N$ vertices, is $D$-regular and $\lambda(G)\leq \lambda$.
We will use the following lemma.
\begin{lemma}[\cite{Hoory06expandergraphs}]\label{lem:connectedBipartiteEigenvalues}
The graph $G$ is connected if and only if $\lambda_2<1$. Furthermore, if $G$ is connected, then $G$ is bipartite if and only if $\lambda_N=-1$.
\end{lemma}
For any subsets $S,T\subseteq V$ let $\langle S,T\rangle_G:=\{e\in E\mid f(e)\cap S\not=\emptyset,f(e)\cap T\not=\emptyset\}$ be the set of edges \emph{crossing} $S$ and $T$. 
\begin{definition}\label{def:expansionRatio}
	For any set $S\subseteq V$, we let $h(S):=\frac{|\langle S,\overline{S}\rangle_G|}{|S|}$ be 
	the \emph{expansion} of $S$. We let $h(G)$ be the \emph{expansion ratio} of $G$ defined by 
	$h(G):=\min_{\{S\subset V \mid |S|\leq N/2\}}h(S)$.
\end{definition}

For any constant $\epsilon>0$ we call a sequence $\{G_m\}_{m\in\Npos}$ of graphs of increasing number of vertices a \emph{family of $\epsilon$-expanders}, if  $h(G_m)\geq \epsilon$ for all $m\in \Npos$.
There exists the following connection between $h(G)$ and $\lambda(G)$.
\begin{theorem}[\cite{Dod1984difference,AM1985lambda1}]\label{thm:boundExpansionRatioInTermsOfLambda}
	Let $G$ be a $D$-regular graph. Then
	$h(G)\geq {D(1-\lambda(G))}/{2}.$
\end{theorem}

This implies that for a sequence of graphs $\{G_m\}_{m\in\Npos}$ of increasing number of vertices, if there is a constant $\epsilon<1$ such that $\lambda(G_m)\leq \epsilon$ for all $m\in \Npos$, then the sequence $\{G_m\}_{m\in \Npos}$ is a family of ${D(1-\varepsilon)}/{2}$-expanders. 
\begin{definition}
	Let $G$ be a $D$-regular graph on $N$ vertices with rotation map $\rot_{G}$ and $I$ a set of size $D$. Then the \emph{square of $G$}, denoted by $G^2$, is a $D^2$-regular graph on $N$ vertices with rotation map
	$\rot_{G^2}(u,(k_1,k_2)):=(w,(\ell_2,\ell_1)),\text{ where}$ 
	\begin{align*}
	\rot_{G}(u,k_1)=&(v,\ell_1) \text{ and }
	\rot_{G}(v,k_2)=(w,\ell_2),
	\end{align*}
	and $u,v,w\in V$, $k_1,k_2,\ell_1,\ell_2\in I$.
\end{definition}
Note that the edges of $G^2$ correspond to walks of length $2$ in $G$ and the adjacency matrix of $G^2$ is the square of the adjacency matrix of $G$.
Note here that if $G$ is bipartite then $G^2$ is not connected, which can be easily explained by using Lemma~\ref{lem:connectedBipartiteEigenvalues}.
\begin{lemma}[\cite{ZigZagProductIntroduction}]\label{lem:expansionOfSquaring}
	If $G$ is a $(N,D,\lambda)$-graph then $G^2$ is a $(N,D^2,\lambda^2)$-graph.	
\end{lemma}
\begin{definition}
	Let $G_1=(V_1,E_1,f_1)$ be a $D_1$-regular graph on $N_1$ vertices, $I_1$ a set of size $D_1$ and $\rot_{G_1}:V_1\times I_1\rightarrow V_1\times I_1$ a rotation map of $G_1$. Let $G_2=(I_1,E_2,f_2)$ be a $D_2$-regular graph, let $I_2$ be a set of size $D_2$ and $\rot_{G_2}:I_1\times I_2\rightarrow I_1\times I_2$ be a rotation map of $G_2$. Then the \emph{zig-zag product of $G_1$ and $G_2$}, denoted by $G_1\zigzag G_2$, is the $D_2^2$-regular graph on $V_1\times I_1$ with rotation map given by
	$\rot_{G_1\zigzag G_2}((v,k),(i,j)):=((w,\ell),(j',i')),\text{ where}$
	\begin{align*}
	&\rot_{G_2}(k,i)=(k',i'), \quad
	\rot_{G_1}(v,k')=(w,\ell'),\text{ and}
	\rot_{G_2}(\ell',j)=(\ell,j'),
	\end{align*}
	and $v,w\in V_1$, $k,k',\ell,\ell'\in I_1$, $i,i',j,j'\in I_2$.
\end{definition}

The zig-zag product $G_1\zigzag G_2$ can be seen as the result of the following construction. First pick some numbering of the vertices of $G_2$. Then replace every vertex in $G_1$ by a copy of $G_2$ where we colour edges from $G_1$, say, red, and edges from $G_2$ blue. We do this in such a way that the $i$-th edge in $G_1$ of a vertex $v$ will be incident to vertex $i$ of the to-$v$-corresponding-copy of $G_2$. Then for every red edge $(v,w)$ and for every tuple $(i,j)\in I_2\times I_2$ we add an edge to the zig-zag product $G_1\zigzag G_2$ connecting $v'$ and $w'$ where $v'$ is the vertex reached from $v$ by taking its $i$-th blue edge and $w'$ can be reached from $w$ by taking its $j$-th blue edge. 
Figure~\ref{fig:zigZagProduct} shows an example, where in the graph on the right hand side we show the $4$ edges that are added to the zig-zag product for the highlighted edge of the graph on the left hand side.
\definecolor{C1}{RGB}{1,1,1}
\definecolor{C2}{RGB}{0,0,170}
\definecolor{C3}{RGB}{251,86,4}
\definecolor{C4}{RGB}{50,180,110}
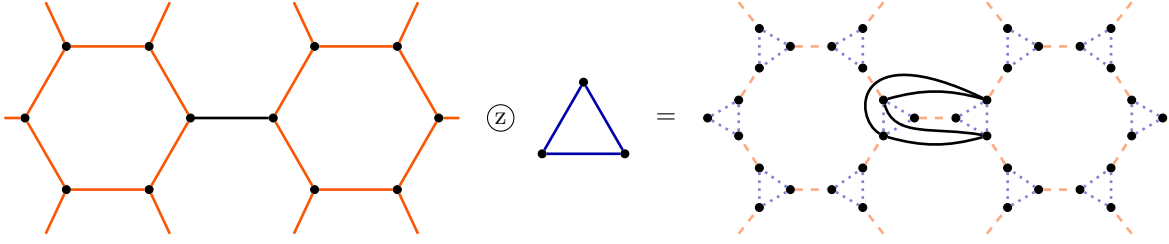
\begin{figure*}
	\centering
	\begin{tikzpicture}[level distance=7mm,scale = 0.55]
	
	\tikzstyle{ns1}=[line width=1]
	\node[draw,circle,fill=black,inner sep=0pt, minimum width=3pt] (5) at (-4,-1.732) {};
	\node[draw,circle,fill=black,inner sep=0pt, minimum width=3pt] (6) at (-2,-1.732) {};
	\node[draw,circle,fill=black,inner sep=0pt, minimum width=3pt] (7) at (2,-1.732) {};
	\node[draw,circle,fill=black,inner sep=0pt, minimum width=3pt] (8) at (4,-1.732) {};
	\node[draw,circle,fill=black,inner sep=0pt, minimum width=3pt] (9) at (-5,0) {};
	\node[draw,circle,fill=black,inner sep=0pt, minimum width=3pt] (10) at (-1,0) {};
	\node[draw,circle,fill=black,inner sep=0pt, minimum width=3pt] (11) at (1,0) {};
	\node[draw,circle,fill=black,inner sep=0pt, minimum width=3pt] (12) at (5,0) {};
	\node[draw,circle,fill=black,inner sep=0pt, minimum width=3pt] (13) at (-4,1.732) {};
	\node[draw,circle,fill=black,inner sep=0pt, minimum width=3pt] (14) at (-2,1.732) {};
	\node[draw,circle,fill=black,inner sep=0pt, minimum width=3pt] (15) at (2,1.732) {};
	\node[draw,circle,fill=black,inner sep=0pt, minimum width=3pt] (16) at (4,1.732) {};

	\draw[ns1,C3]
	(-4.5,-2.8)--(5)--(6)--(-1.5,-2.8)(6)--(10)(11)--(7)--(1.5,-2.8)(7)--(8)--(4.5,-2.8)(-5.5,0)--(9)--(5)(9)--(13)--(-4.5,2.8)(13)--(14)--(-1.5,2.8)(14)--(10)(11)--(15)--(1.5,2.8)(15)--(16)--(12)--(8)(12)--(5.5,0)(16)--(4.5,2.8);
	\draw[ns1,C1] (10)--(11);	
	
	\node[minimum height=10pt,inner sep=0,font=\small] at (6.5,0) {$\zigzag$};

	\node[draw,circle,fill=black,inner sep=0pt, minimum width=3pt] (41) at (7.5,-0.866) {};
	\node[draw,circle,fill=black,inner sep=0pt, minimum width=3pt] (42) at (9.5,-0.866) {};
	\node[draw,circle,fill=black,inner sep=0pt, minimum width=3pt] (43) at (8.5,0.866) {};
	\draw[ns1,C2] (41)--(42)--(43)--(41);
	
	\node[minimum height=10pt,inner sep=0,font=\small] at (10.5,0) {$=$};	
	
	\node[draw,circle,fill=black,inner sep=0pt, minimum width=3pt] (33) at (13.5,-1.732) {};
	\node[draw,circle,fill=black,inner sep=0pt, minimum width=3pt] (34) at (12.75,-1.209) {};
	\node[draw,circle,fill=black,inner sep=0pt, minimum width=3pt] (35) at (12.75,-2.165) {};
	\node[draw,circle,fill=black,inner sep=0pt, minimum width=3pt] (36) at (14.5,-1.732) {};
	\node[draw,circle,fill=black,inner sep=0pt, minimum width=3pt] (37) at (15.25,-1.209) {};
	\node[draw,circle,fill=black,inner sep=0pt, minimum width=3pt] (38) at (15.25,-2.165) {};
	\node[draw,circle,fill=black,inner sep=0pt, minimum width=3pt] (39) at (19.5,-1.732) {};
	\node[draw,circle,fill=black,inner sep=0pt, minimum width=3pt] (40) at (18.75,-1.209) {};
	\node[draw,circle,fill=black,inner sep=0pt, minimum width=3pt] (41) at (18.75,-2.165) {};
	\node[draw,circle,fill=black,inner sep=0pt, minimum width=3pt] (42) at (20.5,-1.732) {};
	\node[draw,circle,fill=black,inner sep=0pt, minimum width=3pt] (43) at (21.25,-1.209) {};
	\node[draw,circle,fill=black,inner sep=0pt, minimum width=3pt] (44) at (21.25,-2.165) {};
	\node[draw,circle,fill=black,inner sep=0pt, minimum width=3pt] (45) at (11.5,0) {};
	\node[draw,circle,fill=black,inner sep=0pt, minimum width=3pt] (46) at (12.25,0.433) {};
	\node[draw,circle,fill=black,inner sep=0pt, minimum width=3pt] (47) at (12.25,-0.433) {};
	\node[draw,circle,fill=black,inner sep=0pt, minimum width=3pt] (48) at (16.5,0) {};
	\node[draw,circle,fill=black,inner sep=0pt, minimum width=3pt] (49) at (15.75,0.433) {};
	\node[draw,circle,fill=black,inner sep=0pt, minimum width=3pt] (50) at (15.75,-0.433) {};
	\node[draw,circle,fill=black,inner sep=0pt, minimum width=3pt] (51) at (17.5,0) {};
	\node[draw,circle,fill=black,inner sep=0pt, minimum width=3pt] (52) at (18.25,0.433) {};
	\node[draw,circle,fill=black,inner sep=0pt, minimum width=3pt] (53) at (18.25,-0.433) {};
	\node[draw,circle,fill=black,inner sep=0pt, minimum width=3pt] (54) at (22.5,0) {};
	\node[draw,circle,fill=black,inner sep=0pt, minimum width=3pt] (55) at (21.75,0.433) {};
	\node[draw,circle,fill=black,inner sep=0pt, minimum width=3pt] (56) at (21.75,-0.433) {};
	\node[draw,circle,fill=black,inner sep=0pt, minimum width=3pt] (57) at (13.5,1.732) {};
	\node[draw,circle,fill=black,inner sep=0pt, minimum width=3pt] (58) at (12.75,1.209) {};
	\node[draw,circle,fill=black,inner sep=0pt, minimum width=3pt] (59) at (12.75,2.165) {};
	\node[draw,circle,fill=black,inner sep=0pt, minimum width=3pt] (60) at (14.5,1.732) {};
	\node[draw,circle,fill=black,inner sep=0pt, minimum width=3pt] (61) at (15.25,1.209) {};
	\node[draw,circle,fill=black,inner sep=0pt, minimum width=3pt] (62) at (15.25,2.165) {};
	\node[draw,circle,fill=black,inner sep=0pt, minimum width=3pt] (63) at (19.5,1.732) {};
	\node[draw,circle,fill=black,inner sep=0pt, minimum width=3pt] (64) at (18.75,1.209) {};
	\node[draw,circle,fill=black,inner sep=0pt, minimum width=3pt] (65) at (18.75,2.165) {};
	\node[draw,circle,fill=black,inner sep=0pt, minimum width=3pt] (66) at (20.5,1.732) {};
	\node[draw,circle,fill=black,inner sep=0pt, minimum width=3pt] (67) at (21.25,1.209) {};
	\node[draw,circle,fill=black,inner sep=0pt, minimum width=3pt] (68) at (21.25,2.165) {};

	\foreach \x in {33,36,...,66}{
		\pgfmathtruncatemacro{\xx}{\x+1};
		\pgfmathtruncatemacro{\xxx}{\x+2};
		\draw[dotted,ns1,C2!50](\x)--(\xx)--(\xxx)--(\x);	
	}
	
	\draw[dashed,ns1,C3!50]
	(18.25,2.8)--(65)(21.75,2.8)--(68)(15.75,2.8)--(62)(12.25,2.8)--(59)(18.25,-2.8)--(41)(21.75,-2.8)--(44)(15.75,-2.8)--(38)(12.25,-2.8)--(35)(33)--(36)(39)--(42)(34)--(47)(46)--(58)(57)--(60)(63)--(66)(43)--(56)(55)--(67)(37)--(50)(49)--(61)(40)--(53)(52)--(64)(48)--(51);
	\draw[ns1,C1] (49).. controls (16.7,0.7) and (17.3,0.7) ..(52);
	\draw[ns1,C1] (50).. controls (16.7,-0.7) and (17.3,-0.7) ..(53);
	
	\draw[ns1,C1] (49).. controls (16,-0.5) and (17,-0.2) .. (53);
	\draw[ns1,C1] (50).. controls (15,0) and (15,2) ..(52);
	\end{tikzpicture}
	\caption{Zig-zag product of a $3$-regular grid with a triangle} \label{fig:zigZagProduct}
\end{figure*}

\begin{theorem}[\cite{ZigZagProductIntroduction}]\label{thm:expansionOfZigZag}
	If $G_1$ is an $(N_1,D_1,\lambda_1)$-graph and $G_2$ is a $(D_1,D_2,\lambda_2)$-graph then $G_1\zigzag G_2$ is a $(N_1\cdot D_1,D_2^2,g(\lambda_1,\lambda_2))$-graph, where 
	\begin{displaymath}
		g(\lambda_1,\lambda_2)=\frac{1}{2}(1-\lambda_2^2)\lambda_1+\frac{1}{2}\sqrt{(1-\lambda_2^2)^2\lambda_1+4\lambda_2^2}.
	\end{displaymath}
	This function has the following properties.
	\begin{enumerate}
		\item If both $\lambda_1<1$ and $\lambda_2<1$ then $g(\lambda_1,\lambda_2)< 1$. 
		\item $g(\lambda_1,\lambda_2)<\lambda_1+\lambda_2$.
	\end{enumerate}	
\end{theorem}

\begin{definition}[\cite{Hoory06expandergraphs}]\label{dfn:expanders}
	Let $D$ be a sufficiently large prime power (e.g. $D=2^{16}$). Let $H$ be a  $(D^4,D,{1}/{4})$ expander (explicit constructions for $H$ exist, cf.~\cite{ZigZagProductIntroduction}.) 
	We define $\{G_{m}\}_{m\in\Npos}$ by 
	\begin{eqnarray}
	G_1:=H^2, \hspace{20pt} G_m:=G_{m-1}^2\zigzag H \text{ for }m >1. \label{eqn:zigzagconstruction}
	\end{eqnarray}
\end{definition}
\begin{proposition}[\cite{Hoory06expandergraphs}]\label{prop:recursiveConstruction}
	For every $m\in\Npos$, the graph $G_m$ is a $(D^{4m},D^2,1/2)$-graph.
\end{proposition}
In the next section we will use the following lemma whose proof is deferred to Appendix~\ref{sec:proof_app}.
\begin{lemma}\label{lem:nonBipartitenessConnectedness}
	Let $G$ be a $D$-regular graph and $S$ be the vertices of a connected component of $G^2$. Then $\lambda(G^2[S])< 1$.
\end{lemma}

\section{A class of expanders definable in FO}\label{sec: definitionFormula}
In this section we define a formula such that the underlying graphs of its models are expanders.
We start with a high-level description of the formula. 
Let $\{G_m\}_{m \in \Npos}$ be as in Definition~\ref{dfn:expanders}. 
Loosely speaking, each model of our formula is a structure which consists of the disjoint union of $G_1,\dots,G_n$ for some $n\in \Npos$ 
with some underlying tree structure connecting $G_{m-1}$ to $G_{m}$ for all $m\in \{2,\dots,n\}$. For illustration see Figure~\ref{fig:modelOfFormula}. 
The tree structure enables us to provide an FO-checkable certificate for this construction.
The tree structure is a  $D^4$-ary tree, that is used to connect a vertex $v$ of $G_{m-1}$ to every vertex of the copy of $H$ which will replace $v$ in $G_{m}$. 
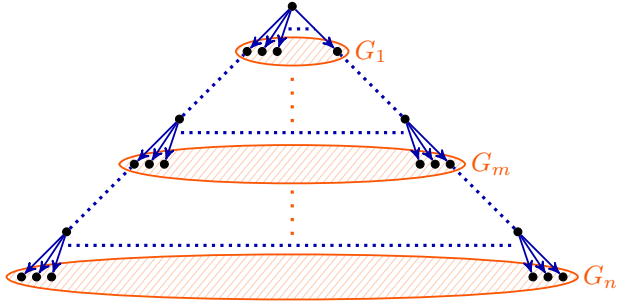
\begin{figure}
	\centering
	\begin{tikzpicture}
	\tikzstyle{ns1}=[line width=0.7]
	\tikzstyle{ns2}=[line width=1.2]
	\def \depth {3.6}
	\def \middle {2.1}
	\def \slope {1}
	\def \r {0.3}
	\def \levelDist {0.6}
	\begin{scope}
	\clip[postaction={fill=white,fill opacity=0.2}] (\depth*\slope+0.2,-\depth)..controls(\depth*\slope+0.1,-\depth+0.4)and(-\depth*\slope-0.1,-\depth+0.4)..(-\depth*\slope-0.2,-\depth)..controls(-\depth*\slope-0.1,-\depth-0.4)and(\depth*\slope+0.1,-\depth-0.4)..(\depth*\slope+0.2,-\depth);			
	\foreach \x in {-7.1,-7,...,15.2}%
	\draw[C3!30](\x, -5)--+(12,14.4);
	\end{scope}
	\draw[ns1,C3](\depth*\slope+0.2,-\depth)..controls(\depth*\slope+0.1,-\depth+0.4)and(-\depth*\slope-0.1,-\depth+0.4)..(-\depth*\slope-0.2,-\depth)..controls(-\depth*\slope-0.1,-\depth-0.4)and(\depth*\slope+0.1,-\depth-0.4)..(\depth*\slope+0.2,-\depth);
	\begin{scope}
	\clip[postaction={fill=white,fill opacity=0.2}] (\middle*\slope+0.2,-\middle)..controls(\middle*\slope+0.1,-\middle+0.34)and(-\middle*\slope-0.1,-\middle+0.34)..(-\middle*\slope-0.2,-\middle)..controls(-\middle*\slope-0.1,-\middle-0.34)and(\middle*\slope+0.1,-\middle-0.34)..(\middle*\slope+0.2,-\middle);		
	\foreach \x in {-7.1,-7,...,15.2}%
	\draw[C3!30](\x, -5)--+(12,14.4);
	\end{scope}
	\draw[ns1,C3](\middle*\slope+0.2,-\middle)..controls(\middle*\slope+0.1,-\middle+0.34)and(-\middle*\slope-0.1,-\middle+0.34)..(-\middle*\slope-0.2,-\middle)..controls(-\middle*\slope-0.1,-\middle-0.34)and(\middle*\slope+0.1,-\middle-0.34)..(\middle*\slope+0.2,-\middle);
	\begin{scope}
	\clip[postaction={fill=white,fill opacity=0.2}] (\levelDist*\slope+0.15,-\levelDist)..controls(\levelDist*\slope+0.1,-\levelDist+0.25)and(-\levelDist*\slope-0.1,-\levelDist+0.25)..(-\levelDist*\slope-0.15,-\levelDist)..controls(-\levelDist*\slope-0.1,-\levelDist-0.25)and(\levelDist*\slope+0.1,-\levelDist-0.25)..(\levelDist*\slope+0.15,-\levelDist);			
	\foreach \x in {-7.1,-7,...,15.2}%
	\draw[C3!30](\x, -5)--+(12,14.4);
	\end{scope}
	\draw[ns1,C3](\levelDist*\slope+0.15,-\levelDist)..controls(\levelDist*\slope+0.1,-\levelDist+0.25)and(-\levelDist*\slope-0.1,-\levelDist+0.25)..(-\levelDist*\slope-0.15,-\levelDist)..controls(-\levelDist*\slope-0.1,-\levelDist-0.25)and(\levelDist*\slope+0.1,-\levelDist-0.25)..(\levelDist*\slope+0.15,-\levelDist);
	\node[draw,circle,fill=black,inner sep=0pt, minimum width=3pt] (1) at (0,0) {};
	\node[draw,circle,fill=black,inner sep=0pt, minimum width=3pt] (2) at (\depth*\slope,-\depth) {};
	\node[draw,circle,fill=black,inner sep=0pt, minimum width=3pt] (3) at (-\depth*\slope,-\depth) {};
	\node[draw,circle,fill=black,inner sep=0pt, minimum width=3pt] (4) at (-\levelDist*\slope,-\levelDist) {};
	\node[draw,circle,fill=black,inner sep=0pt, minimum width=3pt] (5) at (-\levelDist*\slope+0.2,-\levelDist) {};
	\node[draw,circle,fill=black,inner sep=0pt, minimum width=3pt] (6) at (-\levelDist*\slope+0.4,-\levelDist) {};
	\node[draw,circle,fill=black,inner sep=0pt, minimum width=3pt] (7) at (\levelDist*\slope,-\levelDist) {};
	\node[draw,circle,fill=black,inner sep=0pt, minimum width=3pt] (8) at (\depth*\slope-\levelDist*\slope,-\depth+\levelDist) {};
	\node[draw,circle,fill=black,inner sep=0pt, minimum width=3pt] (9) at (-\depth*\slope+\levelDist*\slope,-\depth+\levelDist) {};
	\node[draw,circle,fill=black,inner sep=0pt, minimum width=3pt] (10) at (-\depth*\slope+0.2,-\depth) {};
	\node[draw,circle,fill=black,inner sep=0pt, minimum width=3pt] (11) at (-\depth*\slope+0.4,-\depth) {};
	\node[draw,circle,fill=black,inner sep=0pt, minimum width=3pt] (12) at (\depth*\slope-0.2,-\depth) {};
	\node[draw,circle,fill=black,inner sep=0pt, minimum width=3pt] (13) at (\depth*\slope-0.4,-\depth) {};
	
	\node[draw,circle,fill=black,inner sep=0pt, minimum width=3pt] (14) at (\middle*\slope-\levelDist*\slope,-\middle+\levelDist) {};
	\node[draw,circle,fill=black,inner sep=0pt, minimum width=3pt] (15) at (-\middle*\slope+\levelDist*\slope,-\middle+\levelDist) {};
	\node[draw,circle,fill=black,inner sep=0pt, minimum width=3pt] (16) at (\middle*\slope,-\middle) {};
	\node[draw,circle,fill=black,inner sep=0pt, minimum width=3pt] (17) at (-\middle*\slope,-\middle) {};
	\node[draw,circle,fill=black,inner sep=0pt, minimum width=3pt] (18) at (\middle*\slope-0.2,-\middle) {};
	\node[draw,circle,fill=black,inner sep=0pt, minimum width=3pt] (19) at (-\middle*\slope+0.2,-\middle) {};
	\node[draw,circle,fill=black,inner sep=0pt, minimum width=3pt] (20) at (\middle*\slope-0.4,-\middle) {};
	\node[draw,circle,fill=black,inner sep=0pt, minimum width=3pt] (21) at (-\middle*\slope+0.4,-\middle) {};

	\draw[ns2,C2,dotted] (15)--(4);
	\draw[ns2,C2,dotted] (9)--(17);
	\draw[ns2,C2,dotted] (8)--(16);
	\draw[ns2,C2,dotted] (14)--(7);
	\draw[->, >=stealth',ns1,C2](1)--(4);
	\draw[->, >=stealth',ns1,C2](1)--(5);
	\draw[->, >=stealth',ns1,C2](1)--(6);
	\draw[->, >=stealth',ns1,C2](1)--(7);
	\draw[->, >=stealth',ns1,C2](8)--(2);
	\draw[->, >=stealth',ns1,C2](9)--(3);
	\draw[->, >=stealth',ns1,C2](9)--(10);
	\draw[->, >=stealth',ns1,C2](9)--(11);
	\draw[->, >=stealth',ns1,C2](8)--(12);
	\draw[->, >=stealth',ns1,C2](8)--(13);
	\draw[->, >=stealth',ns1,C2](14)--(16);
	\draw[->, >=stealth',ns1,C2](14)--(18);
	\draw[->, >=stealth',ns1,C2](14)--(20);
	\draw[->, >=stealth',ns1,C2](15)--(17);
	\draw[->, >=stealth',ns1,C2](15)--(19);
	\draw[->, >=stealth',ns1,C2](15)--(21);
	\draw[dotted,ns2,C2](-0.05,-0.3)--(0.3,-0.3);
	\draw[dotted,ns2,C2](-\depth*\slope+0.7*\levelDist*\slope+0.2,-\depth+0.7*\levelDist)--(\depth*\slope-0.7*\levelDist*\slope-0.2,-\depth+0.7*\levelDist);
	\draw[dotted,ns2,C2](-\middle*\slope+0.7*\levelDist*\slope+0.2,-\middle+0.7*\levelDist)--(\middle*\slope-0.7*\levelDist*\slope-0.2,-\middle+0.7*\levelDist);
	\draw[loosely dotted,ns2,C3](0,-\levelDist-0.35)--(0,-\middle+0.5);
	\draw[loosely dotted,ns2,C3](0,-\middle-0.35)--(0,-\depth+0.5);
	\node[minimum height=10pt,inner sep=0,font=\small,C3] at (1.05,-\levelDist) {$G_1$};
	\node[minimum height=10pt,inner sep=0,font=\small,C3] at (2.65,-\middle) {$G_m$};
	\node[minimum height=10pt,inner sep=0,font=\small,C3] at (4.1,-\depth) {$G_{n}$
	};

	\end{tikzpicture}
	\caption{Schematic representation of a model of $\varphi_{\zigzag}$, where the parts in red (grey) only contain relations from $E$ and relations in $F$ are blue (black). Relation $R$ and $L$ are omitted.}\label{fig:modelOfFormula}
\end{figure}
We use $D^4$ relations $\{F_{k}\}_{k \in \indexSetH}$ to enforce an ordering on the $D^4$ children of each vertex. We use additional relations to encode rotation maps. 
For $i,j\in \indexSetRotation$ let  $E_{i,j}$ be a binary relation. For every pair $i,j\in \indexSetRotation$ we represent an  edge  $\{v,w\}$ in  $G_m$ by the two tuples $(v,w)\in E_{i,j}^\mathcal{A}$ 
and $(w,v)\in E_{j,i}^\mathcal{A}$. 
This allows us to encode the relationship $\rot_{G_m}(v,i)=(w,j)$ in first-order logic using the formula $`E_{i,j}(v,w)$'.

We use auxiliary relations $R$ and $L_{k}$ for $k\in \indexSetH$, to force the models to be degree regular. The relation $R$ contains the tuple $(r,r)$ for the  root $r$ of the tree, and $L_{k}$ will contain the tuple $(v,v)$ for every leaf $v$ of the tree. 

We now give the precise definition of the formula. We use $[n]:=\{0,1,\dots,n-1\}$ for $n\in \mathbb{N}$. Let 
\begin{displaymath}
\sigma:=\{ \{E_{i,j}\}_{i,j \in \indexSetRotation},\{F_{k}\}_{k \in \indexSetH},R,\{L_k\}_{k\in \indexSetH}\},
\end{displaymath} 
where $E_{i,j}$, $F_{k}$, $R$ and $L_k$ are binary relation symbols for $i,j\in [D]^2$ and $k\in ([D]^2)^2$.  For convenience
we let $E:=\bigcup_{i,j\in \indexSetRotation}E_{i,j}$ and $F:=\bigcup_{k\in \indexSetH}F_{k}$. 
Note that we can express the relations $E$ and $F$ in our language, by replacing formulas of the form '$E(x,y)$' by '$\bigvee_{i,j\in\indexSetRotation}E_{i,j}(x,y)$' and formulas of the form `$F(x,y)$' by
`$\bigvee_{k\in \indexSetH}F_{k}(x,y)$' below. 
We use the following formula to identify the root
$\varphi_{\operatorname{root}}(x):= \forall y \lnot F(y,x).$

We now define a formula $\varphi_{\operatorname{tree}}$, which expresses that the model restricted to the relation $F$ locally looks like a $D^4$-ary tree. More precisely, the formula defines that the structure has exactly one root, that every other vertex has exactly one parent and every vertex has either no children or exactly one child for each of the $D^4$ relations $F_k$. It also defines the self-loops used to make the structure degree regular.
\begin{align*}
\varphi_{\operatorname{tree}}&:= \exists^{=1} x \varphi_{\operatorname{root}}(x)\land \forall x \Big(\big(\varphi_{\operatorname{root}}(x)\land R(x,x)\big)\lor 
\big(\exists^{=1} y F(y,x)\land \lnot \exists y R(x,y)\land \lnot \exists y R(y,x)\big)\Big)\land
\\&\forall x \bigg(\Big[\neg\exists y F(x,y)\land \bigwedge_{k\in \indexSetH} L_k(x,x)\land \forall y \big(y\not= x \rightarrow \bigwedge_{k\in \indexSetH}\lnot L_k(x,y) \wedge \bigwedge_{k\in \indexSetH}\lnot L_k(y,x)\big)\Big]
\\&\lor \Big[\lnot\exists y \bigvee_{k\in \indexSetH}\big(L_k(x,y)\lor L_k(y,x)\big)\land 
\\&
\bigwedge_{k\in \indexSetH}\exists y_{k} \Big(x\not=y_{k}\land F_{k}(x,y_{k})\land 
(\bigwedge_{k'\in \indexSetH,k'\not=k}\lnot F_{k'}(x,y_k))\land \forall y(y\not=y_k\rightarrow \lnot F_{k}(x,y))\Big)\Big]\bigg).
\end{align*}

\bigskip
The formula $\varphi_{\operatorname{rotationMap}}$ will define the properties the relations in $E$ need to have in order to encode rotation maps of $D^2$-regular graphs.  For this we  make sure that the edge colours encode a map, i.e. for any pair of a vertex $x$ and index $i\in \indexSetRotation$ there is only one pair of vertex $y$ and index $j\in \indexSetRotation$ such that $E_{i,j}(x,y)$ holds and that the map is self inverse, i.e. if $E_{i,j}(x,y)$ then $E_{j,i}(y,x)$.
\begin{align*}
&\varphi_{\operatorname{rotationMap}}:= \forall x \forall y \Big(\bigwedge_{i,j\in \indexSetRotation}(E_{i,j}(x,y)\rightarrow E_{j,i}(y,x))\Big)
\land
\\&\forall x \Big(\bigwedge_{i\in \indexSetRotation}\Big(\bigvee_{j\in \indexSetRotation}\big(\exists^{=1}y E_{i,j}(x,y)\land \bigwedge_{\substack{j'\in \indexSetRotation\\ j'\not=j}}\lnot \exists y E_{i,j'}(x,y)\big)\Big)\Big)
\end{align*}

We now define a formula $\varphi_{\operatorname{base}}$ which expresses that the root $r$ of the tree has a self-loop $(r,r)$ in each relation $E_{i,j}$ and that the $D^4$ children of the root form $G_1$. Let $H$ be the $(D^4,D,1/4)$-graph from Definition~\ref{dfn:expanders}. We assume that  $H$ has vertex set $\indexSetH$. We then identify vertex $k\in\indexSetH$ with the element $y$ such that $(x,y)\in F_k^\mathcal{A}$ for the root $x$. Let  $\rot_H:\indexSetH \times[D]\rightarrow \indexSetH\times[D]$ be any rotation map of $H$. Fixing a rotation map for $H$ fixes the rotation map for $H^2$. Recall that $G_1:=H^2$. We can define $G_1$ by a conjunction over all edges of $G_1$.  
\begin{align*}
\varphi_{\operatorname{base}}:=&\forall x  \Big(\varphi_{\operatorname{root}}(x)\rightarrow\Big[\bigwedge_{i,j\in \indexSetRotation}\Big(E_{i,j}(x,x)\land \forall y \Big(x\not= y\rightarrow \big(\lnot E_{i,j}(x,y)\land \lnot E_{i,j}(y,x)\big)\Big)\Big)\land \\& \bigwedge_{\substack{ \rot_{H^2}(k,i)=(k',i')\\k,k'\in \indexSetH\\i,i'\in \indexSetRotation }}\exists y \exists y'\big(F_k(x,y)\land F_{k'}(x,y')\land E_{i,i'}(y,y')\big)\Big]\Big)
\end{align*}

We will now define a formula $\varphi_{\operatorname{recursion}}$ which will ensure that level $m$ of the tree contains $G_m$. Recall that  $G_m:= G_{m-1}^2\zigzag H$. We therefore express that if there is a path of length two between two vertices $x,z$ then for every pair $i,j\in [D]$ there is an edge connecting the corresponding children of $x$ and $z$ according to the definition of the zig-zag product. Here it is important that $x$ and $z$ either both have no children in the underlying tree structure or they both have children. This will also be encoded in the formula.
\begin{align*}
&\varphi_{\operatorname{recursion}}:= \forall x \forall z\bigg[\Big(\lnot \exists y F(x,y)\land \lnot \exists y F(z,y)\Big)\lor
\bigwedge_{\substack{k_1',k_2'\in \indexSetRotation\\\ell_1',\ell_2'\in \indexSetRotation}}\Big(\exists y \big[E_{k_1',\ell_1'}(x,y)\land E_{k_2',\ell_2'}(y,z)\big]\rightarrow \\&
\bigwedge_{\substack{i,j,i',j'\in [D], k,\ell\in \indexSetH\\\rot_H(k,i)=((k_1', k_2'),i')\\ \rot_H((\ell_2', \ell_1'),j)=(\ell,j')}}\exists x'\exists z'\big[ F_k(x,x')\land F_\ell(z,z')\land
E_{(i,j),(j',i')}(x',z')\big]\Big)\bigg]
\end{align*}

We finally let $\varphi_{\zigzag}:=\varphi_{\operatorname{tree}}\land \varphi_{\operatorname{rotationMap}}\land \varphi_{\operatorname{base}}\land \varphi_{\operatorname{recursion}}$.
This concludes defining the formula. Let $d:=2D^2+D^4+1$, which is chosen in such a way to allow for any element of a $\sigma$-structure in $C_d$ to be in $2D^2$ $E$-relations ($G_m$ is $D^2$ regular and every edge of $G_m$ is modelled by two directed edges), to have either $D^4$ $F$-children or $D^4$ $L$-self-loops and to either have one $F$-parent or be in one $R$-self-loop. 

To each model $\mathcal{A}$ of $\varphi_{\zigzag}$ we will associate an undirected graph $U(\mathcal{A})$ with vertex set $A$. For every tuple in each of the relations of $\mathcal{A}$, the graph $U(\mathcal{A})$ will have an edge. We will define $U(\mathcal{A})$ by a rotation map, which extends the rotation map encoded by the relation $E$. For this let $I:=\{0\}\sqcup \indexSetH\sqcup \indexSetRotation$ be an index set. Formally, we define the \emph{underlying graph} $U(\mathcal{A})$ of a model $\mathcal{A}$ of $\varphi_{\zigzag}$ to be the undirected graph with vertex set $A$ given by the rotation map $\rot_{U(\mathcal{A})}: A\times I\rightarrow A\times I$ defined by 
\begin{align*}
\rot_{U(\mathcal{A})}(v,i):=\begin{cases} 
(v,0) & \text{if }i=0\text{ and }(v,v)\in R^\mathcal{A} \\
(w,j) & \text{if }i=0\text{ and }(w,v)\in F_j^\mathcal{A}\\
(w,0) & \text{if }i\in \indexSetH \text{ and }(v,w)\in F_i^\mathcal{A}\\
(v,i) & \text{if }i\in \indexSetH \text{ and }(v,v)\in L_i^\mathcal{A}\\
(w,j) & \text{if }i\in \indexSetRotation \text{ and }(v,w)\in E_{i,j}^\mathcal{A}. 
\end{cases}
\end{align*} 
We can understand this rotation map as labelling the edges incident to a vertex $v$ as follows: $(v,v)\in R^{\mathcal{A}}$ or $(w,v)\in F_k^\mathcal{A}$ respectively is labelled by  $0$, $(v,w)\in F_k^\mathcal{A}$ or $(v,v)\in L_k^\mathcal{A}$ respectively is labelled by $k$ and $(v,w)\in E_{i,j}^\mathcal{A}$ is labelled by $i$.
Note that $U(\mathcal{A})$ is $(D^2+D^4+1)$-regular. We chose the notion of an underlying graph here instead of the Gaifman graph {(see Appendix \ref{sec:hanftheorem})}, 
and it is  more convenient in particular for using results from \cite{ZigZagProductIntroduction}. However the Gaifman graph can be obtained from the underlying graph by ignoring self-loops and multiple edges. We use $\mathcal{A}\models \varphi$ to denote that $\mathcal{A}$ is a model of an FO sentence $\varphi$ (see Appendix \ref{sec:foappendix}) and 
we show the following. 

\begin{theorem}\label{thm:expansionOfModels}
	There is an $\epsilon>0$ such that the class $\{U(\mathcal{A})\mid \mathcal{A}\models \varphi_{\zigzag}\}$ is a family of $\epsilon$-expanders.
\end{theorem}


In the rest of this section, we give the proof of Theorem \ref{thm:expansionOfModels}.
Let $\mathcal{A}$ be a model of $\varphi_{\zigzag}$. 
Let  $\mathcal{A}|_F:=(A,(F_k^\mathcal{A})_{k\in \indexSetH})$ be an $\{(F_{k})_{k\in \indexSetH}\}$-structure. Recall that we denote the Gaifman graph of $\mathcal{A}|_F$ by $G(\mathcal{A}|_F)$. 
Let $\mathcal{A}|_E$ be the $\{(E_{i,j})_{i,j\in \indexSetRotation}\}$-structure $(A,(E_{i,j}^\mathcal{A})_{i,j\in \indexSetRotation})$.  We further define the \emph{underlying graph} $U(\mathcal{A}|_E)$ of $\mathcal{A}|_E$ as the undirected graph specified by the rotation map  $\rot_{U(\mathcal{A}|_E)}$ defined by $\rot_{U(\mathcal{A}|_E)}(v,i):=(w,j)$ if $(v,w)\in E_{i,j}^\mathcal{A}$. This is well defined as $\mathcal{A}\models \varphi_{\operatorname{rotationMap}}$.
We use the substructures $G(\mathcal{A}|_F)$ and $U(\mathcal{A}|_E)$ to express the structural properties of models of $\varphi_{\zigzag}$. More precisely we want to prove that $G(\mathcal{A}|_F)$ is a rooted complete tree and $U(\mathcal{A}|_E)$ is the disjoint union of the expanders $G_1,\dots,G_{n}$ for some $n\in \mathbb{N}$ (Lemma~\ref{lem:exactFormOfModels}). To prove this we use two technical lemmas (Lemma~\ref{lem:recursion} and Lemma~\ref{lem:connected}). Lemma~\ref{lem:recursion} intuitively shows that the children in $G(\mathcal{A}|_F)$ of each connected part of $U(\mathcal{A}|_E)$  form the  zig-zag product with $H$ of the square of the connected part. Lemma~\ref{lem:connected} shows that $G(\mathcal{A}|_F)$ is connected. To prove Theorem \ref{thm:expansionOfModels} we use that a tree with an expander on each level has good expansion. Loosely speaking, this is true because cutting the tree `horizontally' takes many edge deletions and for cutting the tree `vertically' we  cut many expanders. 
We define isomorphism for undirected graphs with parallel edges and self-loops in the usual way (see Appendix \ref{app:undirectedGraphs}), and we use $G_1\cong G_2$ to denote that $G_1$ is isomorphic to $G_2$.

\begin{lemma}\label{lem:recursion}
	Let $\mathcal{A}$ be a model of $\varphi_{\zigzag}$ and assume $S$ is the set of all vertices belonging to a  connected component  
	of $(U(\mathcal{A}|_E))^2$ not containing the root and let $S':=\{w\in A\mid (v,w)\in F^\mathcal{A}, v\in S\}$.  
	If $S'\not=\emptyset$ then  $U(\mathcal{A}|_E)[S']$ is a connected component of $U(\mathcal{A}|_E)$ and  $U(\mathcal{A}|_E)[S']\cong((U(\mathcal{A}|_E))^2[S])\zigzag H$. 
\end{lemma} 
We use connected components of $(U(\mathcal{A}|_E))^2$, as the square  of a connected component of $U(\mathcal{A}|_E)$ may not be connected, in which case the zig-zag product with $H$ of the square of the connected component cannot be connected. 
\begin{proof}[Proof of Lemma \ref{lem:recursion}] Assume that $S'\not=\emptyset$. We first show that $U(\mathcal{A}|_E)[S']\cong((U(\mathcal{A}|_E))^2[S])\zigzag H$.  
	For this we use the following two claims. 
	\begin{claim}\label{claim:edgeInducesPathOfLenghTwoOnParents}
		If $\rot_{(U(\mathcal{A}|_E))^2[S]\zigzag H}((u,k),(i,j))=((w,\ell),(j',i'))$ for some $u,w\in S$, $k,\ell\in \indexSetH$, $i,j,i',j'\in [D]$  then there is $v\in S$ such that $(u,v)\in E_{k_1',\ell_1'}^\mathcal{A}$ and $(v,w)\in E_{k_2',\ell_2'}^\mathcal{A}$ where $\rot_{H}(k,i)=((k_1', k_2'),i')$ and $\rot_{H}((\ell_2', \ell_1'),j)=(\ell,j')$. 
	\end{claim}
	\begin{proof}
	By the assumption that 
		$\rot_{(U(\mathcal{A}|_E))^2[S]\zigzag H}((u,k),(i,j))=((w,\ell),(j',i'))$ and the definition of the zig-zag product, we have that $\rot_{(U(\mathcal{A}|_E))^2[S]}(u,(k_1', k_2'))=(w,(\ell_2', \ell_1'))$ for   $\rot_{H}(k,i)=((k_1', k_2'),i')$ and $\rot_{H}((\ell_2', \ell_1'),j)=(\ell,j')$. 
		
		Since $\rot_{(U(\mathcal{A}|_E))^2[S]}$ is equal to $\rot_{(U(\mathcal{A}|_E))^2}$ restricted to elements of the set $S$, we have that $\rot_{(U(\mathcal{A}|_E))^2}(u,(k_1', k_2'))=(w,(\ell_2', \ell_1'))$. Then by definition of squaring $\rot_{(U(\mathcal{A}|_E))^2}(u,(k_1', k_2'))=(w,(\ell_2', \ell_1'))$ implies that there is $v$ such that $\rot_{U(\mathcal{A}|_E)}(u,k_1')=(v,\ell_1')$  and $\rot_{U(\mathcal{A}|_E)}(v,k_2')=(w,\ell_2')$.
	\end{proof}
	
	\begin{claim}\label{claim:pathOfLenghTwoCausesNoneOrTwoChildren}
		If  $(u,v)\in E_{k_1',\ell_1'}^\mathcal{A}$ and $(v,w)\in E_{k_2',\ell_2'}^\mathcal{A}$ for $u,v,w\in A$, $k_1',k_2',\ell_1',\ell_2'\in \indexSetH$ and there is $u'\in A$ with $(u,u')\in F^\mathcal{A}$ then there is $w'\in A$ such that $(w,w')\in F^\mathcal{A}$. Furthermore for any $i,i',j,j'\in [D]$ there are $\tilde{u},\tilde{w}\in A$, $k,\ell\in \indexSetH$ such that $(\tilde{u},\tilde{w})\in E_{(i,j),(j'i')}^\mathcal{A}$ for $(u,\tilde{u})\in F_{k}^\mathcal{A}$ and $(w,\tilde{w})\in F_{\ell}^\mathcal{A}$ where $\rot_{H}(k,i)=((k_1', k_2'),i')$ and $\rot_{H}((\ell_2', \ell_1'),j)=(\ell,j')$.
	\end{claim}
	\begin{proof}
		We only use that $\mathcal{A}\models \varphi_{\operatorname{recursion}}$. Since $\varphi_{\operatorname{recursion}}$ has the form $\forall x\forall z \psi(x,z)$ we know that $\mathcal{A}\models \psi(u,w)$.
		Since $(u,u')\in F^\mathcal{A}$ we have  $\mathcal{A}\not\models \big[ \lnot \exists y F(x,y)\land \lnot \exists y F(z,y)\big](u,w)$. Since  $\mathcal{A}\models \exists y \big[E_{k_1',\ell_1'}(x,y)\land E_{k_2',\ell_2'}(y,z)\big](u,w)$  
		\begin{align*}
		&\mathcal{A}\models \bigwedge_{\substack{i,j,i',j'\in [D], k,\ell\in \indexSetH\\\rot_H(k,i)=((k_1', k_2'),i')\\ \rot_H((\ell_2', \ell_1'),j)=(\ell,j')}}\exists x'\exists z'\big[ F_k(x,x')\land 
		F_\ell(z,z')\land E_{(i,j),(j',i')}(x',z')\big](u,w).
		\end{align*}
		Since $H$ is $D$-regular, for every $k'_1,k'_2 \in [D]^2$ and $i,i' \in [D]$, there is $k \in ([D]^2)^2$ such that $\rot_H(k,i) = ((k'_1,k'_2,i')$ (and the same for $\ell'_1,\ell'_2,j,j'$). Thus, the above conjunction is not empty. This further implies that for any $i,i',j,j'\in [D]$ there are $\tilde{u},\tilde{w}\in A$, $k,\ell\in \indexSetH$ as claimed. In particular there is $w' \in A$ such that $(w,w') \in F^{\mathcal{A}}$. 
	\end{proof}

	We will argue that for every element $w\in S$ there is a $w'\in S'$ such that $(w,w')\in F^\mathcal{A}$.  For this pick any $u'\in S'$. Let $u\in S$ be the element such that  $(u,u')\in F^\mathcal{A}$. By combining Lemma \ref{lem:nonBipartitenessConnectedness} and Theorem \ref{thm:expansionOfZigZag} and Lemma \ref{lem:connectedBipartiteEigenvalues} it follows that $((U(\mathcal{A}|_E))^2[S])\zigzag H$ is connected. Therefore, there is a path $(u'_0,\dots,u'_m)$ in $((U(\mathcal{A}|_E))^2[S])\zigzag H$ from $u'_0=(u,(k_1,k_2))$ to $u'_m=(w,(\ell_1,\ell_2))$ for some $k_1,k_2,\ell_1,\ell_2\in \indexSetRotation$. By Claim \ref{claim:edgeInducesPathOfLenghTwoOnParents} there is a path $(u_0,v_0,u_1,v_1,\dots u_{m-1},v_{m-1},u_m)$ in $U(\mathcal{A}|_E)$ from $u_0=u$ to $u_m=w$. By inductively using Claim \ref{claim:pathOfLenghTwoCausesNoneOrTwoChildren} on the path we find $w'$ such that $(w,w')\in F^\mathcal{A}$.

	Combining this with $\mathcal{A}\models \varphi_{\operatorname{tree}}$ implies that the map $f:S\times \indexSetH \rightarrow S'$, given by $f(v,k)=u$ if $(v,u)\in F_{k}^\mathcal{A}$, is well defined. Furthermore, by Claim~\ref{claim:edgeInducesPathOfLenghTwoOnParents} and Claim \ref{claim:pathOfLenghTwoCausesNoneOrTwoChildren}, we have that if $\rot_{(U(\mathcal{A}|_E))^2[S]\zigzag H}((u,k),(i,j))=((w,\ell),(j',i'))$ then \[\rot_{(U(\mathcal{A}|_E))[S']}(f((u,k)),(i,j))=(f((w,\ell)),(j',i')).\] This proves that  $f$ maps each edge in $((U(\mathcal{A}|_E))^2[S])\zigzag H$ injectively to  an edge in $U(\mathcal{A}|_E)[S']$. Then the map $f$ together with the corresponding edge map is an isomorphism from $((U(\mathcal{A}|_E))^2[S])\zigzag H$ to $U(\mathcal{A}|_E)$ as both are $D^2$-regular. 
	
	Moreover, $U(\mathcal{A}|_E)[S']\cong((U(\mathcal{A}|_E))^2[S])\zigzag H$ implies that  $U(\mathcal{A}|_E)[S']$ is connected and $D^2$-regular. Since $\mathcal{A}\models \varphi_{\operatorname{rotationMap}}$ enforces that $U(\mathcal{A}|_E)$ is $D^2$-regular, no vertex $v\in S'$ can have  neighbours which are not in $S'$ and therefore $U(\mathcal{A}|_E)[S']$ is a connected component of $U(\mathcal{A}|_E)$.
\end{proof}

\begin{lemma}\label{lem:connected}
	Let $\mathcal{A}\in C_d$ be a model  of $\varphi_{\zigzag}$. Then $G(\mathcal{A}|_F)$  is connected.
\end{lemma}
\begin{proof}
	Assume that this is false and $G(\mathcal{A}|_F)$ has more than one connected component. Since  $\mathcal{A}\models \varphi_{\operatorname{tree}}$  there is exactly one element $v$ such that $\mathcal{A}\models \varphi_{\operatorname{root}}(v)$. Therefore we can pick $G'$ to be a connected component of $G(\mathcal{A}|_F)$ which does not contain $v$. Let $V$ be the set of vertices of $G'$. 
For the next claim we should have in mind that $(\mathcal{A}|_F)[V]$ can be understood as a directed graph in which every vertex has in-degree $1$ and the corresponding undirected graph $G'$ is connected. Hence $(\mathcal{A}|_F)[V]$ must consist of a set of disjoint directed trees whose roots form a directed cycle. Consequently $G'$ has the structure as given in the following claim.

	\begin{claim}\label{claim:containsCycle}
		$G'$ contains a cycle $(c_0,\dots,c_{\ell-1})$ and for every vertex $v$ of $G'$ there is exactly one path $(p_0,\dots,p_m)$ in $G'$ with $p_0=v$, $p_m$ on the cycle and $p_i$ not on the cycle for all $i\in [m]$.
	\end{claim}
	\begin{proof}
		Let $v_0$ be any vertex in $G'$ and let $S_0=\{v_0\}$. We will now recursively define  $v_i$  to be the vertex of $G'$ such that $(v_i,v_{i-1})\in F^\mathcal{A}$. Such a vertex always exists by the choice of $G'$ (i.e. that the root is not in $G'$) and the fact that $\mathcal{A} \models \varphi_{\operatorname{tree}}$. Furthermore, such a vertex is unique as $\mathcal{A}\models \varphi_{\operatorname{tree}}$. We also let $S_i:=S_{i-1}\cup \{v_i\}$. Since $A$ is finite the chain $S_0\subseteq S_1\subseteq \dots \subseteq S_i\subseteq \dots$ must become stationary at some point. Let $i\in \mathbb{N}$ be the minimum index such that $S_{i-1}=S_i$ and let $j< i$ be such that $v_i=v_j$. Then $(v_i,v_{i-1},\dots,v_{j+1},v_j)$ is a cycle in $G'$ as by construction $(v_k,v_{k-1})\in F^\mathcal{A}$ which implies that $\{v_k,v_{k-1}\}$ is an edge in the Gaifman graph $G(\mathcal{A}|_F)$.
		Let $C=\{c_0,\dots,c_{\ell-1}\}$ be the vertices of the cycle. 
		Since $G'$ is connected a path that satisfies the property as described in the assertion of the claim always exists. So let us argue that such a path is unique. Assume  there are two different such path  $(p_0,\dots,p_m)$ and $(p'_0,\dots,p'_{m'})$ and assume that $p_m=c_i$ and $p'_{m'}=c_j$. Let $k\leq \min\{m,m'\}$ be the minimum index such that $p_k\not=p'_k$. Such an index must exist as the paths are different and as $p_0=p'_0=v$ we also know that $k\geq 1$. Since $\mathcal{A}\models \varphi_{\operatorname{tree}}$ for every vertex $w$ of $G'$ there can only be one vertex $w'$ of $G'$ such that $(w',w)\in F^\mathcal{A}$. As $p_{m-1}\notin C$ and $(c_{(i-1)\,\,\mod\,\, \ell},p_m)\in F^\mathcal{A}$ this means that $(p_{m},p_{m-1})\in F^\mathcal{A}$. Applying the argument inductively we get that $(p_k,p_{k-1})\in F^\mathcal{A}$. The same argument works for the path $(p'_0,\dots, p'_{m'})$ and therefore $(p'_k,p'_{k-1})\in F^\mathcal{A}$. By the choice of $k$ we know that $p_{k-1}=p'_{k-1}$ and $p_k\not=p'_k$ which contradicts $\mathcal{A}\models \varphi_{\operatorname{tree}}$.
	\end{proof}	
	
	Let $S_0$ be the vertex set of the connected component of $U(\mathcal{A}|_E)$ with $c_0\in S_0$. Note  that $S_0$ might not be contained in $G'$.  
	
	We now recursively define the infinite sequence of sets $S_i:=\{w\in A\mid (v,w)\in F^\mathcal{A}, v\in S_{i-1}\}$ for every $i\in\Npos$.  Let $m_i:=\max_{v\in S_i\cap V}\min_{j\in \{0,\dots,\ell-1\}}\{\dist_{G'}(c_j,v)\}$ and let $v_i\in S_i\cap V$ be a vertex of distance $m_i$ from $C$ in $G'$. Note here that $m_i$ is well defined as $c_{i\,\, \mod \,\, \ell}\in S_i$. 
	
	\begin{claim}\label{claim:characterisationOfG|_E[S_i]}
		$U(\mathcal{A}|_E)[S_{i}]=(U(\mathcal{A}|_E)[S_{i-1}])^2\zigzag H$.
		
	\end{claim}
	\begin{proof}
		We  show  the stronger statement that $U(\mathcal{A}|_E)[S_i]$ is a connected component of $U(\mathcal{A}|_E)$ and $(U(\mathcal{A}|_E)[S_{i}])^2\zigzag H=$ $U(\mathcal{A}|_E)[S_{i+1}]$ and  $\lambda(U(\mathcal{A}|_E)[S_i])<1$  for $i\in \mathbb{N}$ by induction.
		
		$U(\mathcal{A}|_E)[S_0]$ is a connected component of $U(\mathcal{A}|_E)$ by choice of $S_0$.	
		Let $\tilde{S}:=\{w\in A\mid (w,v)\in F^\mathcal{A}, v\in S_0 \}$. 
		 
		We now argue that $(U(\mathcal{A}|_E))^2[\tilde{S}]$ is a connected component of $(U(\mathcal{A}|_E))^2$. Assuming the contrary, every connected component of  $(U(\mathcal{A}|_E))^2$ either contains vertices from  both $\tilde{S}$ and $A\setminus \tilde{S}$, or $(U(\mathcal{A}|_E))^2[\tilde{S}]$ splits into more than one connected component. Let $S'$ be the vertices of a connected component as in the first case. Then $|S'|>1$ and hence $S'$ can not contain the root as the root is not in any $E$-relation. Hence by Lemma~\ref{lem:recursion} we get a connected component of $U(\mathcal{A}|_E)$ on the children of $S'$ both containing vertices from $S_0$ and from $A\setminus S_0$ which  contradicts $S_0$ being a connected component of $U(\mathcal{A}|_E)$. Now let $S'$ be a connected component as in the second case, and pick $S'$ such that it does not contain the root. 
		Then by Lemma~\ref{lem:recursion} $S_0$ must have a non-empty intersection with at least two connected components of $U(\mathcal{A}|_E)$ which is a contradiction. 
		
		Thus, by Lemma \ref{lem:nonBipartitenessConnectedness}  $\lambda((U(\mathcal{A}|_E))^2[\tilde{S}])<1$. But by Lemma \ref{lem:recursion} $U(\mathcal{A}|_E)[S_0]=((U(\mathcal{A}|_E))^2[\tilde{S}])\zigzag H$. Then Theorem \ref{thm:expansionOfZigZag} and  $\lambda(H)<1$ ensure that $\lambda(U(\mathcal{A}|_E)[S_0])<1$.

For $i>1$, by induction it holds that $\lambda(U(\mathcal{A}|_E)[S_{i-1}])<1$, which, together with Lemma~\ref{lem:expansionOfSquaring} and Lemma~\ref{lem:connectedBipartiteEigenvalues}, implies that $(U(\mathcal{A}|_E)[S_{i-1}])^2$ is a connected component\footnote{We remark that the statement that $(U(\mathcal{A}|_E)[S_{i-1}])^2$ is a connected component does not directly follow from the fact that $U(\mathcal{A}|_E)[S_{i-1}]$ is a connected component of $U(\mathcal{A}|_E)$, as the square of a connected bipartite graph is not necessarily connected.} of $(U(\mathcal{A}|_E))^2$ and that $(U(\mathcal{A}|_E))^2[S_{i-1}]=(U(\mathcal{A}|_E)[S_{i-1}])^2$. Since $c_{i\,\,\mod\,\, \ell}\in S_i$, by Lemma \ref{lem:recursion}, we have that $U(\mathcal{A}|_E)[S_i]$ is a connected component of $U(\mathcal{A}|_E)$ and $U(\mathcal{A}|_E)[S_i]= (U(\mathcal{A}|_E)[S_{i-1}])^2\zigzag H$.  Furthermore this proves $\lambda(U(\mathcal{A}|_E)[S_i])<1$ using Lemma \ref{lem:expansionOfSquaring} and Theorem \ref{thm:expansionOfZigZag}.
	\end{proof}
	\begin{claim}
		For every $v\in S_i$ there is $w\in V$ such that $(v,w)\in F^\mathcal{A}$.
	\end{claim}
	\begin{proof}
By Claim \ref{claim:characterisationOfG|_E[S_i]} $U(\mathcal{A}|_E)[S_{i+1}]=(U(\mathcal{A}|_E)[S_i])^2\zigzag H$. This means that by definition of squaring and the zig-zag product we know that $|S_{i+1}|=D^4\cdot |S_i|$. But because in addition $\mathcal{A}\models \varphi_{\operatorname{tree}}$  we know that every element $v\in S_i$ will contribute to no more then $D^4$ elements to $S_{i+1}$. This means by construction of $S_{i+1}$ that for every element in $S_i$ there must be $w\in V$ such that $(v,w)\in F^\mathcal{A}$. 
	\end{proof}
	Therefore there is $w_i\in V$ such that $(v_i,w_i)\in F^\mathcal{A}$. Let $(u_0,\dots,u_{m_i})$ be the path in $G'$ from $u_0=v_i$ to $u_{m_i}\in C$. Note that it is impossible that $w_i = u_1$. This is true as for the path $(u_0,...,u_{m_i})$, we have that $(u_{j+1},u_{j})\in F^\mathcal{A}$ for all $j\in [m_i]$. Furthermore, since $v_i=u_0\not=u_1$, assuming that $w_i=u_1$ would imply $(v_i,u_1),(u_2,u_1)\in F^\mathcal{A}$,  which contradicts $\mathcal{A}\models \varphi_{\operatorname{tree}}$.
		 Then $(w_i,u_0,\dots,u_{m_i})$ is a path in $G'$ from $w_i$ to $C$. 
		Since $w_i\in S_{i+1}$ by construction, Claim~\ref{claim:containsCycle}  implies that $m_{i+1}\geq m_i+1$. Therefore $m_i\geq i+m_0$ inductively. But this yields a contradiction, because $\ell+m_0\leq m_\ell=m_0$ and the length of the cycle $\ell>0$. See Figure~\ref{fig:illustrationOfCycleFreenessOfModelsOfFi} for an illustration. Therefore $G(\mathcal{A}|_F)$  must be connected. 
\end{proof}	

\begin{figure}
	\centering
	\begin{tikzpicture}
	\tikzstyle{ns1}=[line width=0.7]
	\tikzstyle{ns2}=[line width=1.2]
	\def \radius {2}
	\def \rad {4.5}
	\def \margin {2.7} 
	\begin{scope}
	\clip[postaction={fill=white,fill opacity=0.2}] (254:1.6)..controls(248:2.60)and(283:3.1)..(258:4.3)..controls(250:4.76)..(246:4.07)..controls(244:3.955)..(247:3.38)..controls(250:3.035)and(228:2.69)..(245:1.5)..controls(250:1.4)..(254:1.6);				
	\foreach \x in {-7.1,-7,...,15.2}%
	\draw[C3!20](\x, -5)--+(12,14.4);
	\end{scope}
	\draw[ns1,C3!50] (254:1.6)..controls(248:2.60)and(283:3.1)..(258:4.3)..controls(250:4.76)..(246:4.07)..controls(244:3.955)..(247:3.38)..controls(250:3.035)and(228:2.69)..(245:1.5)..controls(250:1.4)..(254:1.6);	
	\begin{scope}
	\clip[postaction={fill=white,fill opacity=0.2}] (253:1.6)..controls(250:2.27)and(270:2.54)..(255:2.9)..controls(250:3.08)..(246:2.81)..controls(245:2.765)..(247:2.54)..controls(250:2.405)and(230:2.27)..(246:1.5)..controls(250:1.4)..(253:1.6);			
	\foreach \x in {-7.1,-7,...,15.2}%
	\draw[C3!30](\x, -5)--+(12,14.4);
	\end{scope}
	\draw[ns1,C3] (253:1.6)..controls(250:2.27)and(270:2.54)..(255:2.9)..controls(250:3.08)..(246:2.81)..controls(245:2.765)..(247:2.54)..controls(250:2.405)and(230:2.27)..(246:1.5)..controls(250:1.4)..(253:1.6);
	\begin{scope}
	\clip[postaction={fill=white,fill opacity=0.2}] (293:1.6)..controls(290:2.36)and(310:2.72)..(295:3.12)..controls(290:3.344)..(286:3.08)..controls(285:3.02)..(287:2.72)..controls(290:2.54)and(270:2.36)..(286:1.5)..controls(290:1.4)..(293:1.6);			
	\foreach \x in {-7.1,-7,...,15.2}%
	\draw[C3!30](\x, -5)--+(12,14.4);
	\end{scope}
	\draw[ns1,C3] (293:1.6)..controls(290:2.36)and(310:2.72)..(295:3.12)..controls(290:3.344)..(286:3.08)..controls(285:3.02)..(287:2.72)..controls(290:2.54)and(270:2.36)..(286:1.5)..controls(290:1.4)..(293:1.6);
	\begin{scope}
	\clip[postaction={fill=white,fill opacity=0.2}] (333:1.6)..controls(330:2.45)and(350:2.9)..(335:3.5)..controls(330:3.8)..(326:3.35)..controls(325:3.275)..(327:2.9)..controls(330:2.675)and(310:2.45)..(326:1.5)..controls(330:1.4)..(333:1.6);				
	\foreach \x in {-7.1,-7,...,15.2}%
	\draw[C3!30](\x, -5)--+(12,14.4);
	\end{scope}
	\draw[ns1,C3] (333:1.6)..controls(330:2.45)and(350:2.9)..(335:3.5)..controls(330:3.8)..(326:3.35)..controls(325:3.275)..(327:2.9)..controls(330:2.675)and(310:2.45)..(326:1.5)..controls(330:1.4)..(333:1.6);	
	\begin{scope}
	\clip[postaction={fill=white,fill opacity=0.2}] (214:1.6)..controls(210:2.6)and(232:3.2)..(216:4)..controls(210:4.2)..(206:3.8)..controls(204:3.7)..(207:3.2)..controls(210:2.9)and(188:2.6)..(205:1.5)..controls(210:1.4)..(214:1.6);			
	\foreach \x in {-7.1,-7,...,15.2}%
	\draw[C3!30](\x, -5)--+(12,14.4);
	\end{scope}
	\draw[ns1,C3] (214:1.6)..controls(210:2.6)and(232:3.2)..(216:4)..controls(210:4.2)..(206:3.8)..controls(204:3.7)..(207:3.2)..controls(210:2.9)and(188:2.6)..(205:1.5)..controls(210:1.4)..(214:1.6);
	\node[draw,circle,fill=black,inner sep=0pt, minimum width=5pt] (1) at (210:\radius) {};
	\node[draw,circle,fill=black,inner sep=0pt, minimum width=5pt] (2) at (250:\radius) {};
	\node[draw,circle,fill=black,inner sep=0pt, minimum width=5pt] (3) at (290:\radius) {};
	\node[draw,circle,fill=black,inner sep=0pt, minimum width=5pt] (4) at (330:\radius) {};
	\draw[->, >=stealth',ns1,C2] ({210+\margin}:\radius) arc ({210+\margin}:{250-\margin}:\radius);
	\draw[->, >=stealth',ns1,C2] ({250+\margin}:\radius) arc ({250+\margin}:{290-\margin}:\radius);
	\draw[->, >=stealth',ns1,C2] ({290+\margin}:\radius) arc ({290+\margin}:{330-\margin}:\radius);
	\draw[->, >=stealth',ns1,C2] ({330+\margin}:\radius) arc ({330+\margin}:{370-\margin}:\radius);
	\draw[dashed,->, >=stealth',ns1,C2] ({-30+\margin}:\radius) arc ({-30+\margin}:{210-\margin}:\radius);
	\draw[->, >=stealth',ns1,C2] (250:1.43)--(290:1.45);
	\draw[->, >=stealth',ns1,C2] (250:1.43)--(285:1.52);
	\draw[->, >=stealth',ns1,C2] (250:1.43)--(283:1.61);
	\draw[loosely dotted,ns2,C2] (268:1.6)--(272:2.8);
	\draw[->, >=stealth',ns1,C2] (251:3.02)--(289:3.27);
	\draw[->, >=stealth',ns1,C2] (251:3.02)--(287:3.12);
	\draw[->, >=stealth',ns1,C2] (251:3.02)--(285.2:2.97);
	\draw[->, >=stealth',ns1,C2] (290:1.43)--(330:1.45);
	\draw[->, >=stealth',ns1,C2] (290:1.43)--(325:1.52);
	\draw[->, >=stealth',ns1,C2] (290:1.43)--(323:1.61);
	\draw[loosely dotted,ns2,C2] (308:1.65)--(312:3);
	\draw[->, >=stealth',ns1,C2] (290.5:3.28)--(330:3.69);
	\draw[->, >=stealth',ns1,C2] (290.5:3.28)--(327.3:3.49);
	\draw[->, >=stealth',ns1,C2] (290.5:3.28)--(325.8:3.34);
	\draw[ns1,C2] (330:1.43)--(355:1.45);
	\draw[ns1,C2] (330:1.43)--(352:1.52);
	\draw[ns1,C2] (330:1.43)--(349:1.59);
	\draw[ns1,C2] (330.5:3.7)--(341:3.86);
	\draw[ns1,C2] (330.5:3.7)--(338.3:3.75);
	\draw[ns1,C2] (330.5:3.7)--(335.8:3.65);
	\draw[->, >=stealth',ns1,C2] (210:1.43)--(250:1.45);
	\draw[->, >=stealth',ns1,C2] (210:1.43)--(245:1.52);
	\draw[->, >=stealth',ns1,C2] (210:1.43)--(243:1.61);
	\draw[loosely dotted,ns2,C2] (228:1.65)--(232:3.8);
	\draw[->, >=stealth',ns1,C2] (211:4.13)--(251:4.6);
	\draw[->, >=stealth',ns1,C2] (211:4.13)--(248:4.35);
	\draw[->, >=stealth',ns1,C2] (211:4.13)--(246.5:4.15);
	\draw[->, >=stealth',ns1,C2] (185:1.43)--(208:1.43);
	\draw[->, >=stealth',ns1,C2] (185:1.52)--(205:1.53);
	\draw[->, >=stealth',ns1,C2] (185:1.61)--(202:1.63);
	\draw[->, >=stealth',ns1,C2] (200:4.13)--(210:4.13);
	\draw[->, >=stealth',ns1,C2] (200:4.03)--(208:4);
	\draw[->, >=stealth',ns1,C2] (200:3.93)--(206.5:3.87);
	
	\node[minimum height=10pt,inner sep=0,font=\small] at (0:0) {$C$};
	\node[minimum height=10pt,inner sep=0,font=\small] at (248:\radius-0.2) {$c_{0}$};
	\node[minimum height=10pt,inner sep=0,font=\small] at (252:\radius+0.6) {$S_{0}$};	
	\node[minimum height=10pt,inner sep=0,font=\small] at (292:\radius+0.9) {$S_{1}$};
	\node[minimum height=10pt,inner sep=0,font=\small] at (332:\radius+1.2) {$S_{2}$};
	\node[minimum height=10pt,inner sep=0,font=\small] at (212:\radius+1.6) {$S_{\ell-1}$};
	\node[minimum height=10pt,inner sep=0,font=\small] at (254:\radius+2) {$S_{\ell}=S_0$};
	\end{tikzpicture}
	\caption{Illustration of the proof of Lemma \ref{lem:connected}.}\label{fig:illustrationOfCycleFreenessOfModelsOfFi}
\end{figure}
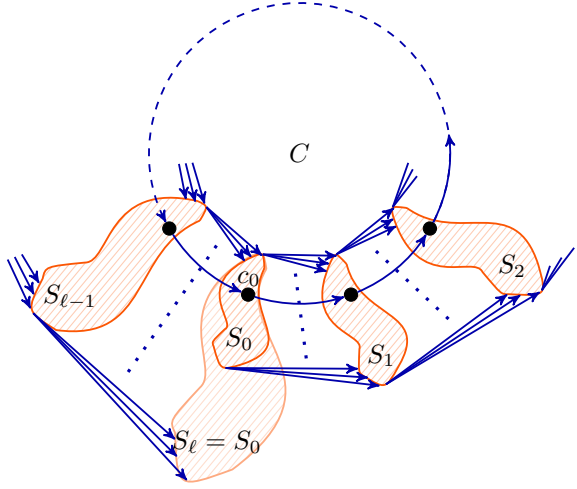
%

\begin{lemma}\label{lem:exactFormOfModels}
	Let $\mathcal{A}\in C_d$ be a (finite) model of $\varphi_{\zigzag}$. Then $|A|=\sum_{m=0}^{n}D^{4m}$ for some $n\in \mathbb{N}$, $G(\mathcal{A}|_F)$ is a  $D^4$-ary complete rooted tree, where the root is the unique element $v\in A$ for which $\mathcal{A}\models \varphi_{\operatorname{root}}(v)$, and $U(\mathcal{A}|_E)[T_m]\cong G_m$ where $G_m$ is defined in Definition~\ref{dfn:expanders} and  $T_m$ is the set of vertices of distance $m$ to $v$ in the tree $G(\mathcal{A}|_F)$ for any $m\in \{1,\dots,n\}$. Furthermore for every $n\in \mathbb{N}$ there is a model of $\varphi_{\zigzag}$ of size $\sum_{m=0}^{n}D^{4m}$.
\end{lemma}
\begin{proof}
	Lemma  \ref{lem:connected} combined with $\mathcal{A}\models \varphi_{\operatorname{tree}}$ proves that $G(\mathcal{A}|_F)$ is a rooted  tree. Let $n$ be the greatest distance of any vertex in $G(\mathcal{A}|_F)$ to the root and let $T_m$ be the vertices of distance  $m$ to the root for $m\leq n$. Then $U(\mathcal{A}|_E)[T_1]\cong G_1$ because $\mathcal{A}\models \varphi_{\operatorname{base}}$. Now assume towards an inductive proof that $U(\mathcal{A}|_E)[T_m]\cong G_m$ for some fixed $m\in \Npos$. Since $\lambda(G_m)<1$   by Lemma~\ref{lem:expansionOfSquaring} and Lemma~\ref{lem:connectedBipartiteEigenvalues} we get that $(U(\mathcal{A}|_E))^2[T_m]$  is a connected component of $(U(\mathcal{A}|_E))^2$. Hence by Lemma \ref{lem:recursion} we get that $U(\mathcal{A}|_E)[T_{m+1}]\cong G_{m+1}$. Since $G_m$ has $D^{4m}$ vertices this also proves that $\mathcal{A}$ has $\sum_{m=0}^{n}D^{4m}$ vertices.
\end{proof}

Now we are ready to finish the proof of Theorem \ref{thm:expansionOfModels}.
\begin{proof}[Proof of Theorem \ref{thm:expansionOfModels}]
	We will prove that for $\epsilon={D^2}/{12}$ the claimed is true.
	Let $\mathcal{A}$ be the model of $\varphi_{\zigzag}$ of size $\sum_{m=0}^{n}D^{4m}$ and $S\subseteq A$ with  $|S|\leq(\sum_{m=0}^{n}D^{4m})/{2}$. 
	Let $T_m$ be the vertices of distance $m$ to the root of the tree $G(\mathcal{A}|_F)$ and let $S_m:=T_m\cap S$.
	
	We can assume that $S>1$ as every vertex has degree at least $\epsilon$.
	Let us first assume that $|S_m|\leq {D^{4m}}/{2}$ for all $m\in [n]$. Then because  $G_m$ is a ${D^2}/{4}$-expander (this follows directly from Theorem \ref{thm:boundExpansionRatioInTermsOfLambda} as $\lambda(G_m)\leq {1}/{2}$ by Proposition~\ref{prop:recursiveConstruction}) and $U(\mathcal{A}|_E)[T_m]\cong G_m$ 
	we know that 
	\begin{align*}
	|\langle S, \overline{S}\rangle_{U(\mathcal{A})} |\geq \sum_{m=1}^{n}\frac{D^2}{4} |S_m|\geq \frac{D^2}{12}\sum_{m=0}^{n}|S_m|=\frac{D^2}{12}|S|.
	\end{align*} 
	Now assume the opposite and choose $m'$ to be the largest index such that 
	\begin{align}\label{eq:choiceOfm'}
		|S_{m'}|>\frac{|T_{m'}|}{2}=\frac{D^{4m'}}{2}.
	\end{align}  
	We will use the following claim.
	\begin{claim}\label{claim:relationOfSizesOfLevels}
		$ \sum_{m=0}^{\tilde{m}-1}|T_{m}|\leq \frac{1}{2}|T_{\tilde{m}}|$ for all $\tilde{m}\leq n$.
	\end{claim}
	\begin{proof}
		Inductively, we argue that $ \sum_{m=0}^{\tilde{m}-1}|T_m|= \sum_{m=0}^{\tilde{m}-2}|T_m|+|T_{\tilde{m}-1}|\leq \frac{1}{2}(3|T_{\tilde{m}-1}|) \leq \frac{1}{2}|T_{\tilde{m}}|$.
	\end{proof}
	Claim \ref{claim:relationOfSizesOfLevels} implies that $\frac{3}{4}\cdot |T_{n}|\geq \frac{1}{2}|T_{n}|+\frac{1}{2}\sum_{m=0}^{n-1}|T_m| =  \frac{1}{2}|A| \geq|S|\geq |S_{n}|$. In the case that $m'=n$, using that $G_n$ is a ${D^2}/{4}$-expander we get 
	\begin{align*}
	|\langle S,\overline{S}\rangle_{U(\mathcal{A})}|\geq \frac{D^2}{4}(|T_{n}|-|S_{n}|)\geq \frac{D^2}{16}|T_{n}|\geq \frac{D^2}{12}|S|.
	\end{align*}
	Assume now that $m'<n$.  Since $S$ is the disjoint union of all $S_m$ we know that the set $\langle S,\overline{S}\rangle_{U(\mathcal{A})}$ contains the disjoint sets $\langle S_m,T_m\setminus S_m  \rangle_{U(\mathcal{A})}$, $\langle T_{m'}\setminus S_{m'},T_{m'} \rangle_{U(\mathcal{A})}$ and  $\langle S_{m'},T_{m'+1}\setminus S_{m'+1} \rangle_{U(\mathcal{A})} $ for all $m'<m\leq n$. Since every vertex in $T_{m'}$ has $D^4$ neighbours in $T_{m'+1}$ and on the other hand every vertex in $T_{m'+1}$ has one neighbour in $T_{m'}$ we know that $|\langle S_{m'},T_{m'+1}\setminus S_{m'+1} \rangle_{U(\mathcal{A})}|=|\langle S_{m'},T_{m'+1} \rangle_{U(\mathcal{A})} |-|\langle S_{m'},S_{m'+1} \rangle_{U(\mathcal{A})} | \geq D^4|S_{m'}|-|S_{m'+1}|\geq D^4(|S_{m'}|-{D^{4m'}}/{2})$. Since additionally  $ {|T_{m'}|}/{2}\geq |T_{m'}\setminus S_{m'}|=D^{4m'}-|S_{m'}|$ and $G_m$ is an ${D^2}/{4}$-expander for every $m$ we get 
	\begin{align*}
	&|\langle S,\overline{S}\rangle_{U(\mathcal{A})}|\geq \sum_{m>m'}\frac{D^2}{4}|S_m|+\frac{D^2}{4}|T_{m'}\setminus S_{m'}|+D^4(|S_{m'}|-\frac{D^{4m'}}{2})\\
	&=\frac{D^2}{4}\sum_{m>m'}|S_m|+\Big(D^4-\frac{D^2}{2}\Big)|S_{m'}|-\Big(D^4-\frac{D^2}{2}\Big)\frac{D^{4m'}}{2}+\frac{D^2}{4}|S_{m'}|\\
	&\stackrel{\text{Equation }\ref{eq:choiceOfm'}}{\geq} \frac{D^2}{4}\sum_{m>m'}|S_m|+\frac{D^2}{8}|S_{m'}|+\frac{D^2}{8}\Big(\frac{|T_{m'}|}{2}\Big)\\
	&\stackrel{\text{Claim }\ref{claim:relationOfSizesOfLevels} }{\geq} \frac{D^2}{4}\sum_{m>m'}|S_m|+\frac{D^2}{8}|S_{m'}|+\frac{D^2}{8}\sum_{m<m'}|T_m|\stackrel{|T_{m}|\geq |S_{m}|}{\geq} \frac{D^2}{12}|S|.
	\end{align*}
\end{proof}

\section{On the non-testability of a $\Pi_2$-property}\label{sec:FOnontestability} 
In this section we that there exists an FO property on relational structures in $\Pi_2$ that is not testable. To do so, we first prove that the property $P_{\varphi_{\zigzag}}$ defined by the formula $\varphi_{\zigzag}$ in Section \ref{sec: definitionFormula} is not testable. Later we prove that $\varphi_{\zigzag}$ is in $\Pi_2$. Finally, we extends our non-testability result to simple graphs.


\paragraph{Non-testability of $P_{\varphi_{\zigzag}}$.}
Recall that $r$-types are the isomorphism class of $r$-balls and that restricted to the class $C_d$ there are finitely many $r$-types. Let $\tau_1,\dots,\tau_t$ be a list of all $r$-types of bounded degree $d$. We let $\rho_{\mathcal{A},r}$ be the $r$-type distribution of  $\mathcal{A}$, \ie   \[\rho_{\mathcal{A},r}(X)_:=\frac{\sum_{\tau\in X}|\{a\in A\mid \mathcal{N}_r^\mathcal{A}(a)\in \tau\}|}{|A|}\] for any $X\subseteq \{\tau_1,\dots,\tau_t\}$. For two $\sigma$-structures $\mathcal{A}$ and $\mathcal{B}$ we define the sampling distance of depth $r$ as $\delta_{\odot}^r(\mathcal{A},\mathcal{B}):=\sup_{X\subseteq \{\tau_1,\dots,\tau_t\}}|\rho_{\mathcal{A},r}(X)-\rho_{\mathcal{B},r}(X)|$. Note that $\delta_{\odot}^r(\mathcal{A},\mathcal{B})$ is just the total variance distance between $\rho_{\mathcal{A},r},\rho_{\mathcal{B},r}$, and it holds that  $\delta_{\odot}^r(\mathcal{A},\mathcal{B})=\frac12\sum_{i=1}^{t}|\rho_{\mathcal{A},r}(\{\tau_i\})-\rho_{\mathcal{B},r}(\{\tau_i\})|$. Then the sampling distance of $\mathcal{A}$ and $\mathcal{B}$ is defined as $\delta_\odot (\mathcal{A},\mathcal{B}):=\sum_{r=0}^{\infty}\frac{1}{2^r}\cdot\delta_\odot^r(\mathcal{A},\mathcal{B})$.

The following theorem was proven for simple graphs and easily extends to $\sigma$-structures. 

\begin{theorem}[\cite{LovaszBook2012}]\label{thm:approximatingNeighbourhoodDistributionBySmallGraph} 
	For every $\lambda>0$ there is a positive integer $n_0$ such that for every $\sigma$-structure $\mathcal{A}\in C_d$ there is a $\sigma$-structure $\mathcal{H}\in C_d$ such that $|H|\leq n_0$ and $\delta_{\odot}(\mathcal{A},\mathcal{H})\leq \lambda$.
\end{theorem}
We make use of the following definition of local properties. 

\begin{definition}[\cite{AdlerH18}]
	Let $\epsilon\in (0,1]$. A property $P\subseteq C_d$ is $\epsilon$-local on $C_d$ if there are numbers $r:=r(\epsilon)\in \mathbb{N}$, $\lambda:=\lambda(\epsilon) >0$ and $n_0:=n_0(\epsilon)\in \mathbb{N}$ such that for any $\sigma$-structure $\mathcal{A}\in P$ and $\mathcal{B}\in C_d$ both on $n\geq n_0$ vertices, if $\sum_{i=1}^{t}|\rho_{\mathcal{A},r}(\{\tau_i\})-\rho_{\mathcal{B},r}(\{\tau_i\})|<\lambda $  then $\mathcal{B}$ is $\epsilon$-close to $P$, where $\tau_1,\dots,\tau_t$ is a list of all $r$-types of bounded degree $d$.

	The property $P$ is local on $C_d$ if it is $\epsilon$-local on $C_d$ for every $\epsilon\in (0,1]$. 
\end{definition}
The following theorem relating testable properties and local properties was proven in \cite{AdlerH18}
\begin{theorem}[\cite{AdlerH18}]\label{thm:Locality} 
	For every property $P\in C_d$, $P$ is testable if and only if $P$ is local on $C_d$.
\end{theorem}
We let $P_{\zigzag}:=P_{\varphi_{\zigzag}}$ for the formula $\varphi_{\zigzag}$ from Section \ref{sec: definitionFormula}. We also let $\sigma$ and $d$ be as defined in Section~$\ref{sec: definitionFormula}$.

\begin{theorem}\label{thm:nonTestabilityForStructures} 
	$P_{\zigzag}$ is not testable on $C_d$.
\end{theorem}
\begin{proof}
	
	We prove non-locality for $P_{\zigzag}$ and  get non-testability with Theorem \ref{thm:Locality}. 
	Let $\epsilon:={1}/(144D^2)$ and let $r\in \mathbb{N}$, $\lambda >0$ and $n_0\in \mathbb{N}$ be arbitrary. We set $\lambda':={\lambda}/( t 2^{r+1})$, where $\tau_1,\dots,\tau_t$ are all $r$-types of bounded degree $d$, and let $n_0'$ be the positive integer from Theorem \ref{thm:approximatingNeighbourhoodDistributionBySmallGraph} corresponding to   $\lambda'$. We now pick $n\in \mathbb{N}$ such that $n=\sum_{i=0}^{k}D^{4i}$ for some $k\in \mathbb{N}$, $n\geq 4n_0$ and $n\geq 4({n_0'}/{\lambda})$. Let $\mathcal{A}\in C_d$ be a model of $\varphi_{\zigzag}$ on $n$ vertices. By Theorem \ref{thm:approximatingNeighbourhoodDistributionBySmallGraph} there is a structure $\mathcal{H}\in C_d$ on $m\leq n_0'$ vertices such that $\delta_{\odot}(\mathcal{A},\mathcal{H})\leq \lambda$. Let $\mathcal{B}$ be the structure consisting of $\lfloor{n}/{m}\rfloor$ copies of $\mathcal{H}$ and $n\mod m$ isolated vertices. Note that we picked $\mathcal{B}$ such that $|A|=|B|$.

	We will first argue that $\mathcal{B}$ is in fact $\epsilon$-far from having the property $P_{\zigzag}$. 
	First we rename the elements from $B$ in such a way that $A=B$ and the number $\sum_{\tilde{R}\in \sigma}|\tilde{R}^\mathcal{A}\Delta \tilde{R}^\mathcal{B}|$ of edge modifications to turn $\mathcal{A}$ and $\mathcal{B}$ into the same structure is minimal. Pick a partition $A=B=S\sqcup S'$  in such a way  that $S\times S'\cap \tilde{R}^\mathcal{B}=\emptyset$, $S'\times S\cap \tilde{R}^\mathcal{B}=\emptyset$ for any $\tilde{R}\in \sigma$ and $||S|-|S'||$  minimal among all such partitions. Assume that $|S|\leq |S'|$. Since the connected components in $\mathcal{B}$ are of size $\leq m$ we know that $||S|-|S'||\leq m$ because otherwise we can get a partition $B=T\sqcup T'$ with $||T|-|T'||<||S|-|S'||$ by picking a connected component of $\mathcal{B}$ whose elements are contained in $S'$ and moving them from $S'$  to $S$. Since $|S|\leq |S'|$ and $m\leq {n}/{4}$ we know that $ {n}/{4}\leq |S|\leq {n}/{2}$.
	This implies that
	\begin{align*}
	\sum_{\tilde{R}\in \sigma}|\tilde{R}^\mathcal{A}\Delta \tilde{R}^\mathcal{B}|&\geq|\langle S,S'\rangle_{U(\mathcal{A})}|\stackrel{\text{Def }\ref{def:expansionRatio}}{\geq}|S|\cdot  h(\mathcal{A})\\&\stackrel{\text{Thm }\ref{thm:expansionOfModels}}{\geq}\frac{n}{4}\cdot\frac{D^2}{12}=\frac{1}{48} D^2n\geq \frac{1}{144 D^2}dn.
	\end{align*}
	Therefore $\mathcal{B}$ is $\epsilon$-far from 
	being in $P_{\zigzag}$.
	
	But the neighbourhood distributions of $\mathcal{A}$ and $\mathcal{B}$ are similar as the following shows, proving that $P_{\zigzag}$ is not local.

	\begin{align*}
	&\sum_{i=1}^{t}|\rho_{\mathcal{A},r}(\{\tau_i\})-\rho_{\mathcal{B},r}(\{\tau_i\})|\\
	&=\sum_{i=1}^{t}\Big|\rho_{\mathcal{A},r}(\{\tau_i\})
	-\frac{n\mod m}{n}\cdot \rho_{K_1,r}(\{\tau_i\})-\Big\lfloor\frac{n}{m}\Big\rfloor \cdot\frac{m}{n}\cdot\rho_{\mathcal{H},r}(\{\tau_i\})\Big|
	%
	\\&\leq\sum_{i=1}^{t}\Big|\rho_{\mathcal{A},r}(\{\tau_i\})-\rho_{\mathcal{H},r}(\{\tau_i\})\Big|+\sum_{i=1}^{t}\Big|\frac{n\mod m}{n}\cdot \rho_{K_1,r}(\{\tau_i\})\Big|\\&+\sum_{i=1}^{t}\Big|\rho_{\mathcal{H},r}(\{\tau_i\})-\Big\lfloor\frac{n}{m}\Big\rfloor \cdot\frac{m}{n}\cdot\rho_{\mathcal{H},r}(\{\tau_i\})\Big|\\
	&\leq\sum_{i=1}^{t}\Big|\rho_{\mathcal{A},r}(\{\tau_i\})-\rho_{\mathcal{H},r}(\{\tau_i\})\Big|+\frac{2m}{n}\\
	&\leq t\cdot\sup_{X\subseteq \mathcal{B}_r}\abs{\rho_{\mathcal{A},r}(X)-\rho_{\mathcal{H},r}(X)}+\frac{2m}{n}\\
	&\leq t\cdot 2^r\cdot \delta_{\odot}(\mathcal{A},\mathcal{H})+\frac{2m}{n}\leq \frac{\lambda}{2}+\frac{\lambda}{2}=\lambda.
	\end{align*}
	The last inequality holds by choice of $\lambda'$ and Theorem \ref{thm:approximatingNeighbourhoodDistributionBySmallGraph}.
\end{proof}

\paragraph{Every FO property on degree-regular structures is in $\Pi_2$.}
We first give the following definition. 
\begin{definition}\label{ex:delta2}
	A Hanf sentence $\exists ^{\geq m} x\, \phi_{\tau}(x)$
	is short for
	\[\exists x_1\ldots x_m  \big(\bigwedge_{1\leq i,j\leq m, i\neq j} x_i\neq x_j\wedge\bigwedge_{1\leq i\leq m} \phi_{\tau}(x_i)\big),\]
	and $\phi_{\tau}(x_i)$ can be expressed by an $\exists^*\forall$-formula where the existential
	quantifiers ensure the existence of the desired $r$-neighbourhood with all tuples in relations / not in relations as required by $\tau$, 
	and the universal quantifier is used to express that there are no other elements in
	the $r$-neighbourhood of $x_i$. 
\end{definition}
Note that by definition, any Hanf sentence is in~$\Sigma_2$. We now show the following lemma. 

\begin{lemma}\label{lem:d-regHNF}	
	Let $d\in \mathbb N$ and let $\phi$ be an FO sentence.
	If every model of $\varphi$ is $d$-regular, then $\varphi$ is $d$-equivalent to a $\Pi_2$ sentence.
\end{lemma}

\begin{proof}
	Before we begin, let us define an $r$-type 
	$\tau$ to be \emph{$d$-regular}, if for all structures $\mathcal A$ and all elements 
	$a\in A$ of $r$-type $\tau$, every $b\in A$ with 
	$\dist(a,b)<r$ has $\deg_{\mathcal A}(b)=d$.
	
	We first prove the following claim.	
	\begin{claim}\label{claim:Pi2}
		Let $d\in \mathbb N$, let $\phi$ be an FO sentence, and let $\psi$ be in HNF with 
		$\psi\equiv_d\phi$ such that $\psi$ is in DNF, where the literals
		are Hanf sentences or negated Hanf sentences. Furthermore, assume that the neighbourhood types in all (positive) Hanf sentences of $\psi$ are $d$-regular. Then $\phi$ is $d$-equivalent to a sentence in $\Pi_2$.
	\end{claim}
	
	\begin{proof}
		Assume $\psi$ is of the form $\exists^{\geq m} x\, \phi_{\tau}(x)$, where $\tau$ is
		$d$-regular.
		As in Definition~\ref{ex:delta2}, we may assume $\phi_{\tau}(x_i)$ is an 
		$\exists^*\forall$-formula, which is a conjunction of an
		$\exists^*$-formula $\phi'_{\tau}(x_i)$ (expressing that
		$x$ has an `induced sub-neighbourhood' of type $\tau$) and a 
		universal formula saying that
		there are no further elements in the neighbourhood.
		We now have that $\psi\equiv_d\exists^{\geq m} x \,\phi'_{\tau}(x)$. To 
		see this, let
		$\mathcal A\models \exists^{\geq m} x \phi'_{\tau}(x)$ and 
		$\deg(\mathcal A)\leq d$. Then 
		$\mathcal A\models \exists^{\geq m} x \phi_{\tau}(x)$ because $\tau$ is 
		$d$-regular. The converse is obvious.
		
		If $\psi$ is of the form $\neg \exists^{\geq m} x\, \phi_{\tau}(x)$, where
		$\phi_{\tau}(x_i)$ is an	$\exists^*\forall$-formula, then 
		$\neg \exists^{\geq m} x\, \phi_{\tau}(x)$ is equivalent to a formula in $\Pi_2$.
		Since $\Pi_2$ is closed under disjunction and conjunction, this proves the claim.
	\end{proof}
	Now the proof follows from Claim~\ref{claim:Pi2}, because if $\phi$ only has $d$-regular models, then by Hanf's Theorem there is a formula $\psi\equiv \phi$ satisfying the assumptions of the claim.
\end{proof}

\paragraph{Existence of a non-testable $\Pi_2$-property.}
With Lemma \ref{lem:d-regHNF} and Theorem \ref{thm:nonTestabilityForStructures}, we are ready to prove the following theorem.

\begin{theorem}\label{thm:pi2}
	There are degree bounds $d\in \mathbb{N}$ such that there exists a property on $C_{d}$ definable by a formula in $\Pi_2$ that is not testable.
\end{theorem}

%
\begin{proof}
	Pick $d=2D^2+D^4+1$ for any large prime power $D$. Then using the construction from \cite{ZigZagProductIntroduction} we can find a $(D^4,D,1/4)$-graph $H$. By Theorem~\ref{thm:nonTestabilityForStructures}, using this base expander $H$ for the construction of the formula $\varphi_{\zigzag}$ we get a property which is not testable on $C_d$. Since all models of $\varphi_{\zigzag}$ are $d$-regular by construction, Lemma \ref{lem:d-regHNF} gives us that $\varphi_{\zigzag}$ is $d$-equivalent to a formula in $\Pi_2$.
\end{proof}

\subsection{Extension to simple (undirected) graphs}
By our previous argument, to show the existence of a non-testable $\Pi_2$-property for simple graphs, \ie undirected graphs without parallel edges and without self-loops, it suffices to construct a non-testable FO graph property of degree regular graphs. 
To do so, we carefully translate the edge-coloured directed graphs of our previous 
example in Section \ref{sec: definitionFormula} to simple graphs. 
We encode  
$\sigma$-structures by representing each type of directed edge by a constant size graph gadget, maintaining the degree regularity. We then translate the formula $\varphi_{\zigzag}$  into a formula $\psi_{\zigzag}$. We obtain a class of simple expanders, that is defined by an FO sentence, and obtain the analogous Theorem.
\begin{theorem}\label{thm:simpleDelta2}
	There exists $d\in \mathbb{N}$ and an FO property of simple graphs of bounded degree $d$ that is not testable. 
\end{theorem}
In the rest of this section, we prove the above theorem. 

\paragraph{Construction of a family of graphs.} Let $d$ be as defined in Section \ref{sec: definitionFormula}. 
Let $G^d(u,v)$ be the graph with vertex set $\{u,v,u_0,\dots,u_{d-2}\}$ and edge set $\{\{w,u_i\},\{v,u_i\},\{u_i,u_j\}\mid i,j\in[d-2],i\not=j\}$. Let $H^d(u,v)$
be the graph with vertex set $\big\{u,v,u_i,u_j',v_i,v_j'\mid i\in \big[\big\lfloor \frac{d-1}{2}\big\rfloor\big],j\in \big[\big\lceil \frac{d-1}{2}\big\rceil\big]\big\}$ and edge set 
\begin{align*}
&\Big\{\{u,u_i\},\{v,v_i\},\{u_i,v_i\}\mid i\in \Big[\Big\lfloor \frac{d-1}{2}\Big\rfloor\Big]\Big\}\cup\\
&\Big\{\{u,u_j'\},\{v,v_j'\},\{u_j',v_j'\}\}\mid j\in \Big[\Big\lceil \frac{d-1}{2}\Big\rceil\Big]\Big\}\cup\\
&\Big\{\{u_i,u_k\},\{v_i,v_k\}\mid i, k\in \Big[\Big\lfloor \frac{d-1}{2}\Big\rfloor\Big],i\not= k\Big\}\cup\\
&\Big\{\{u'_j,u'_k\},\{v'_j,v'_k\}\mid j,k \in \Big[\Big\lceil \frac{d-1}{2}\Big\rceil \Big],j\not=k\Big\}\cup\\
&\Big\{\{u_i,v'_j\},\{u'_j,v_i\}\mid i\in \Big[\Big\lfloor \frac{d-1}{2}\Big\rfloor\Big],j \in \Big[\Big\lceil \frac{d-1}{2}\Big\rceil\Big]\}
\end{align*}
Finally, for every $\ell\in \mathbb{N}$ and $0\leq p\leq \ell$, let $P^d_{\ell,p}(u_0,v_\ell)$ be the graph consisting of $\ell$ copies $G^d(u_0,v_0), \dots, G^d(u_{p-1},v_{p-1}), G^d(u_{p+1},v_{p+1})$, $\dots, G^d(u_\ell,v_\ell)$, one copy $H^d(u_p,v_p)$ and additional edges $\{v_i,u_{i+1}\}$ for each $i\in [\ell]$. Note that $P^d_{\ell,p}(u_0,v_\ell)$ has $\ell\cdot (d+1)+2d$ vertices, the vertices $u_0$ and $v_\ell$ have degree $d-1$ and every other vertex has degree $d$, see Figure \ref{fig:path} for an example.  \newcount\mycount
\begin{figure*}
	\centering
	\scalebox{.7}{
		\begin{tikzpicture}
		\tikzstyle{ns1}=[line width=0.7]
		\tikzstyle{ns2}=[line width=1.2]
		\def \radius {1}
		\def \rad {2}
		\def \shiftOne {(5,0)}
		\def \shiftTwo {(10,0)}
		\def \shiftThree {(10,2)}
		\def \shiftFour {(11,2)}
		\def \shiftFive {(10,-2)}
		\def \shiftSix {(11,-2)}
		\def \shiftSeven {(11,0)}
		\def \shiftEight {(16,0)}
		\def \margin {2.7}
		\node[draw,circle,fill=black,inner sep=0pt, minimum width=5pt] (0) at (180:\rad) {};
		\node[draw,circle,fill=black,inner sep=0pt, minimum width=5pt] (1) at (0:\radius) {};
		\node[draw,circle,fill=black,inner sep=0pt, minimum width=5pt] (2) at (72:\radius) {};
		\node[draw,circle,fill=black,inner sep=0pt, minimum width=5pt] (3) at (144:\radius) {};
		\node[draw,circle,fill=black,inner sep=0pt, minimum width=5pt] (4) at (216:\radius) {};
		\node[draw,circle,fill=black,inner sep=0pt, minimum width=5pt] (5) at (288:\radius) {};
		\node[draw,circle,fill=black,inner sep=0pt, minimum width=5pt] (6) at (0:\rad) {};
		\node[draw,circle,fill=black,inner sep=0pt, minimum width=5pt] [shift={\shiftOne}](7) at (180:\rad) {};
		\node[draw,circle,fill=black,inner sep=0pt, minimum width=5pt] [shift={\shiftOne}](8) at (0:\radius) {};
		\node[draw,circle,fill=black,inner sep=0pt, minimum width=5pt] [shift={\shiftOne}](9) at (72:\radius) {};
		\node[draw,circle,fill=black,inner sep=0pt, minimum width=5pt] [shift={\shiftOne}](10) at (144:\radius) {};
		\node[draw,circle,fill=black,inner sep=0pt, minimum width=5pt] [shift={\shiftOne}](11) at (216:\radius) {};
		\node[draw,circle,fill=black,inner sep=0pt, minimum width=5pt] [shift={\shiftOne}](12) at (288:\radius) {};
		\node[draw,circle,fill=black,inner sep=0pt, minimum width=5pt] [shift={\shiftOne}](13) at (0:\rad) {};
		\node[draw,circle,fill=black,inner sep=0pt, minimum width=5pt] [shift={\shiftTwo}](14) at (180:\rad) {};
		\node[draw,circle,fill=black,inner sep=0pt, minimum width=5pt] [shift={\shiftThree}](15) at (216:\radius) {};
		\node[draw,circle,fill=black,inner sep=0pt, minimum width=5pt] [shift={\shiftThree}](16) at (144:\radius) {};
		\node[draw,circle,fill=black,inner sep=0pt, minimum width=5pt] [shift={\shiftFour}](17) at (0:\radius) {};
		\node[draw,circle,fill=black,inner sep=0pt, minimum width=5pt] [shift={\shiftFour}](18) at (72:\radius) {};
		\node[draw,circle,fill=black,inner sep=0pt, minimum width=5pt] [shift={\shiftFour}](19) at (288:\radius) {};
		\node[draw,circle,fill=black,inner sep=0pt, minimum width=5pt] [shift={\shiftFive}](20) at (180:\radius) {};
		\node[draw,circle,fill=black,inner sep=0pt, minimum width=5pt] [shift={\shiftFive}](21) at (252:\radius) {};
		\node[draw,circle,fill=black,inner sep=0pt, minimum width=5pt] [shift={\shiftFive}](22) at (108:\radius) {};
		\node[draw,circle,fill=black,inner sep=0pt, minimum width=5pt] [shift={\shiftSix}](23) at (324:\radius) {};
		\node[draw,circle,fill=black,inner sep=0pt, minimum width=5pt] [shift={\shiftSix}](24) at (396:\radius) {};
		\node[draw,circle,fill=black,inner sep=0pt, minimum width=5pt] [shift={\shiftSeven}](25) at (0:\rad) {};
		\node[draw,circle,fill=black,inner sep=0pt, minimum width=5pt] [shift={\shiftEight}](26) at (180:\rad) {};
		\node[draw,circle,fill=black,inner sep=0pt, minimum width=5pt] [shift={\shiftEight}](27) at (0:\radius) {};
		\node[draw,circle,fill=black,inner sep=0pt, minimum width=5pt] [shift={\shiftEight}](28) at (72:\radius) {};
		\node[draw,circle,fill=black,inner sep=0pt, minimum width=5pt] [shift={\shiftEight}](29) at (144:\radius) {};
		\node[draw,circle,fill=black,inner sep=0pt, minimum width=5pt] [shift={\shiftEight}](30) at (216:\radius) {};
		\node[draw,circle,fill=black,inner sep=0pt, minimum width=5pt] [shift={\shiftEight}](31) at (288:\radius) {};
		\node[draw,circle,fill=black,inner sep=0pt, minimum width=5pt] [shift={\shiftEight}](32) at (0:\rad) {};
		
		\draw[ns1,C2] (0)--(1);
		\draw[ns1,C2] (0)..controls(-1,1)..(2);
		\draw[ns1,C2] (0)--(3);
		\draw[ns1,C2] (0)--(4);
		\draw[ns1,C2] (0)..controls(-1,-1)..(5);
		\foreach \x in {1,...,5} {
			\foreach \y in {1,...,5}{
				\draw[ns1,C2] (\x)--(\y);
		}}
		\draw[ns1,C2] (6)--(1);
		\draw[ns1,C2] (6)--(2);
		\draw[ns1,C2] (6)..controls(0.7,1.7)..(3);
		\draw[ns1,C2] (6)..controls(0.7,-1.7)..(4);
		\draw[ns1,C2] (6)--(5);
		\draw[ns1,C2] (6)--(7);
		\draw[ns1,C2] (7)--(8);
		\draw[ns1,C2] (7)..controls(4,1)..(9);
		\draw[ns1,C2] (7)--(10);
		\draw[ns1,C2] (7)--(11);
		\draw[ns1,C2] (7)..controls(4,-1)..(12);
		\foreach \x in {8,...,12} {
			\foreach \y in {8,...,12}{
				\draw[ns1,C2] (\x)--(\y);
		}}
		\draw[ns1,C2] (13)--(8);
		\draw[ns1,C2] (13)--(9);
		\draw[ns1,C2] (13)..controls(5.7,1.7)..(10);
		\draw[ns1,C2] (13)..controls(5.7,-1.7)..(11);
		\draw[ns1,C2] (13)--(12);
		\draw[ns1,C2] (13)--(14);
		\draw[ns1,C2] (14)--(15);
		\draw[ns1,C2] (14)--(16);
		\draw[ns1,C2] (14)--(20);
		\draw[ns1,C2] (14)--(22);
		\draw[ns1,C2] (14)..controls(8.3,-2.5)..(21);
		\foreach \x in {15,...,19} {
			\foreach \y in {15,...,19}{
				\draw[ns1,C2] (\x)--(\y);
		}}
		\draw[ns1,C2] (15)--(23);
		\draw[ns1,C2] (16)--(24);
		\draw[ns1,C2] (18)--(22);
		\draw[ns1,C2] (19)--(21);
		\draw[ns1,C2] (17)..controls(12,0)and(9,0)..(20);
		\foreach \x in {20,...,24} {
			\foreach \y in {20,...,24}{
				\draw[ns1,C2] (\x)--(\y);
		}}
		\draw[ns1,C2] (25)--(23);
		\draw[ns1,C2] (25)--(24);
		\draw[ns1,C2] (25)--(17);
		\draw[ns1,C2] (25)--(19);
		\draw[ns1,C2] (25)..controls(12.7,2.5)..(18);
		\draw[ns1,C2] (25)--(26);
		\draw[ns1,C2] (26)--(27);
		\draw[ns1,C2] (26)..controls(15,1)..(28);
		\draw[ns1,C2] (26)--(29);
		\draw[ns1,C2] (26)--(30);
		\draw[ns1,C2] (26)..controls(15,-1)..(31);
		\foreach \x in {27,...,31} {
			\foreach \y in {27,...,31}{
				\draw[ns1,C2] (\x)--(\y);
		}}
		\draw[ns1,C2] (32)--(27);
		\draw[ns1,C2] (32)--(28);
		\draw[ns1,C2] (32)..controls(16.7,1.7)..(29);
		\draw[ns1,C2] (32)..controls(16.7,-1.7)..(30);
		\draw[ns1,C2] (32)--(31);

		\node[minimum height=10pt,inner sep=0,font=\small] at (-2.1,0.3) {$u_{0}$};
		
		\node[minimum height=10pt,inner sep=0,font=\small] at (18.1,0.3) {$v_{3}$};
		\end{tikzpicture} }
	\caption{Illustration of $P^6_{3,2}(u_0,v_3)$.}\label{fig:path}
\end{figure*}
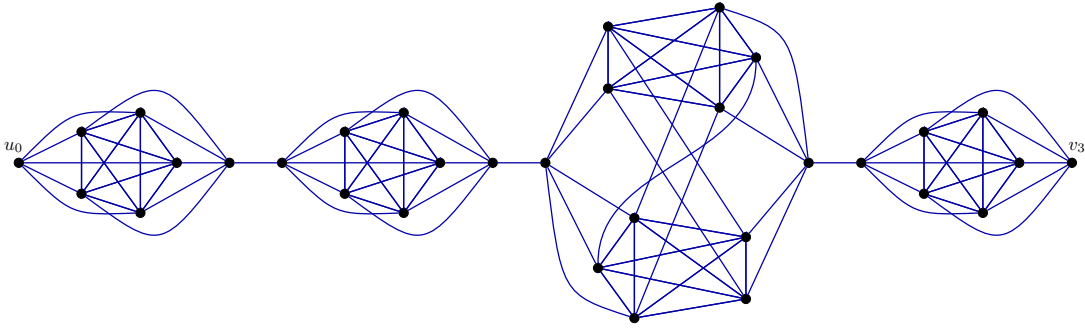

Let $\mathcal{A}\in P_{\zigzag}$ and let $\ell=2\cdot(3 D^4+1)$. We obtain an undirected graph $G=(V,E)$ from $\mathcal{A}$ using the following steps. 
\begin{enumerate}
	\item For every $i_0,i_1,i_2,i_3\in [D]$ we define $p=\sum_{k=0}^{3} i_k\cdot D^k$ and  replace every edge $(x,y)\in E_{(i_0,i_1),(i_2,i_3)}^\mathcal{A}$  by $P^d_{\ell,p}(u_0,v_\ell)$ and additional edges $\{x,u_0\}$ and $\{v_\ell,y\}$. 	Here all vertices of $P^d_{\ell,p}(u_0,v_\ell)$ are pairwise distinct and new, and we call them \emph{auxiliary vertices}. 
	Call this gadget graph an \emph{$E_{(i_0,i_1),(i_2,i_3)}$-arrow with end-vertices $x$ and $y$}.
	\label{Eedge}
	\item For every $i_0,i_1,i_2,i_3\in [D]$ we define $p=D^4+\sum_{k=0}^{3} i_k\cdot D^k$ and  replace every edge $(x,y)\in F_{((i_0,i_1),(i_2,i_3))}^\mathcal{A}$  by $P^d_{\ell,p}(u_0,v_\ell)$ and additional edges $\{x,u_0\}$ and $\{v_\ell,y\}$. 	Here all vertices of $P^d_{\ell,p}(u_0,v_\ell)$ are pairwise distinct and new, and we call them \emph{auxiliary vertices}. 
	Call this gadget graph an \emph{$F_{((i_0,i_1),(i_2,i_3))}$-arrow with end-vertices $x$ and $y$}.
	\label{Fedge}
	\item For every $i_0,i_1,i_2,i_3\in [D]$ we define $p=2D^4+\sum_{k=0}^{3} i_k\cdot D^k$ and  replace every edge $(x,y)\in L_{((i_0,i_1),(i_2,i_3))}^\mathcal{A}$  by $P^d_{\ell,p}(u_0,v_\ell)$ and additional edges $\{x,u_0\}$ and $\{v_\ell,y\}$. 	Here all vertices of $P^d_{\ell,p}(u_0,v_\ell)$ are pairwise distinct and new, and we call them \emph{auxiliary vertices}. 
	Call this gadget graph an \emph{$L_{((i_0,i_1),(i_2,i_3))}$-arrow with end-vertices $x$ and $y$}.
	\label{Ledge}
	\item We define $p=3D^4$ and  replace every edge $(x,y)\in R^\mathcal{A}$  by $P^d_{\ell,p}(u_0,v_\ell)$ and additional edges $\{x,u_0\}$ and $\{v_\ell,y\}$. 	Here all vertices of $P^d_{\ell,p}(u_0,v_\ell)$ are pairwise distinct and new, and we call them \emph{auxiliary vertices}. 
	Call this gadget graph an \emph{$R$-arrow with end-vertices $x$ and $y$}.
	\label{Redge}
	
\end{enumerate}
All vertices, that are not \emph{auxiliary}, are called \emph{original vertices}. Note that from the location $p$ of the gadget $H^d(v_0,v_\ell)$ uniquely encodes the colour of the original directed coloured edge. Also note that each arrow defined above has a direction as the gadget $H^d(v_0,v_\ell)$ is always located in the first half of the path $P^d_{\ell,p}(u_0,v_\ell)$.
The following is easy to observe from the construction.
\begin{fact}\label{rem:originalverts}
	For every $x\in V$, 
	$x$ is an original vertex iff $x$ is contained in no triangle.
\end{fact}
We let $\delta(x)$ be a formula in the language of 
undirected graphs, saying `$x$ is an original vertex', which is easy to do by Fact~\ref{rem:originalverts}. We further let $\beta(x)$ be a formula saying `$x$ is an internal vertex of either an $E_{i,j}$-arrow, or an $F_k$-arrow, or an $L_k$-arrow, or an $R$-arrow for any $i,j\in \indexSetRotation$, $k\in \indexSetH$'. Here an `internal vertex' of an arrow refers to any vertex on this arrow except the two endpoints.
We now translate the formula $\varphi_{\zigzag}$ into a formula $\psi_{\zigzag}$
in the language of undirected graphs using the following first-order formulas $\alpha^E_{i,j}$, $\alpha^F_{k}$, $\alpha^L_{k}$ and $\alpha^R$. 
Let $\alpha^E_{i,j}(x,y)$ 
say `$x$ and $y$ are the end-vertices of an induced $E_{i,j}$-arrow' for $i,j\in \indexSetRotation$, similarly, let $\alpha^F_{k}(x,y)$ 
say `$x$ and $y$ are the end-vertices of an induced $F_k$-arrow' for $k\in \indexSetH$. Furthermore let $\alpha^L_{k}(x,y)$ 
say `$x$ and $y$ are the end-vertices of an induced $L_k$-arrow' for $k\in \indexSetH$ and $\alpha^R(x,y)$ 
say `$x$ and $y$ are the end-vertices of an induced $R$-arrow' . Given $\varphi_{\zigzag}$, formula $\psi_{\zigzag}$ is obtained as follows. In $\varphi_{\zigzag}$ we replace each
expression $E_{i,j}(x,y)$ by $\alpha^E_{i,j}(x,y)$, each $F_k(x,y)$ by $\alpha^F_k(x,y)$, each $L_k(x,y)$ by $\alpha^L_k(x,y)$ and each $R(x,y)$ by $\alpha^R(x,y)$. In addition, we relativise all quantifiers to the
original vertices (replacing every expression of the form $\exists x \,\chi$ by $\exists x\,(\delta(x)\wedge \chi)$ and every expression of the form $\forall x\, \chi$ by $\forall x\,(\delta(x)\rightarrow \chi)$). Let us call the resulting formula $\psi_{\zigzag}'$. Then we set $\psi_{\zigzag}$ to be the conjunction of the formula $\psi_{\zigzag}'$ and the formula $\forall x (\lnot \delta(x)\rightarrow \beta(x))$.
Let $\mathcal{P}_\psi:=\{G\in C_{d}\mid G\models \psi_{\zigzag}\}$. 
In the following, we show that $\mathcal{P}_{\psi} $ is a family of expanders, which directly implies non-testability of $\mathcal{P}_{\psi}$ (by Lemma \ref{lem:FPS19} below). We remark that one could also prove the non-testability of $\mathcal{P}_\psi$ by showing that the aforementioned transformation (from $\sigma$-structures to simple graphs) is a (local) reduction that preserves the testability of properties. 
\begin{lemma}\label{lemma:undirected_expander}
	The models of $\psi_{\zigzag}$ is a family of $\xi$-expanders, for some constant $\xi>0$.
\end{lemma}
\begin{proof}
	Let $G=(V,E)$ be a model of $\psi_{\zigzag}$ and let $\mathcal{A}$ be the corresponding model of $\varphi_{\zigzag}$. Let $S\subset V$ such that $|S|\leq \frac{|V|}{2}$. Let $V_{\operatorname{original}}\sqcup V_{\operatorname{auxiliary}}=V$ be the partition of $V$ into original and auxiliary vertices. Let $S_{\operatorname{original}}:=V_{\operatorname{original}}\cap S$ and $ S_{\operatorname{auxiliary}}:=V_{\operatorname{auxiliary}}\cap S$.
	
	First note that by the above definitions every directed coloured edge in $\mathcal{A}$ corresponds to a constant number $c_D:=2\cdot(3 D^4+1)\cdot ((d+1)+2d)$ of  auxiliary vertices in $V_{\operatorname{auxiliary}}$, where $d=2D^2+D^4+1$. 
	
	Assume $|S_{\operatorname{original}}|> \frac{2}{dc_D} \cdot|S|$. 
	Then there are  $|S|-|S_{\operatorname{original}}|<\frac{dc_D-2}{2}\cdot|S_{\operatorname{original}}|$ vertices in $S_{\operatorname{auxiliary}}$. Hence at least $\frac{d}{2}\cdot|S_{\operatorname{original}}|-\frac{dc_D-2}{2c_D}\cdot|S_{\operatorname{original}}|$ of the arrows have at least one vertex that is not in $S$ and  therefore 
	\begin{align*}
	\langle S,V\setminus S \rangle_G &\geq  \frac{d}{2}\cdot|S_{\operatorname{original}}|-\frac{dc_D-2}{2c_D}\cdot|S_{\operatorname{original}}|\\&=\frac{1}{c_D}\cdot |S_{\operatorname{original}}|\geq \frac{2}{dc_D^2}\cdot|S|.
	\end{align*}
	
	Assume $\frac{1}{2dc_D}|S|< |S_{\operatorname{original}}| \leq  \frac{2}{dc_D}\cdot|S|$. Let $\epsilon= \frac{D^2}{12}$ as defined in the proof of Theorem \ref{thm:expansionOfModels}. 
	Since each edge in the underlying graph $U(\mathcal{A})$ corresponds to exactly one arrow in $G$ we get that $\langle S,V\setminus S \rangle_G\geq \langle S_{\operatorname{original}},V_{\operatorname{original}}\setminus S_{\operatorname{original}} \rangle_{U(\mathcal{A})}$.  Since $\mathcal{A}$ is $d$-regular and every edge gets replaced by $c_D$ auxiliary vertices we get $|V|=(1+\frac{dc_D}{2})|A|$. Then
	\begin{displaymath}
	|S_{\operatorname{original}}|\leq \frac{2}{dc_D}\cdot |S|\leq \frac{1}{dc_D}\cdot |V|= \frac{2+dc_D}{2dc_D}|A|
	\end{displaymath}  
	and $|A\setminus S_{\operatorname{original}}|\geq (\frac{2dc_D}{2+dc_D}-1)|S_{\operatorname{original}}|$.
	Then from Theorem $\ref{thm:expansionOfModels}$ we directly get 
	\begin{align*}
	\langle S,V\setminus S \rangle_G&\geq \langle S_{\operatorname{original}},V_{\operatorname{original}}\setminus S_{\operatorname{original}} \rangle_{U(\mathcal{A})}\\&=\epsilon \min \{|S_{\operatorname{original}}|,|A\setminus S_{\operatorname{original}}|\}\\
	&\geq \epsilon\cdot \frac{1}{2dc_D}\cdot\frac{dc_D-2}{2+dc_D}\cdot|S|.
	\end{align*}
	
	Now assume $|S_{\operatorname{original}}|\leq  \frac{1}{2dc_D}\cdot |S|$. Therefore there are $|S|-|S_{\operatorname{original}}|\geq |S|-\frac{1}{2dc_D}\cdot |S|$ in $S_{\operatorname{auxiliary}}$. Of these at least $\frac{2dc_D-1}{2dc_D}\cdot |S|-|S_{\operatorname{original}}|c_D\geq \frac{2dc_D-1-c_D}{2dc_D}|S|$ vertices in $S_{\operatorname{auxiliary}}$ that are not in a connected component with any element from $S_{\operatorname{original}}$ in the graph $G[S]$. Since any connected component of $G[S]$ with no vertices in $S_{\operatorname{original}}$ contains at most $c_d$ vertices, we get that 
	\begin{align*}
	\langle S,V\setminus S \rangle_G \geq \frac{2dc_D-c_D-1}{2dc_D^2}|S|.
	\end{align*}
	By setting $\xi=\min\{\frac{2dc_D-c_D-1}{2dc_D^2},\epsilon \cdot\frac{1}{2dc_D}\cdot\frac{dc_D-2}{2+dc_D},\frac{2}{dc_D^2}\}>0$ 
	we proved the claimed.
\end{proof}
One can then prove that the property $P_{\psi_{\zigzag}}$ is not testable by using analogous arguments as in the proof of Theorem~\ref{thm:nonTestabilityForStructures}. In the following, we present a slightly different proof
using a result from \cite{fichtenberger2019every}. 
We first introduce a definition of ``hyperfinite graphs''.
\begin{definition}\label{def:hyperfinite}
	Let $\varepsilon\in (0,1]$ and $k\geq 1$. A graph $G$ with maximum degree bounded by $d$ is called \emph{$(\varepsilon,k)$-hyperfinite} if one can remove at most $\varepsilon d |V|$ edges from $G$ so that each connected component of the resulting graph has at most $k$ vertices. For a function $\rho: (0,1]\rightarrow\mathbb{N}^+$, a graph $G$ is called \emph{$\rho$-hyperfinite} if $G$ is $(\varepsilon,\rho(\varepsilon))$-hyperfinite for every $\varepsilon>0$. A set (or property) $\Pi$ of graphs is called \emph{$\rho$-hyperfinite} if every graph in $\Pi$ is $\rho$-hyperfinite. A set (or property) $\Pi$ of graphs is called \emph{hyperfinite} if it is $\rho$-hyperfinite for some function $\rho$. 
\end{definition}

Now we recall that a graph property is a set of graphs that is invariant under graph isomorphism. A subproperty of a property $P$ is a subset of graphs in $P$ that is also invariant under graph isomorphism.
\begin{lemma}[Corollary 1.1 in \cite{fichtenberger2019every}]\label{lem:FPS19}
	Let $C_d$ be the class of graphs of bounded maximum degree $d$. Let $P\subseteq C_d$ be a property that does not contain an infinite hyperfinite subproperty, and let $P'\subseteq C_d$ be arbitrary property such that $P\cap P'$ is an infinite set. Then $P\cap P'$ is not testable.    
\end{lemma}

Now we are ready to finish the proof of Theorem \ref{thm:simpleDelta2}.
\begin{proof}[Proof of Theorem \ref{thm:simpleDelta2}]
	We show that the property $P_{\psi_{\zigzag}}$ does not contain an infinite hyperfinite subproperty. If this is true, then by applying Lemma \ref{lem:FPS19} with $P=P_{\psi_{\zigzag}}$ and $P'=C_d$, we have that $P_{\psi_{\zigzag}}$ is not testable. This will then finish the proof of Theorem \ref{thm:simpleDelta2}.
	
	Suppose towards contradiction that $P_{\psi_{\zigzag}}$ contains an infinite hyperfinite subproperty. That is, there exists an infinite subset $Q\subseteq P_{\psi_{\zigzag}}$ and a function $\rho : (0, 1] \rightarrow \mathbb{N}$ such that $Q$ is $(\varepsilon, \rho(\varepsilon))$-hyperfinite for every $\varepsilon > 0$. 	
	That is, for any graph $G=(V,E)\in Q$, for any $\varepsilon>0$, we can remove $\varepsilon d |V|$ edges from $G$ so that each connected component of the resulting graph has at most $\rho(\varepsilon)$ vertices. Now let $\varepsilon$ be an arbitrarily small constant such that $\rho(\varepsilon)\ll |V|$ and that $\varepsilon\leq \frac{\xi}{100d}$, where $\xi$ is the constant from Lemma \ref{lemma:undirected_expander}. Let $V_1,V_2,\dots$ be a vertex partitioning of $V$ such that $|V_i|\leq \rho(\varepsilon)$ and the number of edges crossing different parts is at most $\varepsilon d |V|$. Let $S$ be a vertex subset that is a union of the first $j$ parts $V_1,\cdots,V_j$ such that $|\cup_{i\leq j-1} V_i|< \frac{|V|}{3}$ and $|\cup_{i\leq j} V_i|\geq\frac{|V|}{3}$. Note that such a set always exists as $|V_i|\leq \rho(\varepsilon)\ll |V|$ and furthermore, $|S|=|\cup_{i\leq j} V_i|<\frac{|V|}{2}$. Now on one hand, $|\langle S, \bar{S}\rangle|$ is at most the number of edges crossing different parts and thus at most $\varepsilon d |V|$. On the other hand, since $G\in P_{\psi_{\zigzag}}$, $G$ is a $\xi$-expander for some constant $\xi$ from Lemma \ref{lemma:undirected_expander}. Thus, $|\langle S, \bar{S}\rangle|\geq \xi \frac{|V|}{3} > \varepsilon d |V|$, which is a contradiction by our setting of $\varepsilon$. Therefore,  $P_{\psi_{\zigzag}}$ does not contain an infinite hyperfinite subproperty. This finishes the proof of the theorem. 
\end{proof}


\section{On the testability of all $\Sigma_2$-properties}\label{sec:testableSigma2}
In this section we let $\sigma=\{R_1,\dots,R_m\}$ be any relational signature and $C_d$ the set of $\sigma$-structures of bounded degree $d$. We prove the following.
\begin{theorem}\label{thm:sigma2}
Every first-order property defined by a $\sigma$-sentence in $\Sigma_2$ is testable in the bounded-degree model.
\end{theorem}

We adapt the notion of indistinguishability of~\cite{alon2000efficient} from the dense model to the bounded degree model.

\begin{definition}\label{def:indistinguishability}
	Two properties $P,Q\subseteq C_d$ are called \emph{indistinguishable} if for every $\epsilon \in (0,1)$ there exists $N=N(\epsilon)$ such that for every structure $\mathcal{A}\in P$ with $|A|>N$ there is a structure $\tilde{\mathcal{A}}\in Q$ with the same universe, that is $\epsilon$-close to $\mathcal{A}$; and for every $\mathcal{B}\in Q$ with $|B|>N$ there is a structure $\tilde{\mathcal{B}}\in P$ with the same universe, that is $\epsilon$-close to $\mathcal{B}$.
\end{definition}
The following lemma follows from the definitions, and is similar to~\cite{alon2000efficient}, though we make use of the canonical testers for bounded degree graphs (\cite{CzumajPS16,goldreich2011proximity}). 
 
\begin{lemma}
	If $P,Q\subseteq C_d$ are indistinguishable properties, then $P$ is testable on $C_d$ if and only if $Q$ is testable on $C_d$.
\end{lemma}
\begin{proof}
We show that if $P$ is testable, then $Q$ is also testable. The other direction follows by the same argument. Let $\epsilon>0$. Since $P$ is testable, there exists an $\frac{\epsilon}{2}$-tester for $P$ with success probability at least $\frac23$. Furthermore, we can assume that the tester (called canonical tester) behaves as follows (see \cite{CzumajPS16,goldreich2011proximity}): it first uniformly samples a constant number of elements, then explores the union of $r$-balls around all sampled elements for some constant $r>0$, and makes a deterministic decision whether to accept, based on an isomorphic copy of the explored substructure. Let  $C=C(\frac{\epsilon}{2},d)$ denote the number of queries the tester made on the input structure. By repeating this tester and taking the majority, we can have a tester $T$ with $c_1\cdot C$ queries and success probability at least $\frac{5}{6}$ for some integer $c_1>0$.

	Let $N$ be a number such that if a structure $\mathcal{B}$ with $n>N$ elements satisfies $Q$, then there exists a  $\tilde{\mathcal{B}}\in P$ with the same universe such that $\dist(\mathcal{B},\tilde{\mathcal{B}})\leq \min\{\frac{\epsilon}{2},\frac{1}{c_2 C\cdot d^{C+2}}\}dn$ for some large constant $c_2>0$. Now we give an $\epsilon$-tester for $Q$. If the input structure $\mathcal{B}$ has size at most $N$, we can query the whole input to decide if it satisfies $Q$ or not. If its size is larger than $N$, then we use the aforementioned $\frac{\epsilon}{2}$-tester for $P$ with success probability at least $\frac{5}{6}$. If $\mathcal{B}$ satisfies $Q$, then there exists $\tilde{\mathcal{B}}\in P$ that differs from $\mathcal{B}$ in no more than ${1}/(c_2 C\cdot d^{C+2}) dn$ places. Since the algorithm samples at most $c_1\cdot C$ elements and queries the $r$-balls around all these sampled elements, for $r< C$, we have that with probability at least $1-\frac{1}{6}$, the algorithm does not query any part where $\mathcal{B}$ and $\tilde{\mathcal{B}}$ differ, and thus its output is correct with probability at least $\frac{5}{6}-\frac{1}{6}=\frac23$. If $\mathcal{B}$ is $\epsilon$-far from satisfying $Q$ then it is $\frac{\epsilon}{2}$-far from satisfying $P$ and with probability at least $\frac{5}{6}>\frac23$, the algorithm will reject $\mathcal{B}$. Thus $Q$ is also testable.
\end{proof}

\paragraph{High-level idea of proof of Theorem~\ref{thm:sigma2}.} 
Let $\varphi\in \Sigma_2$. We prove that the property defined by $\varphi$ can be written as the union of properties, each of which is defined by another formula $\varphi'$ in $\Sigma_2$ where the structure induced by the existentially quantified variables is a fixed structure $\mathcal{M}$ (see Claim \ref{claim:JM}). With some further simplification of $\varphi'$, we obtain a formula $\varphi''$ in $\Sigma_2$ which expresses that the structure has to have $\mathcal{M}$ as an induced substructure (see Claim \ref{claim:non-iso}) and every set of elements of fixed size $\ell$ has to induce some structure from a set of structures $\mathfrak{B}$, and -- depending on the structure from $\mathfrak{B}$ -- 
there might be some connections to the elements of $\mathcal{M}$. We then define a formula $\psi$ in $\Pi_1$ such that the property defined by $\psi$ is indistinguishable from the property defined by $\varphi''$ in the sense that we can transform any structure satisfying  $\psi$, into a structure satisfying $\varphi''$ by modifying no more then a small fraction of the tuples and vice versa (see Claim  \ref{claim:indistinguishable}). The intuition behind this is that every structure satisfying $\varphi''$ can be made to satisfy $\psi$ by removing the structure $\mathcal{M}$ while on the other hand for every structure which satisfies $\psi$ we can plant the structure $\mathcal{M}$ to make it satisfy $\varphi''$. Since it is a priori unclear how the existential and universal quantified variables interact, we have to define $\psi$ very carefully. Here it is important to note that the existence  of occurrences of structures in $\mathfrak{B}$ forcing an interaction with $\mathcal{M}$ is limited because of the degree bound (see Claim \ref{claim:notManyTuples}). Thus such structures can not be allowed to occur for models of $\psi$, as here the number of occurrences can not be limited in any way. Since properties defined by a formula in $\Pi_1$ are testable, this implies with the indistinguishability of $\psi$ and $\varphi''$ that the property defined by $\varphi''$ is testable. Furthermore  by the fact that testable properties are closed under union \cite{goldreich2017introduction}, 
we reach the conclusion that any property defined by a formula in $\Sigma_2$ is testable.

Especially we will not directly give a tester for the property $P_\varphi$  but decompose $\varphi$ into simpler cases. However, every simplification of $\varphi$ used is computable, and the proof below yields a construction of an $\epsilon$-tester for $P_\varphi$ for every $\epsilon\in (0,1)$ and every $\varphi\in \Sigma_2$.\\

For the full proof of Theorem~\ref{thm:sigma2}, we use the following definition.
\begin{definition}
Let $\mathcal{A}$ be a $\sigma$-structure with $A=\{a_1,\dots,a_t\}$. Let $\overline{z}=(z_1,\dots,z_t)$ be a tuple of variables.  
	Then we define $\iota^\mathcal{A}(\overline{z})$ as follows.
\begin{align*}
	\iota^{\mathcal{A}}(\overline{z}):=
	&\bigwedge_{R\in\sigma}\Bigg(\bigwedge_{\big(a_{i_1},\dots,a_{i_{\ar(R)}}\big)\in R^\mathcal{A}}R\big(z_{i_1},\dots,z_{i_{\ar(R)}}\big)
	\land 
	\bigwedge_{\big(a_{i_1},\dots,a_{i_{\ar(R)}}\big)\in A^{\ar(R)}\setminus R^\mathcal{A}}\neg R\big(z_{i_1},\dots,z_{i_{\ar(R)}}\big)\Bigg)	\\
	&\land
	\bigwedge_{\substack{i,j\in[t]\\i\neq j}}(\neg z_i=z_j).
\end{align*}
	
\end{definition}
Note that 
	for every $\sigma$-structure $\mathcal{A}'$ and $\overline{a}'=(a_1',\dots,a_t')\in (A')^t$ we have that $\mathcal{A}'\models \iota^\mathcal{A}(\overline{a}')$ if and only if $a_i\mapsto a_i'$, $i\in \{1,\dots,t\}$ is an isomorphism from $\mathcal{A}$ to $\mathcal{A}'[\{a_1',\dots,a_t'\}]$. In particular, if $\mathcal{A}'\models \iota^\mathcal{A}(\overline{a}')$, then $\{a_1',\dots,a_t'\}$ induces a substructure isomorphic to $\mathcal A$ in $\mathcal A'$.

\begin{proof}[Proof of Theorem~\ref{thm:sigma2}]
Let $\varphi$ be any $\sigma$-sentence in  $\Sigma_2$. Therefore we can assume that $\varphi$ is of the form $\varphi=\exists \overline{x} \,\forall\overline{y} \,\chi(\overline{x},\overline{y})$ where $\overline{x}=(x_1,\dots,x_k)$ is a tuple of $k\in \mathbb{N}$ variables, $\overline{y}=y_1,\dots,y_\ell$ is a tuple of $\ell\in \mathbb{N}$ variables and $\chi(\overline{x},\overline{y})$ is a quantifier free formula. We can further assume 
that $\chi(\overline{x},\overline{y})$ is in disjunctive normal form, and that \begin{eqnarray}
\varphi=\exists \overline{x}\,\forall \overline{y}\bigvee_{i\in I}\Big(\alpha^i(\overline{x})\land \beta^i(\overline{y})\land \operatorname{pos}^i(\overline{x},\overline{y})\land \operatorname{neg}^i(\overline{x},\overline{y})\Big),\label{eqn:phi_1}
\end{eqnarray}
 where $\alpha^i(\overline{x})$ is a conjunction of literals only containing variables from $\overline{x}$, $\beta^i(\overline{y})$ is a conjunction of literals only containing  variables in $\overline{y}$, $\operatorname{neg}^i(\overline{x},\overline{y})$ is a conjunction of negated atomic formulas containing both variables from $\overline{x}$ and $\overline{y}$ and $\operatorname{pos}^i(\overline{x},\overline{y})$ is a conjunction of atomic formulas containing both variables from $\overline{x}$ and $\overline{y}$. 
	Now note that if an expression `$x_j=y_{j'}$' appears in a conjunctive clause, then we can replace every occurrence of $y_{j'}$ by $x_j$ in that clause, which will result in an equivalent formula.

	We now write the formula $\varphi$ given in (\ref{eqn:phi_1}) as a disjunction over all possible structures in $C_d$ the existentially quantified variables could enforce. Since the elements realising the existentially quantified variables will have a certain structure, it is natural to decompose the formula in this way.

	Let $\mathfrak{M}\subseteq C_d$ be a set of models of $\varphi$, such that every model $\mathcal{A}\in C_d$ of $\varphi$ contains an isomorphic copy of some $\mathcal{M}\in \mathfrak{M}$ as an induced substructure, and $\mathfrak{M}$ is minimal with this property.

	\begin{claim}\label{claim:aboutM}
		Every $\mathcal{M}\in \mathfrak{M}$ has at most $k$ elements.
	\end{claim}
	\begin{proof}
		Assume there is $\mathcal{M}\in \mathfrak{M}$ with $|M|>k$. Since every structure in $\mathfrak{M}$ is a model of $\varphi$ there must be a tuple $\overline{a}=(a_1,\dots,a_k)\in M^k$ such that $\mathcal{M}\models \forall \overline{y}\bigvee_{i\in I}\Big(\alpha^i(\overline{a})\land \beta^i(\overline{y})\land \operatorname{pos}^i(\overline{a},\overline{y})\land \operatorname{neg}^i(\overline{a},\overline{y})\Big)$. This implies that for every tuple $\overline{b}\in M^\ell$ we have $\mathcal{M}\models \bigvee_{i\in I}\Big(\alpha^i(\overline{a})\land \beta^i(\overline{b})\land \operatorname{pos}^i(\overline{a},\overline{b})\land \operatorname{neg}^i(\overline{a},\overline{b})\Big)$. Furthermore, since $\{a_1,\dots,a_k\}^\ell\subseteq M^\ell$ we have that $\mathcal{M}[\{a_1,\dots,a_k\}]\models \forall \overline{y}\bigvee_{i\in I}\Big(\alpha^i(\overline{a})\land \beta^i(\overline{y})\land \operatorname{pos}^i(\overline{a},\overline{y})\land \operatorname{neg}^i(\overline{a},\overline{y})\Big)$. This means that $\mathcal{M}[\{a_1,\dots,a_k\}]\models \varphi$. 
		Hence by definition, $\mathfrak{M}$ contains an induced substructure
		$\mathcal{M}'$ of $\mathcal{M}[\{a_1,\dots,a_k\}]$. But then
		$\mathcal{M}'$ is an induced substructure of $\mathcal{M}$ with strictly
		fewer elements than $\mathcal{M}$, a contradiction to the definition of
		$\mathfrak{M}$. 
	\end{proof}
	Therefore $\mathfrak{M}$ is finite.		
		For $\mathcal{M}\in\mathfrak{M}$ let $J_\mathcal{M}:=\{j\in I\mid \mathcal{M}\models \alpha^j(\overline{m})\text{ for some }\overline{m}\in M^{\ell}\}\subseteq I$.
	\begin{claim}\label{claim:JM}
We have		$\varphi \equiv_d \bigvee_{\mathcal{M}\in \mathfrak{M}}\Big(\exists \overline{x}\forall\overline{y} \Big[\iota^\mathcal{M}(\overline{x})\land \bigvee_{j\in J_\mathcal{M}}\Big( \beta^j(\overline{y})\land \operatorname{pos}^j(\overline{x},\overline{y})\land \operatorname{neg}^j(\overline{x},\overline{y})\Big) \Big] \Big)$.
	\end{claim}
	\begin{proof}
Let $\mathcal{A}\in C_d$ be a model of $\varphi$. Then there is a tuple $\overline{a}=(a_1,\dots,a_k)\in A^k$ such that $\mathcal{A}\models \forall y \chi (\overline{a},\overline{y})$. Since $\{a_1,\dots,a_k\}^\ell\subseteq A^\ell$ this implies that $\mathcal{A}[\{a_1,\dots,a_k\}]\models   \forall y \chi (\overline{a},\overline{y})$ and hence $\mathcal{A}[\{a_1,\dots,a_k\}]\models \varphi$.
		In addition, we may assume that we picked $\overline{a}$ in such a way that for any tuple $\overline{a}'=(a_1',\dots,a_k')\in \{a_1,\dots,a_k\}^k$  
	 with $\{a_1',\dots,a_k'\}\subsetneq \{a_1,\dots,a_k\}$ we have that $\mathcal{A}\not\models \forall \overline{y} \chi(\overline{a}',\overline{y})$. (The reason is that if for some tuple $\overline{a}'$ this is not the case then we just replace $\overline{a}$ by $\overline{a}'$ and so on until this property holds). Hence $\mathcal{A}[\{a_1,\dots,a_k\}]$ cannot have a proper induced substructure in $\mathfrak{M}$, and it follows that there is $\mathcal{M}\in \mathfrak{M}$ such that $\mathcal{M}\cong \mathcal{A}[\{a_1,\dots,a_k\}]$. By choice of $J_{\mathcal{M}}$ we get $\mathcal{A}\models \forall\overline{y} \Big[\iota^\mathcal{M}(\overline{a})\land \bigvee_{j\in J_\mathcal{M}}\Big( \beta^j(\overline{y})\land \operatorname{pos}^j(\overline{a},\overline{y})\land \operatorname{neg}^j(\overline{a},\overline{y})\Big) \Big]$ and hence \[\mathcal{A}\models \bigvee_{\mathcal{M}\in \mathfrak{M}}\Big(\exists \overline{x}\forall\overline{y} \Big[\iota^\mathcal{M}(\overline{x})\land \bigvee_{j\in J_\mathcal{M}}\Big( \beta^j(\overline{y})\land \operatorname{pos}^j(\overline{x},\overline{y})\land \operatorname{neg}^j(\overline{x},\overline{y})\Big) \Big] \Big).\]
		
		To prove  the other direction, we now let the structure $\mathcal{A}\in C_d$ be a model of the formula \newline$\bigvee_{\mathcal{M}\in \mathfrak{M}}\Big(\exists \overline{x}\forall\overline{y} \Big[\iota^\mathcal{M}(\overline{x})\land \bigvee_{j\in J_\mathcal{M}}\Big( \beta^j(\overline{y})\land \operatorname{pos}^j(\overline{x},\overline{y})\land \operatorname{neg}^j(\overline{x},\overline{y})\Big) \Big] \Big)$. Consequently there is $\mathcal{M}\in \mathfrak{M}$ and $\overline{a}\in A^k$ such that $\mathcal{A}\models \forall\overline{y} \Big[\iota^\mathcal{M}(\overline{a})\land \bigvee_{j\in J_\mathcal{M}}\Big( \beta^j(\overline{y})\land \operatorname{pos}^j(\overline{a},\overline{y})\land \operatorname{neg}^j(\overline{a},\overline{y})\Big) \Big]$. By choice of $J_\mathcal{M}$ this implies $\mathcal{A} \models \forall \overline{y}\bigvee_{j\in J_\mathcal{M}}\Big(\alpha^j(\overline{a})\land \beta^j(\overline{y})\land \operatorname{pos}^j(\overline{a},\overline{y})\land \operatorname{neg}^j(\overline{a},\overline{y})\Big)$ and hence $\mathcal{A}\models \varphi$.		
	\end{proof}
	
	Since the union of finitely many testable properties is testable (see e.g.~\cite{goldreich2017introduction}),  it is sufficient to show that the property $P_\varphi$ where $\varphi$ is of the form 	
	\begin{eqnarray}
	\varphi =\exists \overline{x}\forall \overline{y}\chi(\overline{x},\overline{y}), \text{ where } \chi(\overline{x},\overline{y})=\Big[\iota^\mathcal{M}(\overline{x})\land \bigvee_{j\in J_\mathcal{M}}\Big(\beta^j(\overline{y})\land \operatorname{pos}^j(\overline{x},\overline{y})\land \operatorname{neg}^j(\overline{x},\overline{y})\Big)\Big], \label{eqn:phi_2}
\end{eqnarray} 
for some $\mathcal{M}\in \mathfrak{M}$ is testable. 
	In the following, we will enforce that for every conjunctive clause of the big disjunction of $\chi$, the universally quantified variables induce a specific substructure.

For $j\in J_\mathcal{M}$ let $\mathfrak{H}_j\subseteq C_d$ be a maximal set of pairwise non-isomorphic structures $\mathcal{H}$ such that   $\mathcal{H}\models \beta^j(\overline{b})$ for some $\overline{b}=(b_1,\dots,b_\ell)\in H^\ell$ with $\{b_1,\dots,b_\ell\}=H$. 	
	
	\begin{claim}\label{claim:non-iso}
We have $\varphi \equiv_d \exists \overline{x}\forall\overline{y} \Big[\iota^\mathcal{M}(\overline{x})\land \bigvee_{{\mathcal{H}\in \mathfrak{H}_j,}\atop{j\in J_\mathcal{M}}} \Big(\iota^\mathcal{H}(\overline{y})\land \operatorname{pos}^j(\overline{x},\overline{y})\land \operatorname{neg}^j(\overline{x},\overline{y})\Big)\Big].$ 
	\end{claim}
	\begin{proof}
		Let $\mathcal{A}\in C_d$ and $\overline{a}=(a_1,\dots,a_k)\in A^k$. First assume that $\mathcal{A}\models \forall \overline{y}\chi(\overline{a},\overline{y})$. Hence for any tuple $\overline{b}\in A^\ell$ there is an index $j\in J_\mathcal{M}$ such that $\mathcal{A}\models \beta^j(\overline{b})\land \operatorname{pos}^j(\overline{a},\overline{b})\land \operatorname{neg}^j(\overline{a},\overline{b})$. 
		Then $\mathcal{A}\models \beta^j(\overline{b})$ implies that $\mathcal{A}[\{b_1,\dots,b_\ell\}]\cong \mathcal{H}$ for some $\mathcal{H}\in \mathfrak{H}_j$. Hence $\mathcal{A}\models \iota^\mathcal{H}(\overline{b})$ and $\mathcal{A}\models  \Big[\iota^\mathcal{M}(\overline{a})\land \bigvee_{{\mathcal{H}\in \mathfrak{H}_j,}\atop{j\in J_\mathcal{M}}} \Big(\iota^\mathcal{H}(\overline{b})\land \operatorname{pos}^j(\overline{a},\overline{b})\land \operatorname{neg}^j(\overline{a},\overline{b})\Big)\Big]$.
		
		For the other direction, we let $\mathcal{A}\models \forall\overline{y} \Big[\iota^\mathcal{M}(\overline{a})\land \bigvee_{{\mathcal{H}\in \mathfrak{H}_j,}\atop{j\in J_\mathcal{M}}} \Big(\iota^\mathcal{H}(\overline{y})\land \operatorname{pos}^j(\overline{a},\overline{y})\land \operatorname{neg}^j(\overline{a},\overline{y})\Big)\Big]$. Then for every tuple $\overline{b}\in A^\ell$ there is an index $j\in J_\mathcal{M}$ and $\mathcal{H}\in \mathfrak{H}_j$ such that $\mathcal{H}\models \iota^\mathcal{H}(\overline{b})\land \operatorname{pos}^j(\overline{a},\overline{b})\land \operatorname{neg}^j(\overline{a},\overline{b})$. Therefore $\mathcal{A}[\{b_1,\dots,b_\ell\}]\cong \mathcal{H}$ and we know that $\mathcal{A}\models \beta^j(\overline{b})$. Therefore $\mathcal{A}\models \beta^j(\overline{b})\land \operatorname{pos}^j(\overline{a},\overline{b})\land \operatorname{neg}^j(\overline{a},\overline{b})$ and since this is true for any $\overline{b}\in A^\ell$ we get $\mathcal{A}\models \varphi$.
	\end{proof}

Thus, it suffices to assume that 
\begin{eqnarray}
\varphi =\exists \overline{x}\forall\overline{y} \chi(\overline{x},\overline{y}), \text{ where  }\chi(\overline{x},\overline{y}):=\Big[\iota^\mathcal{M}(\overline{x})\land \bigvee_{{\mathcal{H}\in \mathfrak{H}_j,}\atop{j\in J_\mathcal{M}}} \Big(\iota^\mathcal{H}(\overline{y})\land \operatorname{pos}^j(\overline{x},\overline{y})\land \operatorname{neg}^j(\overline{x},\overline{y})\Big)\Big]\label{eqn:phi_3}
\end{eqnarray}
for some $\mathcal{M}\in \mathfrak{M}$. 
	
Next we will define a universally quantified formula $\psi$  and show that $P_\varphi$ is indistinguishable from the property $P_\psi$. To do so we will need the two claims below.
Intuitively, Claim~\ref{claim:notManyTuples} says that models of $\varphi$ of bounded degree
do not have many `interactions'
between existential and universal variables -- only a constant number of tuples in relations
combine both types of variables.
Note that for a structure $\mathcal{A}$ and tuples $\overline{a}\in A^k$, $\overline{b}=(b_1,\dots,b_\ell)\in A^\ell$ the condition $\mathcal{A}\models \iota^\mathcal{H}(\overline{b})\land \operatorname{pos}^j(\overline{a},\overline{b})\land \operatorname{neg}^j(\overline{a},\overline{b})$ can force an element of $\overline{b}$ to be in a tuple (of a relation of $\mathcal A$) with an element of $\overline{a}$, even if $\operatorname{pos}^j(\overline{x},\overline{y})$ only contains literals of the form $x_i=y_{i'}$. (For example,  
it may be the case that for some tuple $\overline{b}'\in \{b_1,\dots,b_\ell\}^\ell$, every clause $\iota^\mathcal{H'}(\overline{y})\land \operatorname{pos}^{j'}(\overline{x},\overline{y})\land \operatorname{neg}^{j'}(\overline{x},\overline{y})$ for which $\mathcal{A}\models \iota^\mathcal{H'}(\overline{b}')\land \operatorname{pos}^{j'}(\overline{a},\overline{b}')\land \operatorname{neg}^{j'}(\overline{a},\overline{b}')$ enforces a tuple to contain some element of $\overline{b}'$ and some element of $\overline{a}$.) We will now define a set $J$ to pick out the clauses that do not enforce a tuple to contain both an element from $\overline{a}$ and $\overline{b}$. Note that we still allow elements from $\overline{b}$ to be amongst the elements in $\overline{a}$. In Claim~\ref{claim:notManyTuples} we show that for every $\mathcal{A}\in C_d$, $\overline{a}\in A^k$ for which $\mathcal{A}\models \forall \overline{y} \chi(\overline{a},\overline{y})$ there are only a constant number of tuples $\overline{b}\in A^\ell$ that only satisfy clauses which enforce a tuple to contain both an element from $\overline{a}$ and from $\overline{b}$.

	Let $j\in \mathcal{M}$, $\mathcal{H}\in \mathfrak{H}_j$ and $\overline{h}=(h_1,\dots,h_\ell)\in H^\ell$ such that $\mathcal{H}\models \iota^\mathcal{H}(\overline{h})$. We define the set $P_{j,\mathcal{H}}:=\{h_i\mid i\in \{1,\dots,\ell\}, \operatorname{pos}^j(\overline{x},\overline{y})\text{ does not contain }y_i=x_{i'} \text{ for any }i'\in \{1,\dots,k\}\}.$
	Now we let $J\subseteq J_\mathcal{M}\times C_d$ be the set of pairs $(j,\mathcal{H})$, with $\mathcal{H}\in \mathfrak{H}_j$ with the following two properties. Firstly $\operatorname{pos}^j(\overline{x},\overline{y})$ only contains literals of the form $x_{i'}=y_{i}$ for some $i\in\{1,\dots,\ell\}$, $i'\in \{1,\dots,k\}$. Secondly  the disjoint union $\mathcal{M}\sqcup \mathcal{H}[P_{j,\mathcal{H}}]\models \varphi$. $J$ now precisely specifies the clauses that can be satisfied by a structure $\mathcal{A}$ and tuple $\overline{a}\in A^k$ and $\overline{b}\in A^\ell$ where $\mathcal{A}$ does not contain any tuples both containing elements from $\overline{a}$ and $\overline{b}$.

	\begin{claim}\label{claim:notManyTuples}
		Let $\mathcal{A}\in C_d$ and $\overline{a}=(a_1,\dots,a_k)\in A^k$. If  $\mathcal{A}\models \forall \overline{y}\,\chi(\overline{a},\overline{y})$ then there are at most $k \cdot d$ tuples $\overline{b}\in A^\ell $ such that $\mathcal{A}\not\models \bigvee_{(j,\mathcal{H})\in J}(\iota^\mathcal{H}(\overline{b})\land \operatorname{pos}^j(\overline{a},\overline{b})\land \operatorname{neg}^j(\overline{a},\overline{b}))$.
	\end{claim}
	\begin{proof}
		Since	$\mathcal{A}\models \forall \overline{y}\,\chi(\overline{a},\overline{y})$, it holds that $\mathcal{A}\models \forall \overline{y} \bigvee_{{\mathcal{H}\in \mathfrak{H}_j,}\atop{j\in J_\mathcal{M}}} \Big(\iota^\mathcal{H}(\overline{y})\land \operatorname{pos}^j(\overline{a},\overline{y})\land \operatorname{neg}^j(\overline{a},\overline{y})\Big)$ by Equation(\ref{eqn:phi_3}). Now let 
		$B:=\{\overline{b}\in A^\ell\mid \mathcal{A}\not\models \bigvee_{(j,\mathcal{H})\in J}(\iota^\mathcal{H}(\overline{b})\land \operatorname{pos}^j(\overline{a},\overline{b})\land \operatorname{neg}^j(\overline{a},\overline{b}))\}\subseteq A^\ell$. Then every $\overline{b}\in B$ adds at least one to $\sum_{i=1}^{k}\deg_\mathcal{A}(a_i)$. Since $\mathcal{A}\in C_d$ implies that $\sum_{i=1}^{k}\deg_\mathcal{A}(a_i)\leq k\cdot d$ we get that $|B|\leq k \cdot d$.
	\end{proof}	
	\begin{claim}\label{claim:propertiesOfUniversalProperties}
		Let $\psi$ be a formula of the form $\psi =  \forall \overline{z} \bigvee_{i\in I}c^i(\overline{z})$ where $\overline{z}=(z_1,\dots,z_t)$ is a tuple of variables and $c^i$ is a conjunction of literals. Let $\mathcal{A}\in C_d $ with $|A|> d\cdot \ar(\sigma)\cdot t$ 
		and let $b\in A$ be an arbitrary element. Let $\mathcal{A}\models \psi$ and let $\mathcal{A}'$ be obtained from $\mathcal{A}$ by `isolating' $b$, i.\,e.\ by deleting all tuples containing $b$ from $R^\mathcal{A}$ for every $R\in \sigma$. Then  $\mathcal{A}'\models\psi$. 
	\end{claim}
	\begin{proof}
		First note that $\mathcal{A}'\models\bigvee_{i\in I}c^i(\overline{a})$ for any tuple $\overline{a}=(a_1,\dots,a_t)\in (A\setminus \{b\})^t$ as no tuple over the set of elements $\{a_1,\dots,a_t\}$ has been deleted. 
		Let $\overline{a}=(a_1,\dots,a_t)\in A^t$ be a tuple containing $b$. Pick $b'\in A$ such that $\operatorname{dist}_\mathcal{A}(a_j,b')>1$ for every $j\in \{1,\dots,t\}$. Such an element exists as $|A|> d\cdot \ar(R)\cdot t$. Let $\overline{a}'=(a_1',\dots,a_t')$ be the tuple obtained from $\overline{a}$ by replacing any occurrence of $b$ by $b'$. Hence $a_j\mapsto a_j'$ defines an isomorphism from $\mathcal{A}'[\{a_1,\dots,a_t\}]$ to $\mathcal{A}[\{a_1',\dots,a_t'\}]$ since $b$ is an isolated element in $\mathcal{A}'[\{a_1,\dots,a_t\}]$ and $b'$ is an isolated element in $\mathcal{A}[\{a_1',\dots,a_t'\}]$. Since $\mathcal{A}\models \bigvee_{i\in I}c^i(\overline{a}')$, it follows that $\mathcal{A}'\models \bigvee_{i\in I}c^i(\overline{a})$.
	\end{proof}
Let $J'\subseteq J$ be the set of all pairs $(j,\mathcal{H})$ for which $\operatorname{pos}^j(\overline{x},\overline{y})$ is the empty conjunction. $J'$ contains $(j,\mathcal{H})$ for which we want to use $\iota^\mathcal{H}(\overline{y})$ to define the formula $\psi$.
	\begin{claim}\label{claim:indistinguishable}
		The property $P_\varphi$ with $\varphi$ as in (\ref{eqn:phi_3}) is indistinguishable from the property $P_\psi$ where $\psi:=\forall \overline{y} \bigvee_{(j,\mathcal{H})\in J'}\iota^\mathcal{H}(\overline{y})$.  
	\end{claim}
	\begin{proof} Let $\epsilon>0$ and  $N(\epsilon)=N:= \frac{k\cdot \ell^2\cdot d\cdot \ar(R)}{\epsilon}$ and $\mathcal{A}\in C_d$ be  any structure with $|A|>N$. 
		
		First assume that $\mathcal{A}\models \varphi$. The strategy is to isolate any element $b$  which is contained in a tuple $\overline{b}\in A^\ell$ such that  $\mathcal{A}\not \models \bigvee_{(j,\mathcal{H})\in J'}\iota^\mathcal{H}(\overline{b})$ by deleting all tuples containing $b$. This will result in a structure which is  $\epsilon$-close to $\mathcal{A}$ and a model of $\psi$.
		
		Let $\overline{a}\in A^k$ be a tuple such that $\mathcal{A}\models \forall \overline{y}\chi(\overline{a},\overline{y})$. Let $B\subseteq A^\ell$ be the set of tuples $\overline{b}\in A^\ell $ such that $\mathcal{A}\not\models \bigvee_{(j,\mathcal{H})\in J}(\iota^\mathcal{H}(\overline{b})\land \operatorname{pos}^j(\overline{a},\overline{b})\land \operatorname{neg}^j(\overline{a},\overline{b}))$.  Then $|B|\leq \ell\cdot d\cdot \ar(R)$ by Claim \ref{claim:notManyTuples}.  Hence the structure $\mathcal{A}'$ obtained from $\mathcal{A}$ by deleting all tuples containing an element of  $C:=\{a_1,\dots,a_k,b_1,\dots,b_\ell\mid(b_1,\dots,b_\ell)\in B\}$ is $\epsilon$-close to $\mathcal{A}$. 
		Since $\mathcal{A}\models \forall \overline{y}\chi(\overline{a},\overline{y})$ implies $\mathcal{A}\models \forall \overline{y}\bigvee_{{\mathcal{H}\in \mathfrak{H}_j,}\atop{j\in J_\mathcal{M}}} \iota^\mathcal{H}(\overline{y})$ by Claim \ref{claim:propertiesOfUniversalProperties} we know that $\mathcal{A}'\models  \forall \overline{y}\bigvee_{{\mathcal{H}\in \mathfrak{H}_j,}\atop{j\in J_\mathcal{M}}} \iota^\mathcal{H}(\overline{y})$. 
		For any tuple $\overline{b}=(b_1,\dots,b_\ell)\in (A\setminus C)^\ell$ we have by definition of $J'$ that $\mathcal{A}\models \iota^\mathcal{H}(\overline{b})$ for some $(j,\mathcal{H})\in J'$. Furthermore $\mathcal{A}[\{b_1,\dots,b_\ell\}]=\mathcal{A}'[\{b_1,\dots,b_\ell\}]$ and hence $\mathcal{A}'\models \bigvee_{(j,\mathcal{H})\in J'}\iota^\mathcal{H}(\overline{b})$.
		Let $\overline{b}=(b_1,\dots,b_\ell)\in A^\ell$ be any tuple containing element from $C$ and let $c_1,\dots,c_t\in C$ be those elements.  
		Pick $t$ elements $c_1',\dots,c_{t}'\in A\setminus C$ such that $\operatorname{dist}_\mathcal{A}(a_i,c_{i'}')>1$ and $\operatorname{dist}_\mathcal{A}(c_{i'}',b_{i})>1$ for suitable $i,i'$. This is possible as $|A|> (k+2\ell)\cdot d\cdot \ar(R)$ which guarantees the existence of $k+2\ell$ elements of pairwise distance $1$.
		Let $\overline{b}'=(b_1',\dots,b_\ell')$ be the vector obtained from $\overline{b}$ by replacing $c_{i}$  with $c'_i$. Since $\overline{b}'\in A^\ell$ there must be $j'$, $\mathcal{H'}\in \mathfrak{H}_j$ such that  $\mathcal{A}\models \iota^\mathcal{H'}(\overline{b}')\land \operatorname{pos}^{j'}(\overline{a},\overline{b}')\land \operatorname{neg}^{j'}(\overline{a},\overline{b}')$. By choice of $c_1',\dots,c_{t}'$ we have that $\operatorname{pos}_{j'}(\overline{x},\overline{y})$ must be the empty conjunction and hence $(j',\mathcal{H'})\in J'$. 
		Since additionally $b_{i}\mapsto b_{i}'$ defines an isomorphism of $\mathcal{A}[\{b_1',\dots,b_\ell'\}]$ and $ \mathcal{A}'[\{b_1,\dots,b_\ell\}]$ this implies that $\mathcal{A}'\models  \bigvee_{(j,\mathcal{H})\in J'}\iota^\mathcal{H}(\overline{b})$ for all $\overline{b}\in A^\ell$ and hence $\mathcal{A}'\models \psi$.\\

		Now we prove the other direction. Let $\mathcal{A}\models \psi$ with $|A|>N$. The idea here is to plant the structure $\mathcal{M}$ somewhere in $\mathcal{A}$. While this takes less then an $\epsilon$ fraction of edge modifications the resulting structure will be a model of $\varphi$.
		
		Take any set $B\subseteq A$ of $|M|$ elements. Let $\mathcal{A}'$ be the structure obtained from $\mathcal{A}$ by deleting all edges incident to any element contained in $B$. Let $\mathcal{A}''$ be the structure obtained from $\mathcal{A}'$ by adding all tuples such that the structure induced by $B$ is isomorphic to $\mathcal{M}$. This takes no more then $2\ell\cdot d\cdot \ar(R)<\epsilon\cdot d \cdot |A|$ edge modifications Let $\overline{a}\in B^k$ be such that $\mathcal{A}\models \iota^\mathcal{M}(\overline{a})$.  By Claim \ref{claim:propertiesOfUniversalProperties} we get $\mathcal{A}'\models \psi$. Therefore pick any tuple $\overline{b}=(b_1,\dots,b_\ell)\in (A\setminus B)^\ell$. Since by construction we have that all $b_i$'s are of distance at least one from $\overline{a}$ we have that  $\mathcal{A}''\models \bigvee_{(j,\mathcal{H})\in J'}(\iota^\mathcal{H}(\overline{b})\land \operatorname{neg}^j(\overline{a},\overline{b}))$. By choice of $\mathcal{M}$ we also know that $\mathcal{A}''\models  \bigvee_{{\mathcal{H}\in \mathfrak{H}_j,}\atop{j\in J_\mathcal{M}}}\Big(\iota^\mathcal{H}(\overline{b})\land \operatorname{pos}^j(\overline{a},\overline{b})\land \operatorname{neg}^j(\overline{a},\overline{b})\Big)$ for all $\overline{b}\in B^\ell$. 
		Therefore pick $\overline{b}=(b_1,\dots,b_\ell)$ containing both elements from $B$ and from $A\setminus B$. Now pick a tuple $\overline{b}'=(b_1',\dots,b_\ell')\in (A\setminus B)^\ell$ that equals $\overline{b}$ in all positions containing an element from $A\setminus B$. As noted before there is $(j,\mathcal{H})\in J'$ such that $\mathcal{A}''\models (\iota^\mathcal{H}(\overline{b}')\land \operatorname{neg}^j(\overline{a},\overline{b}'))$. By the definition of $J,J'$ this means that  $\mathcal{A}''[\{a_1,\dots, a_k,b'_1\dots b_\ell'\}]\models\varphi$. Since $\overline{b}\in \{a_1,\dots, a_k,b'_1\dots b_\ell'\}^\ell$  this implies $\mathcal{A}''[\{a_1,\dots, a_k,b'_1\dots b_\ell'\}]\models \bigvee_{{\mathcal{H}\in \mathfrak{H}_j,}\atop{j\in J_\mathcal{M}}}\Big(\iota^\mathcal{H}(\overline{b})\land \operatorname{pos}^j(\overline{a},\overline{b})\land \operatorname{neg}^j(\overline{a},\overline{b})\Big)$. Then $\mathcal{A}''\models \bigvee_{{\mathcal{H}\in \mathfrak{H}_j,}\atop{j\in J_\mathcal{M}}}\Big(\iota^\mathcal{H}(\overline{b})\land \operatorname{pos}^j(\overline{a},\overline{b})\land \operatorname{neg}^j(\overline{a},\overline{b})\Big)$ and hence $\mathcal{A}''\models \varphi$.
	\end{proof}
	Since $\psi \in \Pi_1$ we have that $P_\psi$ is testable, and hence $P_\varphi$ is testable by Claim \ref{claim:indistinguishable}.
\end{proof}

\section{Testing properties of neighbourhoods}
\label{sec:freeness}
In this section we only consider simple graphs, \ie undirected graphs without self-loops and without parallel edges, and  for any  $d\in \mathbb{N}$ let $C_d$  be the class of simple graphs of bounded degree $d$.
We view simple graphs as structures over the signature $\sigma_{\operatorname{graph}}:=\{E\}$, where $E$ encodes a binary, symmetric and irreflexive relation. This allows transferring the notions from Section~\ref{sec:preliminaries} to graphs. 




Let $r\geq 1$ and let $\tau$ be an $r$-type and let $\varphi_{\tau}(x)$ be a FO formula saying that $x$ has $r$-type $\tau$. 
We say that a graph $G$ is \emph{$\tau$-neighbourhood regular}, if $G\models \forall x\varphi_{\tau}(x)$. We say that a graph $G$ is \emph{$\tau$-neighbourhood free}, if $G \models \lnot \exists x \varphi_{\tau}(x)$. Let $\tau_1,\dots,\tau_t$ be a list of all $r$-types in $C_d$. If $F\subseteq \{\tau_1,\dots,\tau_t\}$ we say that $G$ is $F$-free, if $G$ is $\tau$-neighbourhood free for all $\tau\in F$. 

Observe that both  $\tau$-neighbourhood-freeness and $\tau$-neighbourhood regularity can be defined by formulas in $\Pi_2$ for any neighbourhood type $\tau$. Hence the next Lemma shows that there exist neighbourhood properties that are in $\Pi_2$, but not in $\Sigma_2$.
\begin{lemma}\label{lem:existencesigma2}
There exist $1$-types $\tau,\tau'$ such that neither $\tau$-neighbourhood freeness nor $\tau'$-neighbourhood regularity can be defined by a formula in $\Sigma_2$. 
\end{lemma}

Note that the above lemma implies that we cannot simply invoke the testers for testing $\Sigma_2$ properties from Theorem \ref{thm:sigma2} to test these two properties.


\begin{proof}[Proof of Lemma \ref{lem:existencesigma2}]
	For $n\in \mathbb{N}$, let $C_n$ be the cycle graph with vertex set $[n]:=\{0,1,\dots,n-1\}$. 
	Let $P_{n-1}$ be the path graph with vertex set $[n-1]$. We first show the following claim. 
	
	 \begin{claim}\label{claim:sigma2DistinguishingCycleFromPath}
Let $\varphi= \exists \overline{x}\forall \overline{y} \chi(\overline{x},\overline{y})$ where $\overline{x}=(x_1,\dots,x_k)$, $\overline{y}=(y_1,\dots,y_\ell)$ are tuples of variables and $\chi(\overline{x},\overline{y})$ is a quantifier free formula.  If $C_n\models \varphi$ then  $P_{n-1}\models \varphi$ for any $n>k$.
	\end{claim}
	\begin{proof}
		
Assume on the contrary that for some $n>\max\{k,\ell\}$, it holds that $C_n \models \varphi$, while $P_{n-1}\not\models \varphi$. Since $C_n\models \varphi$ there are $k$ vertices $v_1,\dots,v_k$ in $C_n$ such that $C_n\models \forall \overline{y} \chi((v_1,\dots,v_k),\overline{y})$. Since $n>k$, there exists at least one vertex $i\in [n]$ that is not amongst $v_1,\dots,v_k$. Let $v_j':= (v_j+n-1-i)\mod n$ be a vertex of $P_{n-1}$. Since $P_{n-1}\not\models \varphi$ and $v_j'\in [n-1]$, we have that $P_{n-1}\not\models \forall \overline{y}\chi((v_1'\dots,v_k'),\overline{y})$. Hence there must be vertices $w_1',\dots,w_\ell'$ in $P_{n-1}$ such that $P_{n-1}\not\models \chi((v_1',\dots,v_k'),(w_1',\dots,w_\ell'))$. Now let $w_j:=(w_j'+i+1)\mod n$. Then $v_j\mapsto v_j'$ and $w_j\mapsto w_j'$ defines an isomorphism from $C_n[\{v_1,\dots,v_k,w_1,\dots,w_\ell\}]$ and $P_{n-1}[\{v_1',\dots,v_k',w_1',\dots,w_\ell'\}]$. Hence $C_n\not \models \chi((v_1,\dots,v_k),(w_1,\dots,w_\ell))$ which contradicts that $C_n\models \varphi$.
	\end{proof}
	
Now we let $\tau$ be the $1$-neighbourhood type saying that the center vertex $x$ has exactly one neighbour. Let $\tau'$ be the $1$-neighbourhood type saying that the center vertex has two non-adjacent vertices. Since $C_n$ is $\tau$-neighbourhood free and $\tau'$-neighbourhood regular, while $P_{n-1}$ is neither, the statement of the lemma follows from Claim~\ref{claim:sigma2DistinguishingCycleFromPath}. 
\end{proof}

Now we state our main algorithmic results in this section. The first result shows that if $\tau$ is an $r$-type with degree smaller than the degree bound of the class of graphs, then the $\tau$-neighbourhood-freeness is testable.
\begin{theorem}\label{thm:dNeighbourhoodFreeness}
	Let $\tau$ be an $r$-type, where $r\geq 1$. If $\tau \subseteq C_{d'}$ and $d'<d$, then $\tau$-neighbourhood freeness is uniformly testable on the class $C_{d}$ with constant running time.
\end{theorem}
The second result shows if $\tau$ is a $1$-type, then $\tau$-neighbourhood-freeness is testable.
\begin{theorem}\label{thm:1NeighbourhoodFreeness}
	For every $1$-type $\tau$, $\tau$-neighbourhood freeness is uniformly testable on the class $C_{d}$ with constant time.
\end{theorem}
The third result says that $\tau$-neighbourhood regularity is testable for every $1$-type $\tau$ consisting of cliques, which only overlap in the centre vertex. 
\begin{theorem}\label{thm:neighbourhoodRegularity}
	Let $\tau$ be a $1$-type such that  vertex $a$ having $1$-type $\tau$ in $B$  implies that $B\setminus \{a\}$ is a union of disjoint cliques for every $1$-ball $B$ with centre $a$. Then $\tau$-neighbourhood regularity is uniformly testable on $C_{d}$ in 
	constant time. 
\end{theorem}

By previous discussions, the above theorems imply that there are formulas in $\Pi_2\setminus\Sigma_2$ which are testable.

\subsection{Proofs of Theorem \ref{thm:dNeighbourhoodFreeness}, \ref{thm:1NeighbourhoodFreeness} and \ref{thm:neighbourhoodRegularity}}
We use the algorithm $\sampler_{r,s}$, that, given access to a graph $G\in C_{d}$,
samples a set $S$ of $s$ vertices of $G$ uniformly and independently and explores their $r$-balls. The algorithm returns
the \emph{distribution vector} $\bar v$ of length $t$ of the $r$-types of this sample, \ie $\bar v_i=|\{v\in S\mid \mathcal{N}_r^G(v)\in \tau_i\}|/s$.
\begin{lemma}[Lemma 5.1 in~\cite{NewmanSohler2013}]\label{lem:estimate-frequencies}
	Let $\lambda\in (0,1]$, $r\in \mathbb{N}$ and $G\in C_{d}$ with $n$ vertices.
	Let $s\geq ({t^2}/{\lambda^2})\ln(t+40)$.
	Then the vector $\bar v$ returned
	by $\sampler_{r,s}(G)$ satisfies
	$\sum_{i=1}^{t}|\rho_{G,r}(\{\tau_i\})-\bar v_i|  \leq \lambda$  with probability at least $19/20$ 
\end{lemma}

The following Lemma provides a framework that will be used in Theorems~\ref{thm:dNeighbourhoodFreeness},\ref{thm:1NeighbourhoodFreeness} and \ref{thm:neighbourhoodRegularity}.
\begin{lemma}\label{lemma:testerFramework}
	Let $F$ be a finite set of $r$-types of bounded maximum degree $d$ and let
	$P\subseteq C_{d}$ be the set of all graphs being $\F$-free. Let
	$M\subseteq N$  be a decidable set such that
	$G=(V,E)\in P$ implies that $|V|\notin M$. Let $f_M:N\rightarrow N$
	be a function such that $M$ can be decided in time $f_M$. Assume for every
	$\epsilon\in (0,1]$ there exist $\lambda:=\lambda(\epsilon) \in (0,1]$ and
	$n_0:=n_0(r,\epsilon)\in \mathbb{N}$  such that every graph $G\in C_{d}$ on $n\geq
	n_0$, $n\notin M$ vertices, which is $\epsilon$-far from $P$, contains
	more than $\lambda n$ elements $v$ with
	$\mathcal{N}_r^G(v)\in \tau\in F$. Then $P$ is uniformly
	testable on $C_{d}$ in time $\mathcal{O}(f_M)$.  
\end{lemma}
\begin{proof}
	Consider the following probabilistic algorithm  $T$, which is given direct access to a graph $G\in C_{d}$ and gets the number of vertices $n$ as input.  Let $s=({t^2}/{\lambda^2})\ln(t+40)$.
	\medskip
	\begin{enumerate}
		\item Reject if $n\in M$.
		\item If $n < n_0$: use a precomputed table to decide exactly if $G\in P$.
		\item Otherwise run $\sampler_{r,s}(G)$ to get $\bar v$ satisfying 	$\sum_{i=1}^{t}|\rho_{G,r}(\{\tau_i\})-\bar v_i|  \leq \lambda$ with probability at least ${19}/{20}$.
		\item Reject $G$ if $\sum_{\tau_i\in F}\bar v_{i}>0$. Accept otherwise.
	\end{enumerate}
	
	The query complexity of $T$ is clearly constant, since $s$ is constant and
	the number of vertices in any $r$-neighbourhood is bounded by $d^{r+1}+1$
	for graphs in $C_{d}$. The running time of the first step is $f_M(n)$ and
	for the other steps it is constant.
	
	To prove that $T$ is an $\epsilon$-tester, first assume that $G\in P$. Then $n\notin M$ and $\mathcal{N}_r^G\in \tau \notin F$ for all vertices $v$ . Hence $\sum_{\tau_i\in F}\bar v_{i}=0$ and $T$ accepts $G$.
	Now consider that $G$ is $\epsilon$-far from $P$. If $n\in M$ then
	$G$ is rejected in the first step. Hence let $n\notin M$, and assume
	$\sum_{i=1}^{t}|\rho_{G,r}(\{\tau_i\})-\bar v_i|  \leq \lambda$, which occurs with
	probability at least ${19}/{20}\geq {2}/{3}$. Then
	\begin{align*}
	&\sum_{\tau_i\in F}\bar v_{i}=\sum_{\tau_i\in F}\rho_{G,r}(\{\tau_i\})-\sum_{\tau_i\in F}\big(\rho_{G,r}(\{\tau_i\})-\bar v_{i}\big)>
	\\& \lambda-\Big|\sum_{\tau_i\in F}\big(\rho_{G,r}(\{\tau_i\})-\bar v_{i}\big)\Big|\geq\lambda-\sum_{\tau_i\in F}\big|\rho_{G,r}(\{\tau_i\})-\bar v_{i} \big|\geq 0,
	\end{align*}
	where the first inequality holds by the assumption that in graphs that are $\epsilon$-far from $P$ there are more then $\lambda n$ vertices of type in $F$ made in Lemma \ref{lemma:testerFramework}. Hence $T$ rejects $G$.
\end{proof}
To illustrate the use of the set $M$ in Lemma~\ref{lemma:testerFramework}, let
$P$ be the property of being $K_4$-neighbourhood regular. Let $G_m$ be the
graph consisting of $m$ disjoint copies of $K_4$ and one isolated vertex. First
note that $G_m$ contains $4m+1$ vertices. Being $K_4$-regular implies that
every vertex has degree $3$. But because every graph contains an even number of
vertices of odd degree, $G_m$ cannot be made $K_4$-neighbourhood regular by
edge modifications. Therefore $G_m$ is $\epsilon$-far from $P$. But for
$m\rightarrow \infty$ the probability of sampling the isolated vertex in $G_m$
tends to $0$ meaning that with high probability the tester with $M=\emptyset$
will accept $G_m$. We will show in Theorem~\ref{thm:neighbourhoodRegularity}
that $P$ is testable if we set $M=N \setminus \{4m\mid m\in N\}$.
\begin{lemma}\label{lemma:farImpliesManyCounterexamples-forbidden$r$-ball-oneDegreeMissing}
	For $r\geq 1$ let  $\tau$ be an $r$-type. Let $B$ be an $r$-ball with constant $a$ of type $\tau$. Let $\tilde{d}< d$ and
	assume that $\mathcal{N}_{r-1}^B{a}$ contains a vertex $b$ with
	$\deg_B(b)=\tilde{d}$ and that $\deg_B(v)\not=\tilde{d}+1$ for all
	vertices $v$ in $\mathcal{N}_{r-1}^B{a}$. Let $\epsilon \in (0,1]$ be fixed,
	$n_0={2d^2}/{\epsilon}$ and $\lambda={\epsilon
		d}/(14(1+d^{2r+1}))$. Then every graph $G\in C_{d}$ on
	$n\geq n_0$ vertices which is $\epsilon$-far from being
	$\tau$-neighbourhood free contains more than $\lambda n$ vertices
	of $r$-type $\tau$.
\end{lemma}
\begin{proof}We proceed by contraposition. Assume $G\in
	C_{d}$ is a graph on $n\geq n_0$ vertices containing no more than $\lambda n
	$ vertices $v$ of $r$-type $\tau$. 
	
	\textbf{Case $\tilde{d}=0$, $d=1$.} In this case the tester additionally rejects if $n$ is odd. When $n$ is even we add an edge between any pair of vertices of degree $0$, obtaining a graph $G'$ which is $\lambda n\leq \epsilon dn$ close to $G$. 
	
	\textbf{Case $\tilde{d}=0$, $d>1$.} Then we add one edge to every pair of vertices of degree $0$. If there is only one vertex $v$ of degree $0$ left, we  add an edge from $v$ to any other vertex of degree $<d$. If there is no such vertex then there must be vertex $u$ contained in two edges and we replace one edge $\{u,w\}$ by $\{v,w\}$. We obtain $G'$ which is $2\lambda n\leq \epsilon dn$ close to $G$.
	
	\textbf{Case $\tilde{d}=1$.} We  add edges between pairs of degree $1$ vertices. If there are two left, connected by an edge, we delete that edge. If there is only one vertex $v$ of degree $1$ left, then there is another vertex $u$ of odd degree. By removing an edge $\{u,w\}$ and adding $\{v,w\}$ we get that $\deg_G(v),\deg{w}\geq 1$. We obtain $G'$ which is $2\lambda n\leq \epsilon dn$ close to $G$.
	
	\textbf{Case $\tilde{d}\geq 2$.} Let us pick a set
	$\{v_1,\dots,v_k\}$ of $k\leq \lambda n$ vertices of degree $\tilde{d}$ such
	that for every vertex $v$ of $r$-type $\tau$ there is an
	index $1\leq i\leq k$ with $v_i\in N_{r-1}^G(v)$. We will
	distinguish the following two cases.
	
	First assume that there are less than $\lambda n$ vertices of degree
	$\tilde{d}$, of pairwise distance greater than $2r$ and of distance greater
	than $2r$ from $\{v_1,\dots,v_k\}$. Then there are less than $2\lambda n
	(1+d^{2r+1})$ vertices of degree $\tilde{d}$ in total. Let $G'$ be a
	graph obtained from $G$ by the following modifications. For every vertex
	$v$ of degree $\tilde{d}$ we pick edges
	$\{v,v_1\},\{v,v_2\},\{w,w'\},\{u,u'\}$ such that $v,w,u$ have pairwise
	distance at least $3$. We delete the  edges $\{v,v_1\},\{v,v_2\}$, $\{w,w'\},\{u,u'\}$ and insert the edges
	$\{v_1,w\},\{v_2,u\},\{w',u'\}$, reducing the degree of $v$ while
	maintaining the degrees of all other vertices. The resulting graph has no
	vertex of degree $\tilde{d}$. Note that if such edges do not exist at any
	point during the iteration the graph contains no more than $2d^3\leq
	\epsilon dn$ edges, and we delete them all resulting in a graph with no
	vertex of degree $\tilde{d}$.
	In total we did no more than $7\cdot
	2\lambda n (1+d^{2r+1})\leq \epsilon dn$ edge modifications which implies that
	$G'$ is $\epsilon$-close to $G$. In addition, $G'$ is
	$\tau$-neighbourhood free, because a neighbourhood
	of type $\tau$ would imply having a vertex of degree
	$\tilde{d}$.
	
	Now assume that there are at least $\lambda n$ vertices of degree
	$\tilde{d}$, of pairwise distance greater than $2r$ and of distance greater
	than $2r$ from $\{v_1,\dots,v_k\}$. Let $\{v_1',\dots,v_k'\}$ be a set of
	vertices of degree $\tilde{d}$ such that $\dist_G(v_i,v_j')> 2r$ for all
	$1\leq i,j\leq k$ and $\dist(v_i',v_j')> 2r$ for all $1\leq i<j\leq k$. Let
	$G'$ be the graph obtained from $G$ by inserting the edges $\{v_i,v_i'\}$.
	First note that this takes no more than $\lambda n \leq \epsilon d n$ edge
	modifications which implies that $G$ is $\epsilon$-close to $G'$.  Further
	assume that $v'$ is a vertex in $G'$ of $r$-type $\tau$.
	By choice of the set $\{v_1,\dots,v_k\}$  we altered the isomorphism type of each vertex of type $\tau$ in $G$. Therefore
	$\mathcal{N}_r^{G'}(v')\not=\mathcal{N}_{r}^G(v')$.
	Therefore $\mathcal{N}_r^{G'}(v')$  contains an
	inserted edge $(v_i,v_i')$. We will first prove that either
	$\dist_{G'}(v',v_i)<r$ or $\dist_{G'}(v',v_i')<r$. Assume that this is not the case. Then there is a path 
	$P=$\linebreak $(v_i=w_{-r},w_{-r+1},\dots,w_{-1},w_0=v',w_1,\dots,w_{r-1},w_r=v_i')$
	such that $w_j\not=v_i$ and $w_j\not=v_i'$ for all $-r<j<r$. Let $-r\leq
	j<r$ be the largest index such that $w_j\in
	\{v_1,\dots,v_k,v_1',\dots,v_k'\}$. Then the path $(w_j,\dots,w_r=v_i')$ is
	a path in $G$ of length $\leq 2r$, which contradicts the choice of  $v_1,\dots,v_k,v_1',\dots,v_k'$. Since
	$\deg_{G'}(v_i)=\deg_{G'}(v_i')=\tilde{d}+1$, this implies that
	$\mathcal{N}_{r-1}^G(v')$ contains a vertex of degree $\tilde{d}+1$,
	which contradicts that $v'$ has $r$-type
	$\tau$. Hence $G'$ is
	$\tau$-neighbourhood free.
\end{proof}
\begin{lemma}\label{lemma:farImpliesManyCounterexamples-forbidden$r$-ball-OnlyDegree$d$}
	For $r\geq 1$ let $\tau$ be an $r$-type. Let $B$ be an $r$-ball with constant $a$ of type $\tau$. Assume $\deg_B(v)=d$ for all vertices $v\in N^{B}_{r-1}(a)$. Let $\epsilon\in (0,1]$ be fixed and let $\lambda=\epsilon$. Then every graph $G\in C_{d}$ on $n\geq 1$ vertices which is $\epsilon$-far from being $\tau$-neighbourhood free contains more than $\lambda n$ vertices of $r$-type $\tau$.
\end{lemma}
\begin{proof}
	If $d=0$, then the Lemma holds. 
	Hence we can assume that $B$ is not just an isolated vertex.
	We proceed by contraposition. Assume $G=(V,E)\in C_{d}$ is a graph on $n\geq 1$
	vertices containing no more than $\lambda n $ vertices $v$ of $r$-type
	$\tau$. Let $G'$ be the graph obtained from $G$ by
	isolating all vertices $v$ of $r$-type $\tau$. First note
	that $G'$ is $\epsilon$-close to $G$ since we did no more than $d\lambda n
	\leq \epsilon dn$ edge modifications. Now assume that $v'$ is a vertex of
	$r$-type $\tau$. Since we isolated all vertices having
	$r$-type $\tau$ we know that
	$\mathcal{N}_{r}^{G'}(v')\not=
	\mathcal{N}_{r}^G(v')$.  Therefore there
	must be a vertex $v$ in $N_{r}^G(v')$ such that
	$v$ has type $\tau$, because otherwise the
	$r$-ball of $v'$ could not witness any of the edge modifications. This means
	that there is a path $(v'=v_0,v_1,\dots,v_{k-1},v_k=v)$ of length $k\leq r$
	in $G$. Now pick the maximum index $i$ such that
	$\dist_{G'}(v',v_i)<\infty$. First observe that because $v=v_k$ is isolated
	in $G'$ we get that $i<k$ and therefore $\dist_{G'}(v',v_i)<r$.  Since
	$\dist_{G'}(v',v_{i+1})=\infty$ by construction and $\{v_i,v_{i+1}\}\in
	E$, we get
	$\deg_{\mathcal{N}_{r}^{G'}(v')}(v_i)=\deg_{G'}(v_i)<\deg_G(v_i)\leq
	d$. Since $\mathcal{N}_{r}^{G'}(v')\in \tau$ this yields a
	contradiction to our previous assumption that all vertices in
	$\mathcal{N}^{B}_{r-1}{a}$ have degree $d$. Hence the graph $G'$ can not
	contain a vertex $v'$ of $r$-type $\tau$ and is therefore
	$\tau$-neighbourhood free. 
\end{proof}
The next Lemma follows from Lemmas~\ref{lemma:farImpliesManyCounterexamples-forbidden$r$-ball-OnlyDegree$d$} and~\ref{lemma:farImpliesManyCounterexamples-forbidden$r$-ball-oneDegreeMissing}
for $r=1$, since the $0$-ball has one vertex.
\begin{lemma}\label{lemma:farImpliesManyCounterexamples-forbidden1-ball}
	Let  $\tau$ be a $1$-type. Let $\epsilon\in (0,1]$ be fixed, $n_0={2d^2}/{\epsilon}$ and $\lambda={\epsilon d}/(14(1+d^{3}))$. Then every graph $G\in C_{d}$ on $n\geq n_0$ vertices which is $\epsilon$-far from being $\tau$-neighbourhood free contains more than $\lambda n$ vertices of $1$-type $\tau$.\qed
\end{lemma}
Lemma~\ref{lemma:testerFramework} with $F=\{\tau\}$   and $M=\emptyset$ combined with Lemma~\ref{lemma:farImpliesManyCounterexamples-forbidden$r$-ball-oneDegreeMissing} prove Theorem~\ref{thm:dNeighbourhoodFreeness}. Theorem~\ref{thm:1NeighbourhoodFreeness} follows from Lemma~\ref{lemma:testerFramework} with $\F=\{\tau\}$, or $\F=\emptyset$ in the case that $d$ is not a degree bound for $\tau$, and Lemma~\ref{lemma:farImpliesManyCounterexamples-forbidden1-ball}.

\begin{proof}[Proof of Theorem \ref{thm:neighbourhoodRegularity}] We define $P$ to be the property of being $\tau$-neighbourhood regular and let $\setOfMaxCliques$ be the set of maximal cliques in $G=(V,E)$. Let us define the function $\maxcl^G:V\times\mathbb{N}\rightarrow \mathbb{N}$ where $\maxcl^G(v,i):=|\{K\in \setOfMaxCliques\mid |K|=i,v\in K\}|$  is the number of maximal $i$-cliques containing $v$.
	\begin{claim}\label{claim:DefiningMFor1NieghbourhoodRegularity}
		If $G=(V,E)\in P$ then $\maxcl^B(a,i)\cdot n\equiv 0\mod i$. 
	\end{claim}
	\begin{proof}
		First note that $G\in P$ implies that $\mathcal{N}_1^G(v)\in \tau$ for all $v\in V$. Then $\maxcl^B(a,i)=\maxcl^G(v,i)$ for all $v\in V$ and 
		$\maxcl^B(a,i)\cdot n  =\sum_{v\in V}\maxcl^G(v,i) =|\{K\in\setOfMaxCliques\mid |K|=i \}|\cdot i\equiv 0 \mod i$.
	\end{proof}
	Let $M:=\{n\in \mathbb{N}\mid \text{there is }1\leq i \leq d \text{ such that }\maxcl^B(a,i)\cdot n\not\equiv 0 \mod i\}$. 
	Note that deciding whether $n\in M$ 
	only requires standard arithmetic operations.
	\begin{claim}\label{claim:farImpliesManyCounterexamples-negihbourhoodRegularity}
		For $\epsilon\in (0,1]$ let $\lambda={\epsilon}/(20d^6) $ and $n_0=20 d^8$.  Than any graph $G\in C_{d}$ on $n\geq n_0$, $n\notin M$ vertices, which is $\epsilon$-far from $P$, contains more than $\lambda n$ vertices $v$ with $1$-type $\tau$.
	\end{claim}
	\begin{proof}
		We proceed by contraposition. Let $G=(V,E)\in C_{d}$ be a graph on $n\geq n_0$, $n\notin M$ vertices and assume that $G$ contains $\leq \lambda n$ vertices of $1$-type $\tau$. We will now discribe an algorithmic procedure which takes $<\epsilon d n$ edge modifications and transforms $G$ into a graph $G^{(4)}\in P$, which will prove the claim. 
		
		Let $\tilde{E}^{(1)}:=\{e\in E\mid \text{there are distinct } K, K'\in \setOfMaxCliques, |K\cap K'|>1,e\subseteq K\}$.  Let $G^{(1)}$ be the graph $G^{(1)}=(V, E^{(1)})$, where $E^{(1)}=E\setminus\tilde{E}^{(1)}$. First note that $G^{(1)}$ has no distinct maximal cliques $K,K'$ with $|K\cap K'|>1$. Furthermore 
		\begin{align*}
		|\tilde{E}^{(1)}|&\leq {d\choose 2}\cdot |\{K \in \setOfMaxCliques \mid \text{exists } K'\in \setOfMaxCliques , |K\cap K'|>1 \}|\leq \frac{d^3\lambda n}{2},
		\end{align*}
		where the second inequality holds because every $K\in \setOfMaxCliques$ such that there is $K'\in\setOfMaxCliques$ with $|K\cap K'|>1$ and $K\neq K'$ must contain one of the $\lambda n$ vertices $v$ of $1$-type $\tau$ and there are $\leq d\lambda n$ maximal cliques containing such a vertex. In addition, the removal of the edges in $\tilde{E}^{(1)}$ will affect no more than $d^4\lambda n$ vertices because there are no more than $d^3\lambda n$ vertices contained within an edge of $\tilde{E}^{(1)}$, each of their $1$-neighbourhoods contains $d$ vertices and any vertex, whose $1$-neighbourhood is affected, must be of distance $1$ to one of the vertices contained in an edge in $\tilde{E}^{(1)}$. Hence $G^{(1)}$ contains $\leq (d^4+1)\lambda n<2d^4\lambda n$ vertices $v$ of $1$-type $\tau$.

		Note that in $G^{(1)}$ for all vertices $v$ the graph $\mathcal{N}^{G^{(1)}}_{1}(v)\setminus \{v\}$ is a disjoint union of cliques but there might be $K\in \setOfMaxCliques[G^{(1)}]$ such that $\maxcl^B(a,|K|)=0$. We define the edge set  $\tilde{E}^{(2)}:=\{e\in E^{(1)}\mid \text{exists }K\in \setOfMaxCliques[G^{(1)}], e\subseteq K,\maxcl^B(a,|K|)=0\}$ and let $G^{(2)}$ be the graph $G^{(2)}=(V, E^{(2)}))$, where $E^{(2)}=E^{(1)}\setminus\tilde{E}^{(2)}$.
		Furthermore 
		\begin{align*}
		|\tilde{E}^{(2)}|&\leq d\cdot |\{v\mid \text{exists }K\in \setOfMaxCliques[G^{(1)}], v\in K,\maxcl^B(a,|K|)=0\}|\leq d\cdot 2d^4\lambda n,
		\end{align*} 
		where the first inequality holds because every clique in $G^{(1)}$ has size $\leq d$ and the second because $\mathcal{N}^{G^{(1)}}_{1}(v)\notin \tau$ for every $v\in \{v\mid \text{exists }K\in \setOfMaxCliques[G^{(1)}], v\in K,\maxcl^B(a,|K|)=0\}$.
		
		Note that $\maxcl^B(a,|K|)\not= 0$ for all $K\in
		\setOfMaxCliques[G^{(2)}]$, but there might be $v\in V$ and
		$i\leq d$ with $\maxcl^B(a,i)\not=\maxcl^{G^{(2)}}(v,i)$. Moreover,
		note that because $n\geq n_0$ there are at least $2d$ $4$-balls in
		$G^{(2)}$ which are completely disjoint from the $4$-balls of any vertex
		$v$ of $1$-type $\tau$. $G^{(3)}$ will also have this
		property.
		Let $G^{(3)}=(V,E^{(3)})$ be the graph obtained from $G$ by the following operations.
		For every pair $v,v'$ such that there is $i\leq d$ with
		$\maxcl^B(a,i)>\maxcl^{G^{(2)}}(v,i)$ and
		$\maxcl^B(a,i)<\maxcl^{G^{(2)}}(v',i)$, let $w$ be a vertex of type
		$\tau$ which has at least distance $4$ to $v$ and to $v'$.
		Let $ K'=\{v'_1,\dots,v'_{i-1},v'\}\in \setOfMaxCliques[G^{(2)}]$ and
		$K=\{v_1,\dots,v_{i-1},w\}\in \setOfMaxCliques[G^{(2)}]$. Delete the
		edges $\{\{v',v'_j\},\{w,v_j\}\mid  j\in [i-1]\}$ and add the edges
		$\{\{v,v_j\}, \{w,v'_j\}\mid j\in [i-1]\}$. Note that the vertices
		$v_1,\dots,v_{i-1},v'_1,\dots,v'_{i-1},w$ are still contained in the same
		number of cliques as before, while $v$ is contained in one additional
		$i$-clique and $v'$ is contained in one less. 
		
		Note that in $G^{(3)}$  either $\maxcl^B(a,i)\geq
		\maxcl^{G^{(3)}}(v,i)$ for all vertices $v$ or $\maxcl^B(a,i)\leq
		\maxcl^{G^{(3)}}(v,i)$ for all $v$ for every $i\in [d]$.
		Let $G^{(4)}$ be the graph obtained from $G^{(3)}$ by the following
		operations. For every $i$ such that there is a vertex $v$ with
		$\maxcl^B(a,i)<\maxcl^{G^{(3)}}(v,i)$, we pick $i$ not necessarily
		distinct vertices $v_1,\dots,v_i$ with
		$\maxcl^B(a,i)<\maxcl^{G^{(3)}}(v_j,i)$ for $1\leq j\leq i$. Note that
		this is possible because $\sum_{v\in
			V(G^{(3)})}\maxcl^{G^{(3)}}(v,i)\equiv 0 \mod i$ and $\maxcl^B(a,i)\cdot
		n\equiv 0 \mod i$ by assumption $n\notin M$ and hence we have $\sum_{v\in
			E^{(3)}}(\maxcl^{G^{(3)}}(v,i)-\maxcl^B(a,i))\equiv 0 \mod i$. Let
		$K_j\in \setOfMaxCliques[G^{(3)}]$ such that $v_j\in K_j$ for every
		$1\leq j\leq i$. Let $K=\{w_1,\dots,w_i\}\in \setOfMaxCliques[G^{(3)}]$
		such that the distance between any pair $v_j,w_k$ is at least $4$. Remove
		the set of edges $\{\{w_j,w_k\},\{v_j,v\}\mid v\in K_j, j,k\in [i]\}$ and
		add the set of edges $\{\{w_j,v\}\mid v\in K_j,j\in [i]\}$. Note that
		this reduces the number of maximal $i$-cliques $v_1,\dots,v_i$ are in by
		one, while leaving the number of cliques $w_1,\dots,w_i$ are in the
		same.
		Similarly, for  every $i$  such that there is a vertex $v$ with
		$\maxcl^B(a,i)>\maxcl^{G^{(3)}}(v,i)$ we pick $i$ not necessarily
		distinct vertices $v_1,\dots,v_i$ with
		$\maxcl^B(a,i)>\maxcl^{G^{(3)}}(v_j,i)$ for $ j\in [i]$. Let
		$w_1,\dots,w_i$ be vertices with $\maxcl^B(a,i)=\maxcl^{G^{(3)}}(w_j,i)$
		such that $w_1,\dots,w_i$ are of distance at least $4$ from every $v_j$, $1\leq j\leq i$, and $w_1,\dots,w_i$ are pairwise of distance at least
		$4$.
		Let $K_j\in \setOfMaxCliques[G^{(3)}]$ with $w_j\in K_j$ for $ j\in
		[i]$. Remove the set of edges  $\{\{w_j,w\}\mid w\in K_j,1\leq j\leq i\}$
		and add the set of edges $\{\{v_j,w\}\{w_j,w_k\}\mid w\in K_j,j,k\in
		[i]\}$. Note that this adds one to the number of $i$-cliques
		$v_1,\dots,v_i$ are in, while leaving the number of $i$-cliques
		$w_1,\dots,w_i$ are in the same.
		
		By construction $G^{(4)}\in P$. The number of edge modifications in
		total  is $|E^{(1)}|+|E^{(2)}|$ plus the number of modifications it takes
		to transform $G^{(2)}$ into $G^{(4)}$. First note that 
		$\sum_{i=3}^{d}\sum_{v\in
			V}|\maxcl^D(a,i)$ $-\maxcl^{G^{(2)}}(v,i)|\leq 2d \cdot
		2d^4\lambda n$ since every of the $2d^4\lambda n$ vertices $v$ in
		$G^{(2)}$ of $1$-type $\tau$ can contribute at most $2d$ to
		the sum above. Since transforming $G^{(2)}$ into $G^{(4)}$ we proceed
		greedily, meaning we reduce the number $\sum_{i=3}^{d}\sum_{v\in
			V}|\maxcl^D(a,i)-\maxcl^{G^{(2)}}(v,i)|$ by at least one in
		every step, and every such reduction takes a maximum of $4d^2$ edge
		modifications in total we need less than 
		\begin{displaymath}
		|E^{(1)}|+|E^{(2)}|+4d^2\cdot 2d\cdot 2d^4\lambda n\leq 20 d^7\lambda n=\epsilon d n
		\end{displaymath}
		edge modifications. 
	\end{proof}
	Let $F:=\{\tau'\mid \tau \text{ is a }1\text{-type },\tau\not=\tau'\}$. Note that $|F|\leq t<\infty$, where equality occurs when $B\notin C_{d}$. Then Claim~\ref{claim:farImpliesManyCounterexamples-negihbourhoodRegularity} combined with Lemma~\ref{lemma:testerFramework} for $M$ and $F$ defined as above proves the Theorem. 
\end{proof}

\section{Conclusion}\label{sec:conclusion}
We studied testability of properties definable in first-order logic in the bounded-degree model of property testing for graphs and relational structures, where \emph{testability} of a property means if it is testable with constant query complexity. We showed that all properties in $\Sigma_2$ are testable (Theorem~\ref{thm:sigma2}), 
and we complemented this by exhibiting a property in $\Pi_2$ that is not testable (Theorem~\ref{thm:pi2}). 

Similar results were obtained in the dense graph model in~\cite{alon2000efficient}, albeit with very different methods. Indeed, non-testability of first-order logic in the bounded-degree model is somewhat unexpected: Testing algorithms proceed by sampling vertices and then exploring their local neighbourhoods, and
it is well-known that first-order logic can only express `local' properties. On
graphs and structures of bounded degree this is witnessed by Hanf's strong normal form of first-order logic~\cite{Hanf1965}, which is built around the absence and presence of different isomorphism types of local neighbourhoods.
However, our negative result shows that locality of first-order logic is not sufficient for testability. This also answers an open question from~\cite{AdlerH18}. 

We obtained our non-testable properties by 
encoding the zig-zag construction of bounded degree expanders into first-order logic on
relational structures (Theorem~\ref{thm:nonTestabilityForStructures}) and then extending this to undirected graphs (Theorem~\ref{thm:simpleDelta2}). We believe that this will be of independent interest. 

    Finally, we took an approach suggested by Hanf's normal form, and we proved testability of some 
	 first-order properties that speak about isomorphism types of neighbourhoods, including testability 
	 of $1$-neighbourhood-freeness, and  $r$-neighbourhood-freeness under a mild assumption on the 
	 degrees (Theorem~\ref{thm:dNeighbourhoodFreeness}, Theorem~\ref{thm:1NeighbourhoodFreeness}, and Theorem~\ref{thm:neighbourhoodRegularity}). 
In particular, these theorems imply that there are properties defined by formulas in $\Pi_2\setminus \Sigma_2$ which are testable.
Moreover, these properties are neither monotone nor hereditary, so they are an interesting example towards the remote goal of a full characterisation of the testable properties in the bounded-degree model.
\vspace{1em}

\paragraph{Acknowledgements.} 
We thank Sebastian Ordyniak for inspiring discussions.
The second author thanks Micha\l{} Pilipczuk and Pierre Simon for inspiring discussions at the docCourse on Sparsity in Prague 2018.

\bibliographystyle{alpha}
\bibliography{testingFO}

\newpage
\appendix
\begin{center}\huge\bf Appendix \end{center}

\section{Basics of graphs and first-order logic}\label{sec:basic_appendix}
We let $\mathbb{N}$ denote the set of natural numbers including $0$, and $\Npos:=\mathbb N\setminus\{0\}$. For $n\in \mathbb{N}$ we let $[n]:=\{0,1,\dots,n-1\}$ denote the set of the first $n$ natural numbers. For a set $S$ and $k\in \mathbb{N}$ we denote the Cartesian product $S\times\dots \times S$ of $k$ copies of $S$ by $S^k$. We use $[S]^2$ to denote the set of all two-element subsets of $S$, and we denote the disjoint union of sets by $\sqcup$.
\subsection{Undirected graphs}\label{app:undirectedGraphs}
Unless otherwise specified we allow graphs to have self-loops and parallel edges. We represent an undirected graph $G$ as a triple $(V,E,f)$, where $V$ is the set of vertices, $E$ is the set of edges and $f:E\rightarrow V\cup [V]^2$ is the incidence map. An isomorphism from $G_1=(V_1,E_1,f_1)$ to $G_2=(V_2,E_2,f_2)$ is a pair of bijective maps $(h_V,h_E)$, where $h_V:V_1\rightarrow V_2$ and $h_E:E_1\rightarrow E_2$, such that $h_V(f_1(e))=f_2(h_E(e))$ for any $e\in E_1$, where $h_V(X):=\{h_V(x)\mid x\in X\}$ for any set $X\subseteq V_1$. Undirected graphs without self-loops and parallel edges are called \emph{simple}. 
For a simple graph $G$, we also use the notation $G=(V,E)$, where $V$ is the vertex set and
$E\subseteq [V]^2$. 
The \emph{degree} $\deg_G(v)$ of a vertex $v$ in a graph $G$ is the number of edges to which $v$ is incident. In particular, self-loops contribute one to the degree. We will say that a graph $G$ is \emph{$d$-regular} for some $d\in \mathbb{N}$ if every vertex in $G$ has degree $d$. We specify paths in graphs by tuples of vertices. The \emph{length} of a path on $n$ vertices is $n-1$.  We define the distance between two vertices $v$ and $w$ in a graph $G$, denoted $\dist_G(v,w)$, as the length of a shortest path from $v$ to $w$ or $\infty$ if there is no path from $v$ to $w$ in $G$. Any subset $S\subseteq V$ of vertices \emph{induces} a graph  $G[S]:=(S,\{e\in E\mid f(e)\in S\cup [S]^2\},f|_{S})$. A \emph{connected component} of $G$ is a graph induced by a maximal set $S$, such that each pair $v,w\in S$ has finite distance in $G$. A graph is connected if it has only one connected component. We refer the reader to~\cite{diestelBook}
for the basic notions of graph theory.

We also consider rooted undirected trees. By specifying a root we can uniquely direct the edges away from the root. This allows us to use the terminology of \emph{children} and \emph{parents} for undirected rooted trees. We call a  tree \emph{$k$-ary} if every vertex has either none or exactly $k$ children and we call it complete if, for every $i\in \mathbb N$, there are either exactly $k^i$ or no vertices of distance $i$ to the root of the tree.
\subsection{Relational structures and first-order logic}\label{sec:foappendix}
We will briefly introduce structures and first-order logic and point the reader to~\cite{EF95} for a more detailed introduction.
A  \textit{signature} is a finite set $\sigma =\{R_1,\dots,R_\ell,c_1,\dots,c_m\}$ of relation symbols $R_i$ and constant symbols $c_j$. Every relation symbol $R_i$  has an arity  $\ar(R_i)\in \Npos$.
A \textit{$\sigma$-structure} is a tuple $\mathcal{A}=(A,R_1^\mathcal{A},\dots,R_\ell^\mathcal{A},c_1^\mathcal{A},\dots,c_m^\mathcal{A})$, where $A$ is a \emph{finite} set, called the \emph{universe} of $\mathcal A$, $R_i^\mathcal{A}\subseteq A^{\ar(R_i)}$ is an $\ar(R_i)$-ary relation on $A$ and $c_j^\mathcal{A}\in A$. For a signature $\sigma$ and a $\sigma$-structure $\mathcal{A}$ we call $|A|$ the \emph{size} of $\mathcal{A}$. Two $\sigma$-structures $\mathcal{A}$ and $\mathcal{B}$ are \emph{isomorphic} if there is a bijective map from $A$ to $B$ that preserves all relations and constants. If a signature $\sigma$ contains no constant symbols we call $\sigma$ a \emph{relational signature} and $\sigma$-structures \emph{relational structures}. 
Note that if $\sigma=\{E_1,\ldots,E_{\ell}\}$ is a signature where each $E_i$ is a binary relation
symbol, then $\sigma$-structures are directed graphs 
with at most $\ell$ edge-colours.

In the following we let $\sigma$ be a relational signature.
For a $\sigma$-structure $\mathcal A$ and a subset $B\subseteq A$, we let $\mathcal A[B]$
denote the \emph{substructure} of $\mathcal A$ \emph{induced} by $B$, i.\,e.\ $\mathcal A[B]$ has universe $B$ and $R^{\mathcal A[B]}:=R^{\mathcal A}\cap B^{\text{ar}(R)}$ for all $R\in \sigma$.
A structure $\mathcal B$ is a \emph{substructure} of $\mathcal A$ if  $\mathcal B=\mathcal A[B]$
for some $B\subseteq A$.
The \emph{degree} of an element $a\in A$ denoted by $\deg_\mathcal{A}(a)$ is defined to be the number of tuples in $\mathcal{A}$ containing $a$.
We define the \textit{degree} of $\mathcal{A}$ denoted by $\deg(\mathcal{A})$ to be the maximum degree of its elements. For any $d\in \mathbb{N}$ we let $C_d$ be the class of all $\sigma$-structures of bounded degree $d$. This notion of degree follows \cite{AdlerH18}. 
We say that $\mathcal A$ is \emph{$d$-regular}, if $\deg_\mathcal{A}(a)=d$ for every element $a\in A$.

The set $\FOsigma$ of all first-order formulas over $\sigma$ is recursively built from atomic formulas
$'x=y'$ and $'R_i(x_1,\dots,x_{\ar(R)})'$, for $R\in \sigma$ and variables $x,y,x_1,\dots, x_{\ar(R)}$, and is closed under Boolean connectives $\lnot,\land$ and $\lor$ and existential ($\exists$) and universal ($\forall$) quantification over elements of the structure. We let $\FO:=\bigcup_{\sigma\text{ signature}} \FOsigma$.
We use $\exists^{\geq m}x\,\phi$ (and $\exists^{= m}x\,\phi$,, respectively)
as a shortcut for the FO formula expressing that  the number of witnesses $x$ satisfying $\phi$
is at least $m$ (exactly $m$, respectively).
We say that a variable occurs \emph{freely} in an FO formula if at least one of its occurrences is not bound by any quantifier.
We use $\varphi(x_1,\dots,x_k)$ to express that the set of variables which occur  freely in the FO formula $\varphi$ is a subset of $\{x_1,\dots,x_k\}$. For a formula $\varphi(x_1,\dots,x_k)$, a $\sigma$-structure $\mathcal{A}$ and $a_1,\dots,a_k\in A$ we write $\mathcal{A}\models \varphi(a_1,\dots,a_k)$ if $\varphi$ evaluates to true after assigning $a_i$ to $x_i$, for $1\leq i\leq k$. A \emph{sentence} of FO is a formula with no free variables. For an FO sentence $\varphi$ we say that $\mathcal{A}$ is a \emph{model} of $\varphi$ if $\mathcal{A}\models \varphi$.

\subsection{Hanf normal form and $d$-regular structures}\label{sec:hanftheorem}
The \textit{Gaifman graph} of a $\sigma$-structure $\mathcal{A}$ is the simple, undirected graph $G(\mathcal{A})=(A,E)$, where $\{x,y\}\in E$, if there is an $R\in \sigma$ and a tuple $\overline{b}=(b_1,\dots,b_{\ar(R)})\in R^\mathcal{A}$, such that $x=b_j$ and $y=b_k$ for some $1\leq k,j\leq \ar(R)$ with $j\not=k$. We use $G(\mathcal{A})$ to apply graph theoretic notions to relational structures. 
For two elements $a,b\in A$ we define the \emph{distance} between $a$ and $b$ in $\mathcal{A}$ as
$\dist_\mathcal{A}(a,b):=\dist_{G(\mathcal{A})}(a,b)$. 
For $r\in \mathbb{N}$ and   $a\in A$, the \textit{$r$-neighbourhood} of $a$ is the set $N_r^\mathcal{A}(a):=\{b\in A: \dist_\mathcal{A}(a,b)\leq r\}$. 

If $c$ is a constant symbol with $c\notin\sigma$, we let 
$\sigma_c:=\sigma\cup\{c\}$. 
For $r\in \mathbb N$ and $a\in A$,
the \textit{$r$-ball around $a$} is defined as the $\sigma_c$-structure 
$\mathcal{N}_r^\mathcal{A}(a):=(\mathcal A[N_r^\mathcal{A}(a)],a)$, and $a$ is called the \emph{centre}
of $\mathcal{N}_r^\mathcal{A}(a)$.
In general we call a $\sigma_c$-structure $\mathcal{B}$ an \emph{$r$-ball}, 
if $\mathcal B\cong \mathcal{N}_r^\mathcal{A}(c^{\mathcal B})$.
We call the isomorphism classes of $r$-balls \emph{$r$-types}. For an $r$-type $\tau$ and an element $a\in A$ we say that $a$ \emph{has} ($r$-)type $\tau$ if $\mathcal{N}_r^\mathcal{A}(a)\in \tau$. Observe that for fixed $r,d\in \mathbb N$, there are only finitely many $r$-types in structures in $C_d$.
Moreover, given such an $r$-type $\tau$, there is a formula $\phi_{\tau}(x)$ such that 
for every $\sigma$-structure $\mathcal A$ and for every $a\in A$, $\mathcal A\models\phi_{\tau}(a)$ iff
$\mathcal{N}_r^\mathcal{A}(a)\in \tau$.
A \emph{Hanf-sentence} is a sentence of the form $\exists ^{\geq m} x \phi_{\tau}(x)$, for some $m\in\Npos$, where $\tau$ is an 
$r$-type. Here  $r$ is the \emph{locality radius} of the Hanf-sentence. 
An FO sentence is in \emph{Hanf normal form}, if it is a Boolean combination\footnote{By Boolean combination we 
	always mean \emph{finite} Boolean combination.} of Hanf sentences.
Two formulas $\phi(\bar x)$ and $\psi(\bar x)$ of signature $\sigma$ are called
\emph{$d$-equivalent}, if they are equivalent on $C_d$, i.\,e.\ for all $\mathcal A\in C_d$ and
$\bar a \in A^{|\bar a|}$ with $|\bar a|=|\bar x|$ we have 
$\mathcal A\models\phi(\bar a)$ iff $\mathcal A\models\psi(\bar a)$.
Hanf's locality theorem for first-order logic~\cite{Hanf1965} implies the following.

\begin{theorem}[Hanf~\cite{Hanf1965}]
	Let $d\in\mathbb N$. Every sentence of first-order logic is $d$-equivalent to a 
	sentence in Hanf normal form.
\end{theorem}

\section{Deferred Proofs from Section~\ref{sec:zigZagProduct}}
\label{sec:proof_app}
\begin{proof}[Proof of Lemma \ref{lem:nonBipartitenessConnectedness}]
	Let $1=\lambda_1\geq \lambda_2\geq \dots\geq \lambda_N$ be the eigenvalues of $G^2[S]$. Since $G^2[S]$ is connected, Lemma \ref{lem:connectedBipartiteEigenvalues} implies that $\lambda_1>\lambda_2$. Now assume that $-1$ is an eigenvalue of $G^2[S]$ with eigenvector $\overline{v}$. Then the vector $\overline{v}'$ defined by $\overline{v}'_v=\overline{v}_v$ for all $v\in S$ and $\overline{v}'_v=0$ otherwise is the eigenvector for eigenvalue $-1$ of the graph $G^2$. But $G^2$ can not have a negative eigenvalue as every eigenvalue of $G^2$ is a square of a real number. Therefore $\lambda_1\not= \lambda_N$ and $\lambda(G^2[S])<1$ as claimed. 
\end{proof}
%
%




\end{document}